\documentclass[a4]{llncs}
\usepackage[english]{babel}
\usepackage[latin1]{inputenc} 
\usepackage[T1]{fontenc}
\usepackage{epsfig}
\usepackage{graphicx}
\usepackage{amssymb}
\usepackage{amsfonts}

\usepackage{amsmath}
\usepackage{verbatim}
\usepackage{url}
\usepackage[usenames]{color}
\usepackage{gastex}
\usepackage{stmaryrd}
\usepackage{algorithmic}
\usepackage{algorithm}
\usepackage[all]{xypic}
\usepackage{epic,eepic}
\usepackage{tikz}
\usetikzlibrary{automata}
\usetikzlibrary{positioning}
\usetikzlibrary{arrows}
\usepackage{subfigure}
\usepackage{fancybox}
\title{On the Complexity of Verifying Regular Properties on Flat Counter Systems\thanks{Work
    partially supported by the EU Seventh Framework Programme under grant 
   agreement No. PIOF-GA-2011-301166 (DATAVERIF).}}
\author{St\'ephane Demri$^{2,3}$ \and Amit Kumar Dhar$^1$ \and Arnaud Sangnier$^1$}
\institute{$^{1}$ LIAFA, Univ Paris Diderot, Sorbonne Paris Cité, CNRS, France \\  $^2$ New York University, USA,  $^{3}$ LSV, CNRS, France}

\newcommand{\nat}{\mathbb{N}}
\newcommand{\rel}{\mathbb{Z}}
\newcommand{\interval}[2]{[#1,#2]}
\newcommand{\Nat}{\mathbb{N}}
\newcommand{\Zed}{\mathbb{Z}}

\newcommand{\Bool}{\mathbb{B}}
\newcommand{\Boolplus}{\mathbb{B}^+}

\newcommand{\tuple}[1]{\langle #1 \rangle}
\newcommand{\pair}[2]{\langle #1, #2 \rangle}
\newcommand{\triple}[3]{\langle #1, #2,#3 \rangle}
\newcommand{\set}[1]{\{ #1 \}}

\newcommand{\defeq}{\overset{\textsf{def}}{\Leftrightarrow}}

\newcommand{\defstyle}[1]{{\em #1}}

\newcommand{\until}{\mathtt{U}}

\newcommand{\mynext}{\mathtt{X}}

\newcommand{\sometimes}{\mathtt{F}}

\newcommand{\myprevious}{\mathtt{X}^{-1}}
\newcommand{\since}{\mathtt{S}}

\newcommand{\logspace}{\textsc{LogSpace}}
\newcommand{\nlogspace}{\textsc{NLogSpace}}
\newcommand{\ptime}{\textsc{PTime}}

\newcommand{\np}{\textsc{NP}}

\newcommand{\pspace}{\textsc{PSpace}}

\newcommand{\aformula}{\phi} 
\newcommand{\aformulabis}{\psi} 
\newcommand{\aformulater}{\varphi} 

\newcommand{\true}{\mathit{true}}

\newcommand{\Interp}[1]{\llbracket #1 \rrbracket}

\newcommand{\aalphabet}{\mathtt{\Sigma}}

\newcommand{\aletter}{a}
\newcommand{\aletterbis}{b}

\newcommand{\vect}[1]{\mathtt{\textbf{#1}}}

\newcommand{\counters}{{\tt C}}
\newcommand{\varset}{{\tt Z}}
 \newcommand{\acounter}{\avariable}
\newcommand{\afactor}{a}
\newcommand{\aconstant}{b}
\newcommand{\guards}{{\tt G}}
\newcommand{\aguard}{{\tt g}}
\newcommand{\anupdate}{\vect{u}}

\newcommand{\avect}{\vect{v}}

\newcommand{\asys}{S}
\newcommand{\transsysof}[1]{TS(#1)}
\newcommand{\states}{Q}
\newcommand{\astate}{q}

\newcommand{\edges}{\Delta}
\newcommand{\anedge}{\delta}

\newcommand{\source}[1]{\mathit{source}(#1)}
\newcommand{\target}[1]{\mathit{target}(#1)}
\newcommand{\guard}[1]{\mathit{guard}(#1)}
\newcommand{\update}[1]{\mathit{update}(#1)}

\newcommand{\trans}{\rightarrow}
\newcommand{\labtrans}[1]{\xrightarrow{#1}}
\newcommand{\confs}{C}
\newcommand{\aconf}{c}
\newcommand{\aword}{w}

\newcommand{\arun}{\rho}

\newcommand{\wordof}[1]{\mathit{trans}(#1)}
\newcommand{\statesof}[1]{\mathit{st}(#1)}

\newcommand{\mc}[2]{{\rm MC}(#1,#2)}

\newcommand{\logicfrag}{{\rm L}}

\newcommand{\flatcs}{\mathcal{CFS}}

\newcommand{\kps}{\mathcal{KPS}}
\newcommand{\cps}{\mathcal{CPS}}
\newcommand{\pps}{\mathcal{PPS}}

\newcommand{\aseg}{p}
\newcommand{\first}[1]{\mathit{first}(#1)}
\newcommand{\last}[1]{\mathit{last}(#1)}

\newcommand{\aloop}{l}
\newcommand{\aschema}{P}
\newcommand{\apps}{\mathfrak{P}}
\newcommand{\languageof}[1]{\mathcal{L}(#1)}
\newcommand{\loopsof}[2]{iter_{#1}(#2)}

\newcommand{\avariable}{{\sf x}}
\newcommand{\avariablebis}{{\sf y}}
\newcommand{\avariableter}{{\sf z}}

\newcommand{\alang}{{\rm L}}

\newcommand{\card}[1]{{\rm card}(#1)}
\newcommand{\egdef}{\stackrel{\mbox{\begin{tiny}def\end{tiny}}}{=}} 
\newcommand{\equivdef}{\stackrel{\mbox{\begin{tiny}def\end{tiny}}}{\equivaut}} 
\newcommand{\equivaut}{\;\Leftrightarrow\;}

\newcommand{\aset}{X}
\newcommand{\asetbis}{Y}
\newcommand{\asetter}{Z}

\newcommand{\aconstraintsystem}{\aformula(\avariable_1,  \ldots, \avariable_{k-1})}

\newcommand{\equivEF}[1]{\equiv_{#1}^{ }}
\newcommand{\equivrel}[1]{\approx_{#1}^{ }}
\newcommand{\alabelling}{\mathbf{l}}
\newcommand{\blabelling}{f}

\newcommand{\powerset}[1]{2^{#1}}
\newcommand{\atransition}{\anedge}
\newcommand{\step}[1]{\xrightarrow{\!\!#1\!\!}}
\newcommand{\amap}{f}
\newcommand{\length}[1]{{\rm len}(#1)}
\newcommand{\lengthpathschema}[1]{{\rm len}(#1)} 
\newcommand{\nbloops}[1]{{\rm nbloops}(#1)}

\newcommand{\varprop}{{\rm AT}}
\newcommand{\avarprop}{\mathtt{p}}
\newcommand{\avarpropbis}{\mathtt{q}}

\newcommand{\size}[1]{{\rm size}(#1)}

\newcommand{\cut}[1]{}

\newcommand{\fsa}{\mathcal{A}}
\newcommand{\bauto}{\mathcal{B}}
\newcommand{\autoedge}{\delta}
\newcommand{\buchi}{B\"uchi}
\newcommand{\avar}{{\sf z}}

\newcommand{\qheight}[1]{\mathit{qh}(#1)}

\newcommand{\astruct}{\aword}
\newcommand{\bstruct}{\aword'}
\newcommand{\astrategy}{\sigma}
\newcommand{\aninv}{\mathfrak{I}}

\newcommand{\lmc}{\mu {\rm TL}}
\newcommand{\basis}{B}

\newcommand{\abasis}{\vect{b}}
\newcommand{\aperiod}{\vect{P}}

\newcommand{\play}[3]{\Pi_{#1}^{#2}(#3)}
\newcommand{\lef}[1]{\mathit{left}({#1})}
\newcommand{\rig}[1]{\mathit{right}({#1})}

\newcommand{\cpschema}{\aseg_1 \allowbreak(\aloop_1)^*\allowbreak \cdots \allowbreak\aseg_{k-1}\allowbreak (\aloop_{k-1})^*\allowbreak \aseg_k\allowbreak (\aloop_k)^{\omega}}

\newcommand{\aplay}{\Pi}
\newcommand{\acps}{\mathtt{cps}}
\newcommand{\apcps}{\mathtt{pcps}}

\newcommand{\aspecification}{A}

\newcommand{\prob}[2]{
\vspace{-0.8em}
\begin{center}
\framebox{
\begin{minipage}{11.5cm}
\begin{description}
\itemsep 0 cm
\item[Input:] #1
\item[Output:] #2
\end{description}
\end{minipage}
}
\end{center}
}

\newcommand{\aspeclanguage}{\mathcal{L}}
\newcommand{\mcpb}[2]{{\rm MC(}#1,#2{\rm )}}
\newcommand{\FO}{{\rm FO}}
\newcommand{\BA}{{\rm BA}}
\newcommand{\ABA}{{\rm ABA}}
\newcommand{\ETL}{{\rm ETL}}

\newcommand{\aautomaton}{\mathcal{A}}
\newcommand{\transitions}{\delta}
\newcommand{\mytuple}[2]{\langle #1, \ldots,#2 \rangle}
\newcommand{\auto}{Auto}

\newif \ifconference \conferencetrue 
\newif \iftechreport \techreportfalse 

\makeatletter
\let\c@definition\c@theorem
\let\c@lemma\c@theorem
\let\c@corollary\c@theorem
\let\c@proposition\c@theorem
\makeatother
\begin{document}
\maketitle%
%
%
%
%
\begin{abstract}

Among the approximation methods for the verification of counter systems, one of them 
consists in model-checking
their flat unfoldings. 
\iftechreport
That is why, optimal algorithms for model-checking flat counter systems
are being designed since this may represent the core of a verification process.
\fi 
Unfortunately, the complexity characterization of model-checking problems for such operational
models is not always well studied except for 
reachability queries or for Past LTL. In this paper, we characterize the complexity of model-checking problems on flat counter
systems for the specification languages including first-order logic, linear mu-calculus,
infinite \ifconference automata, \else (such as alternating  B\"uchi automata), \fi and related formalisms.  
Our results span different complexity classes (mainly from \ptime \ to \pspace) and they apply
to languages in which arithmetical constraints on counter values are systematically allowed. 
As far as the proof techniques are concerned, we provide a uniform approach
that focuses on the main issues. 
\end{abstract}

\section{Introduction}
\label{section-introduction}
\paragraph{Flat counter systems.}

Counter systems, finite-state automata equip\-ped with program variables (counters) interpreted over non-negative integers, 
are known to be ubiquitous in formal verification. 
Since counter systems can actually simulate Turing machines \cite{minsky-computation-67}, it is  undecidable to check the existence of a run satisfying a given (reachability, temporal, etc.) property. However it is possible to approximate the behavior of counter systems by looking at a subclass of witness runs 
for which an analysis is feasible. A standard method consists
in considering a finite union of path schemas for abstracting the whole bunch of runs, as done in~\iftechreport \cite{Leroux03,Leroux&Sutre05} \else \cite{Leroux&Sutre05}\fi. More precisely, 
given a finite set of transitions $\edges$, a \defstyle{path schema} is an $\omega$-regular expression over $\edges$ of the
form $\alang = \aseg_1 (\aloop_1)^* \cdots \aseg_{k-1} (\aloop_{k-1})^* \aseg_k (\aloop_k)^{\omega}$  where both $\aseg_i$'s and $\aloop_i$'s
are paths in the control graph and moreover, the $\aloop_i$'s are loops. A path schema defines a set of infinite runs that respect
a sequence of transitions that belongs to $\alang$. We write $\mathtt{Runs}(c_0, \alang)$ to denote such a set of runs 
starting at the initial configuration $c_0$ whereas
$\mathtt{Reach}(c_0, \alang)$ denotes the set of configurations occurring in the runs of $\mathtt{Runs}(c_0, \alang)$.
A counter system is \defstyle{flattable} whenever the set of configurations reachable from $c_0$ 
is equal to $\mathtt{Reach}(c_0, \alang)$ for some finite union of path schemas $\alang$\iftechreport , i.e.
there exists a flat unfolding of the counter system with identical reachability sets\fi. 
Similarly, a \defstyle{flat counter
system}, a system in which each control state belongs to at most one
simple loop,  verifies that the set of runs from $c_0$ is equal to $\mathtt{Runs}(c_0, \alang)$ for some 
finite union of path schemas $\alang$. 
Obviously, flat counter systems are flattable. Moreover,  reachability sets of flattable counter systems are known 
to be Presburger-definable, see e.g.~\cite{Boigelot98,Comon&Jurski98,Finkel&Leroux02b}. 
That is why, verification of flat counter systems belongs to the core of methods for model-checking arbitrary counter systems
and it is desirable to  characterize the computational  complexity of model checking problems on this kind of
systems (see e.g. results about loops in~\cite{Bozga&Iosif&Konecny10}). 
Decidability results for verifying safety and reachability properties on flat counter systems have been obtained
in~\cite{Comon&Jurski98,Finkel&Leroux02b,Bozga&Iosif&Konecny10}.
For the verification of temporal properties, it is much more difficult to get sharp complexity characterization. For instance,
it is known that verifying flat counter systems with CTL$^{\star}$ enriched with arithmetical constraints is decidable~\cite{demri-model-10}
whereas it is only \np-complete with Past LTL~\cite{DemriDharSangnier12} (\np-completeness already holds with flat Kripke structures~\cite{Kuhtz&Finkbeiner11}). 

\paragraph{Our motivations.}
%
Our objectives  are to provide a thorough  classification of
 model-checking problems on flat counter systems
when linear-time properties are considered. So far complexity is known with 
Past LTL~\cite{DemriDharSangnier12} but even the decidability status
with linear $\mu$-calculus is unknown. 
Herein, we wish to consider several formalisms specifying linear-time properties (FO, linear $\mu$-calculus, infinite automata)
and to determine the complexity of model-checking problems on flat counter systems. 
Note that FO is as expressive as Past LTL but much more concise whereas linear $\mu$-calculus is strictly 
more expressive than Past LTL, which motivates
the choice for these formalisms dealing with linear properties.
\iftechreport
Needless to say, the formalisms should admit arithmetical
constraints about counter values in the specifications and we aim at proposing a general method that could be re-used with other formalisms
in order to guarantee  robustness of our solution. 
\fi

\paragraph{Our contributions.}
We characterize the computational complexity of model-checking problems on flat counter systems 
for several prominent linear-time specification languages whose alphabets are related to atomic propositions but also to linear constraints on counter values. We obtain the following results: 
\begin{itemize}
\itemsep 0 cm 
\item The problem of {\bf model-checking first-order formulae on flat counter systems is \pspace-complete} 
(Theorem~\ref{theorem-fo}).
          Note that model-checking  classical first-order formulae over arbitrary Kripke structures is already known to be non-elementary. 
%
%
However the flatness assumption allows to drop the complexity to \pspace \ even though linear
constraints on counter values are used in the specification language.
\item {\bf Model-checking  linear $\mu$-calculus  formulae on flat counter systems is \pspace-complete} (Theorem \ref{theorem-pspace-all}).
   Not only linear $\mu$-calculus is known to be more expressive than first-order logic (or than Past LTL) but also
   the decidability status of the problem on flat counter systems was open \cite{demri-model-10}. So, we establish decidability and
   we provide a complexity characterization. 
\item {\bf Model-checking  B\"uchi automata over flat counter systems is \np-complete} 
(Theorem~\ref{theorem-ba-np-complete}). 
\ifconference
\item {\bf Global model-checking is 
possible} for all the above mentioned formalisms (Corollary~\ref{corollary-global-model-checking}).
\fi
\end{itemize}
\ifconference
\noindent The omitted proofs can be found in the Appendix.
\else
For all the above-mentioned formalisms, we show that global model-checking is 
possible (Corollary~\ref{corollary-global-model-checking}).  
\fi 
\iftechreport
In order to characterize runs of flat counter systems, we introduce the notion of \defstyle{constrained path schema} 
and related decision problems
whose decision procedures are invoked in some essential way to solve the model-checking problems. A constrained path schema is essentially 
a pair of the form 
$$\pair{\aseg_1 (\aloop_1)^* \cdots \aseg_{k-1} (\aloop_{k-1})^* \aseg_k (\aloop_k)^{\omega}}{\aformula(\avariable_1, 
\ldots, \avariable_{k-1})}$$  
where the first component is an $\omega$-regular expression over a finite alphabet 
$\aalphabet$ with $\aseg_i$'s in $\aalphabet^*$ and $\aloop_i$'s in $\aalphabet^{+}$ whereas $\aformula(\avariable_1, \ldots, \avariable_{k-1})$ is a 
linear constraint. 
Note the similarities with the $\omega$-regular expressions involved to define flattable counter systems and observe that 
the first component of constrained path schemas does not deal with counters at all. 
The intended interpretation of such a structure is a set of $\omega$-words in $\aalphabet^{\omega}$ 
such that the number of times the loops $\aloop_i$'s are visited is constrained by the condition expressed by the second component. 
Note also that constrained path schema can be viewed as a subclass of 
CQDD (\defstyle{constrained queue-content decision diagrams})~\cite{Bouajjani&Habermehl99} except that we are also dealing with infinite sequences. 
CQDDs have been introduced for symbolically representing infinite sets of
configurations in FIFO automata -- our use will be different. Parikh automata introduced in~\cite{Klaedtke&Rueb03}
also specify arithmetical constraints about numbers of letter occurrences without assuming flatness. 
The typical problem we consider on constrained path schemas consists in checking whether the intersection between
a language defined by a constrained path schema and a language defined by another formalism (first-order logic, linear $\mu$-calculus, etc.) 
is non-empty.

Let us explain  how the algorithms for solving model-checking problems on flat counter systems are structured and how they use
constrained path schemas. For any specification language, the algorithm is divided as follows (inputs are an initial configuration,
a flat counter system $\asys$ and a specification $A$):
\begin{enumerate}
\itemsep 0 cm
\item Guess a constrained path schema 
$\pair{\aseg_1 (\aloop_1)^* \cdots \aseg_{k-1} (\aloop_{k-1})^* \aseg_k (\aloop_k)^{\omega}}{\aformula}$ 
in polynomial time. Its alphabet is made of atomic propositions and guards from $\asys$ and $\aformulabis$. 
\item Guess $n_1, \ldots, n_{k-1}$ such that 
$\aformula(n_1,  \ldots, n_{k-1})$ holds true in polynomial time.
Bounds on the $n_i$'s depend on the size of  $\aformulabis$ and $\aformula$. 
\item Check whether $\aseg_1 (\aloop_1)^{n_1} \cdots \aseg_{k-1} (\aloop_{k-1})^{n_k} \aseg_k (\aloop_k)^{\omega}$ 
satisfies the input specification $A$ when arithmetical constraints are read symbolically only. 
\end{enumerate}
Step 1. can be done in non-deterministic polynomial time 
and it amounts to reduce model-checking on flat counter systems to problems
on constrained path schemas. It uses both the fact that runs of flat counter systems can be characterized 
by Boolean combinations of linear constraints  and the fact that visits of simple loops in the runs 
can be decomposed in such a way that satisfaction of guards verifies some regularity that simplifies the analysis. 
Step 2. can be done too in  non-deterministic polynomial time by using the fact that linear constraints
have small solutions of polynomial size~\cite{Pottier91} and by establishing a stuttering property for the logical language. 
Step 3. will be done in  non-deterministic polynomial time or in  non-deterministic polynomial
space depending on the specification language.
For instance, for first-order logic or for linear $\mu$-calculus, Step 3. is in non-deterministic polynomial space.
The main difficulty in obtaining optimal bounds for Step 3. rests on the fact that
the $n_i$'s have exponential values: in the best case, this step can be performed in polynomial time as for Past LTL 
(see e.g.~\cite{DemriDharSangnier12})  thanks to a strong stuttering theorem.
\fi

\section{Preliminaries}
\label{section-definitions}
\iftechreport
We write $\nat$ [resp. $\rel$] to denote the set of natural numbers [resp. integers] and $\interval{i}{j}$ to denote the set $\set{k \in \Zed: i \leq
  k \ {\rm and} \ k \leq j}$. In the sequel, integers are encoded with a binary representation. For $\vec{v}\in\Zed^n$, $\vec{v}[i]$ denotes the
$i^{th}$ element of $\vec{v}$ for every $i\in \interval{1}{n}$. 
\iftechreport
For a finite alphabet $\aalphabet$, $\aalphabet^*$ represents the set of finite words over $\aalphabet$, $\aalphabet^+$ the set
of finite non-empty words over $\aalphabet$ and $\aalphabet^\omega$ the set of infinite words or $\omega$-words over $\aalphabet$. For a finite word
$\aword=\aletter_1\ldots \aletter_k$ over $\aalphabet$, we write $\length{\aword}$ to denote its \defstyle{length} $k$.  For $0 \leq i <
\length{\aword}$, $\aword(i)$ represents the $(i+1)$-th letter of the word, here $\aletter_{i+1}$.
\else
For a finite alphabet $\aalphabet$, we denote the set of finite words, finite non-empty words and infinite words (or $\omega$-words) over $\aalphabet$ by $\aalphabet^*$, $\aalphabet^+$ and $\aalphabet^\omega$ respectively. For a finite word
$\aword=\aletter_1\ldots \aletter_k$ over $\aalphabet$, we write $\length{\aword}$ to denote its \defstyle{length} $k$.  For $0 \leq i <
\length{\aword}$, $\aword(i)$ represents the $(i+1)$-th letter of the word.
\fi
\fi 
\subsection{Counter Systems}
Counter constraints are defined below as a subclass of Presburger formulae whose free variables are understood as counters.  Such constraints are used
to define guards in counter systems but also to define arithmetical constraints in temporal formulae.  Let $\counters = \set{\acounter_1, \acounter_2,
  \ldots}$ be a countably infinite set of \defstyle{counters} (variables interpreted over non-negative integers) and $\varprop = \set{\avarprop_1,
  \avarprop_2, \ldots}$ be a countable infinite set of propositional variables (abstract properties about program points).  We write $\counters_n$ to
denote the restriction of $\counters$ to $\set{\acounter_1, \acounter_2, \ldots, \acounter_n}$.
The set of  \defstyle{guards} $\aguard$ using the counters from $\counters_n$, written $\guards(\counters_n)$,
is made of Boolean combinations of \defstyle{atomic guards} of the form 
$\sum_{i=0}^n \ \afactor_i \cdot \acounter_i \sim \aconstant$ 
where  the $\afactor_i$'s are in $\rel$, $\aconstant \in \nat$ and $\sim \in \set{=,\leq,\geq,<,>}$.
%
For $\aguard \in \guards(\counters_n)$ and a vector $\vect{v} \in \nat^n$, we say that $\avect$ satisfies $\aguard$, written $\vect{v}
\models \aguard$, if the formula obtained by replacing each $\acounter_i$ by $\vec{v}[i]$ holds.
\iftechreport
\begin{definition}[Counter system] For a natural 
number $n \geq 1$, a \defstyle{counter system} of dimension $n$ (shortly a counter system) $\asys$ is a tuple
 $\tuple{\states,\counters_n,\edges,\alabelling}$ where:\\
\iftechreport
\begin{itemize}
\itemsep 0 cm 
\item $\states$ is a finite set of \defstyle{control states}.
\item $\alabelling: \states \rightarrow \powerset{\varprop}$ is a \defstyle{labeling function}.
\item $\edges \subseteq \states \times \guards(\counters_n) \times \Zed^n \times \states$ is a finite set of edges labeled by guards and updates of
  the counter values (\defstyle{transitions}).
\end{itemize}
\else
$\states$ is a finite set of \defstyle{control states}, $\alabelling: \states \rightarrow \powerset{\varprop}$ is a \defstyle{labeling function}, $\edges \subseteq \states \times \guards(\counters_n) \times \Zed^n \times \states$ is a finite set of transitions labeled by guards and updates.
\fi
\end{definition}
For  $\anedge=\tuple{\astate,\aguard,\anupdate,\astate'}$ in $\edges$, we  use the following notations:
$\source{\anedge}=\astate$, $\target{\anedge}=\astate'$, $\guard{\anedge}=\aguard$ and $\update{\anedge}=\anupdate$.
\else
 For  $n \geq 1$, a \defstyle{counter system} of dimension $n$ (shortly a counter system) $\asys$ is a tuple
 $\tuple{\states,\counters_n,\edges,\alabelling}$ where:
$\states$ is a finite set of \defstyle{control states}, 
$\alabelling: \states \rightarrow \powerset{\varprop}$ is a \defstyle{labeling function}, 
$\edges \subseteq \states \times \guards(\counters_n) \times \Zed^n \times \states$ is a finite set of 
transitions labeled by guards and updates.
\fi 
%
%
As usual, to a counter system $\asys=\tuple{\states,\counters_n,\edges,\alabelling}$, we associate a labeled transition system
$\transsysof{\asys}=\tuple{\confs,\trans}$ where $\confs=\states \times \nat^n$ is the set of \defstyle{configurations} and $\trans \subseteq \confs
\times \Delta \times \confs$ is the \defstyle{transition relation} defined by: $\triple{\pair{\astate}{\avect}}{\delta}{\pair{\astate'}{\avect'}} \in
\trans$ (also written $\tuple{\astate,\avect} \labtrans{\delta} \tuple{\astate',\avect'}$) iff
\iftechreport
the conditions below are satisfied:
\begin{itemize}
\itemsep 0 cm 
\item $\astate=\source{\anedge}$ and $\astate'=\target{\anedge}$,
\item $\avect \models \guard{\anedge}$ and $\avect'=\avect + \update{\anedge}$.
\end{itemize}
\else
$\anedge=\tuple{\astate,\aguard,\anupdate,\astate'} \in\edges$, $\avect \models \aguard$ and $\avect'=\avect + \anupdate$.
\fi
 Note that in such a transition system, the counter values are non-negative since $\confs=\states \times \nat^n$.
%
%
%
%

Given an initial configuration $\aconf_0 \in \states \times \nat^n$, a \defstyle{run} $\arun$ starting from 
$\aconf_0$ in $\asys$ is an infinite path
in the associated transition system $\transsysof{\asys}$ denoted as:
\ifconference
$
\arun := \aconf_0 \labtrans{\anedge_0} \cdots \labtrans{\anedge_{m-1}} 
\aconf_{m} \labtrans{\anedge_{m}} \cdots 
$
\else
$$
\arun := \aconf_0 \labtrans{\anedge_0} \cdots \labtrans{\anedge_{m-1}} 
\aconf_{m} \labtrans{\anedge_{m}} \cdots 
$$
\fi where $\aconf_i \in \states \times \nat^n$ and $\anedge_i \in \edges$ for all $i \in \nat$.  
%
We say that a counter system is \defstyle{flat} if every node in the underlying graph belongs to at most one simple cycle (a cycle being simple if no
edge is repeated twice in it)~\cite{Comon&Jurski98,Leroux&Sutre05,DemriDharSangnier12}.
\iftechreport 
In a flat counter system, simple cycles can be organized as a DAG where two simple cycles are in
the relation whenever there is path between a node of the first cycle and a node of the second cycle
%
%
\fi
We denote by $\flatcs$ the class of flat
counter systems.
A \defstyle{Kripke structure} $\asys$ can be seen as a counter system without 
counter and is denoted by
$\tuple{\states,\edges,\alabelling}$ where $\edges \subseteq \states \times \states$ and $\alabelling : \states \rightarrow \powerset{\varprop}$. 
Standard notions on counter systems, as configuration, run or flatness, naturally apply
to Kripke structures.

\subsection{Model-Checking Problem}
\label{section-model-checking-problem}

We define now our main model-checking problem on flat counter
systems parameterized by a specification language 
$\aspeclanguage$. 
First, we need to introduce the notion of constrained
alphabet whose letters should be understood as Boolean combinations of
atomic formulae (details follow). 
A \defstyle{constrained alphabet} is a triple of the form $\triple{at}{ag_n}{\aalphabet}$ where
\ifconference
$at$ is a finite subset of $\varprop$,
$ag_n$ is a finite subset of atomic guards from $\guards(\counters_n)$ and 
$\aalphabet$ is a subset of $\powerset{at \cup ag_n}$.
\else
\begin{itemize}
\itemsep 0 cm 
\item $at$ is a finite subset of $\varprop$,
\item $ag_n$ is a finite subset of atomic guards from $\guards(\counters_n)$,
\item $\aalphabet$ is a subset of $\powerset{at \cup ag_n}$.
\end{itemize}
\fi 
The size of a constrained alphabet is given by $\size{\triple{at}{ag_n}{\aalphabet}}=\card{at}+\card{ag_n}+\card{\aalphabet}$ where $\card{\aset}$ denotes the cardinality of the set $\aset$. Of course, any standard alphabet (finite set of letters) can be easily viewed as a constrained alphabet
(by ignoring the structure of letters). 
Given an infinite run $\arun := \tuple{\astate_0,\avect_0} \trans  \tuple{\astate_1,\avect_1} \cdots$
from a  counter system with $n$ counters and 
an $\omega$-word over a constrained alphabet $\aword = \aletter_0, \aletter_1, \ldots \in \aalphabet^{\omega}$, 
we say that
$\arun$ \defstyle{satisfies} $\aword$, written $\arun \models \aword$, whenever for $i \geq 0$, we have
\iftechreport
\begin{enumerate}
\itemsep 0 cm
\item for every $\avarprop \in (\aletter_i \cap at)$ [resp. $\avarprop \in (at \setminus \aletter_i)$], 
      $\avarprop \in \alabelling(\astate_i)$ [resp. $\avarprop \not \in \alabelling(\astate_i)$], 
\item for every $\aguard \in (\aletter_i \cap ag_n)$ [resp. $\aguard \in (ag_n \setminus \aletter_i)$], 
      $\avect_i \models \aguard$ [resp. $\avect_i \not \models \aguard$].
\end{enumerate}
\else
$\avarprop \in \alabelling(\astate_i)$ [resp. $\avarprop \not \in \alabelling(\astate_i)$] for every $\avarprop \in (\aletter_i \cap at)$ [resp. $\avarprop \in (at \setminus \aletter_i)$] and $\avect_i \models \aguard$ [resp. $\avect_i \not \models \aguard$] for every $\aguard \in (\aletter_i \cap ag_n)$ [resp. $\aguard \in (ag_n \setminus \aletter_i)$].
\fi

A \defstyle{specification language} $\aspeclanguage$ over a constrained alphabet $\triple{at}{ag_n}{\aalphabet}$  is a set of \defstyle{specifications} 
$\aspecification$, each of it defining a set $\alang(\aspecification)$ of $\omega$-words over $\aalphabet$. 
\iftechreport
In this paper, our working specification languages are related to first-order logic (over $\omega$-sequences),
nondeterministic B\"uchi automata,  alternating B\"uchi automata, linear $\mu$-calculus and
ETL. 
Note that we explicit introduce the notions of specification languages and specifications
for the sake of clarity in forthcoming definitions, but they cover what is commonly accepted. 
\fi 
We will also sometimes consider specification languages over (unconstrained) 
standard finite alphabets (as usually defined). 
We now define the \defstyle{model-checking problem} over flat counter systems with 
specification language $\aspeclanguage$  (written $\mcpb{\aspeclanguage}{\flatcs}$):
\iftechreport

\prob{A flat counter system $\asys$, 
               a configuration $\aconf$ and a specification $\aspecification$ from $\mathcal{L}$
               }{Is there a run $\arun$ starting at $\aconf$ and 
                 $\aword \in \aalphabet^{\omega}$ in $\alang(\aspecification)$ such that
                 $\arun \models \aword$?}
\else
it takes as input a flat counter system $\asys$, 
 a configuration $\aconf$ and a specification $\aspecification$ from $\mathcal{L}$
and asks whether there is 
 a run $\arun$ starting at $\aconf$ and 
                 $\aword \in \aalphabet^{\omega}$ in $\alang(\aspecification)$ such that
                 $\arun \models \aword$.
\cut{
{\bf Input:}A flat counter system $\asys$, 
               a configuration $\aconf$ and a specification $\aspecification$ from $\mathcal{L}$\\
{\bf Output:}Is there a run $\arun$ starting at $\aconf$ and 
                 $\aword \in \aalphabet^{\omega}$ in $\alang(\aspecification)$ such that
                 $\arun \models \aword$?\\
}
\fi      
We write $\arun \models \aspecification$ whenever there is  $\aword \in \alang(\aspecification)$ such that
                 $\arun \models \aword$.
%
%

\subsection{A Bunch of Specification Languages}
\label{spec-def} 
\paragraph{Infinite Automata.} 
\iftechreport
We define now the specification languages BA and ABA using
non deterministic B\"uchi automata and with alternating B\"uchi automata respectively.
These two formalisms are known to have the same expressive power but not the same conciseness 
and therefore it makes sense to distinguish them.
An obvious way to define a specification in BA over the constrained alphabet $\triple{at}{ag_n}{\aalphabet}$
would be to consider a B\"uchi automaton over the alphabet $\aalphabet$.
In order to be a bit 
concise, we allow that transitions are labelled by Boolean combinations of atoms from
$at \cup ag_n$.
\else
Now let us define the specification languages BA and ABA, respectively
with nondeterministic B\"uchi automata and with alternating B\"uchi automata.
We consider here transitions labeled by Boolean combinations of atoms from
$at \cup ag_n$.
\fi
\iftechreport
 In that way, a specification $\aspecification$ in BA is of
the form $\triple{Q}{E}{q_0,F}$ where $Q$ is a finite set of states,
$q_0 \in Q$ is the initial state, $F \subseteq Q$ is the set of accepting states
and $E$ is a finite subset of $Q \times \Bool(at \cup ag_n) \times Q$ where
$\Bool(at \cup ag_n)$ denotes the set of Boolean combinations built over 
$at \cup ag_n$. Specification $\aspecification$ is just a concise representation for
the B\"uchi automaton $\bauto_{\aspecification} = \triple{Q}{\transitions}{q_0,F}$
where $\transitions$ is a subset of $Q \times \powerset{at \cup ag_n} \times Q$ and
$q \step{\aletter} q' \in \transitions$ $\equivdef$
there is $q \step{\aformulabis} q' \in E$ such that $\aletter \models \aformulabis$ in the
propositional sense (an atom in $(at \cup ag_n) \setminus \aletter$ is interpreted by false).
We say that $\aspecification$ is over the constrained alphabet $\triple{at}{ag_n}{\aalphabet}$,
whenever, for all edges $q \step{\aformulabis} q' \in E$, $\aformulabis$ holds at most
for valuations/letters from $\aalphabet$. 
Hence, the language $\alang(\aspecification) \subseteq \aalphabet^{\omega}$ 
is defined as the language $\alang(\bauto_{\aspecification})$ with the standard notion for B\"uchi
automata
 (accepting runs start at $q_0$ and visit infinitely some accepting state in $F$). 
So, strictly speaking, specifications and B\"uchi automata are not identical objects but
specifications in BA are concise representations of B\"uchi automata
(an edge in $\aspecification$
can lead to an exponential number of transitions in $\bauto_{\aspecification}$).
Note also that any B\"uchi automaton built over the alphabet $\aalphabet$  
can be transformed in polynomial time into a specification over $\aalphabet$.
For each letter $\aletter$, transitions with $\aletter$ in the specification are labelled by
the Boolean formula $\aformulater_{\aletter}$  defined as a conjunction made of positive
literals from $\aletter$ and negative literals from $(at \cup ag_n) \setminus \aletter$.  
By way of example, below, we present a specification in BA for the property discussed earlier about the
flat counter system presented in Figure~\ref{figure-example}.
\begin{center}
\scalebox{0.75}{
\begin{tikzpicture}[shorten >=1pt,node distance=3cm,auto,>=stealth']

  \node[state,initial] (q_1)                  {};
  \node[state]         (q_2)  [right=of q_1]  {};
  \node[state,accepting]         (q_3)  [right=of q_2]  {};
   \path[->] (q_1) edge [out=125,in=55,loop]  node {$\top$} (q_1); 
   \path[->] (q_1) edge   node {$\avarprop \wedge (\acounter_1 - \acounter_2 = 0)$} (q_2);
   \path[->] (q_2) edge  [swap,bend right]   node {$\top$} (q_3);
   \path[->] (q_3) edge  [swap,bend right]   node {$\avarprop \wedge (\acounter_1 - \acounter_2 = 0)$} (q_2);
\end{tikzpicture}
}
\end{center}
Similarly,  a 
\else
A
\fi
specification $\aspecification$ in ABA is a structure of
the form $\triple{Q}{E}{q_0,F}$ where  
$E$ is a finite subset of $Q\times\Bool(at \cup ag_n) \times \Boolplus(Q) $ and
$\Boolplus(Q)$ denotes the set of positive Boolean combinations built over 
$Q$. 
\iftechreport Again, specification 
\else Specification
\fi
$\aspecification$ is a concise representation for
the alternating B\"uchi automaton $\bauto_{\aspecification} = \triple{Q}{\transitions}{q_0,F}$
where $\transitions: Q \times \powerset{at \cup ag_n} \rightarrow \Boolplus(Q)$ and
$\transitions(q,\aletter) \egdef \bigvee_{\tuple{q,\aformulabis,\aformulabis'} \in E, \ \aletter \models \aformulabis} 
\ \aformulabis'$.  We say that 
$\aspecification$ is over the constrained alphabet $\triple{at}{ag_n}{\aalphabet}$,
whenever, for all edges $\tuple{q,\aformulabis,\aformulabis'} \in E$, 
$\aformulabis$ holds at most
for letters from $\aalphabet$ (i.e. the transition relation of $\bauto_{\aspecification}$ belongs to $Q \times \aalphabet \rightarrow \Boolplus(Q)$
).
We have then $\alang(\aspecification) = \alang(\bauto_{\aspecification})$
with the usual acceptance criterion for alternating B\"uchi automata.  \ifconference The specification language BA is defined in a similar way using 
\buchi~automata. Hence the transition relation $E$ of $\aspecification=\triple{Q}{E}{q_0,F}$ in BA is included in $Q \times \Bool(at \cup ag_n) \times Q$ and the transition relation of the B\"uchi automaton $\bauto_{\aspecification}$ is then included in $Q\times\powerset{at \cup ag_n} \times Q$.  
\fi
\paragraph{Linear-time Temporal Logics.} Below, we present briefly three logical languages that are tailored
to specify runs of counter systems, namely ETL (see e.g.\cite{Wolper83,Piterman00}),
Past LTL (see e.g.~\cite{Sistla&Clarke85}) and linear $\mu$-calculus (or $\lmc$), see e.g.~\cite{Vardi88}. 
A specification in one of these logical specification languages is just a  formula. 
The differences with their standard versions in which models are $\omega$-sequences of propositional valuations
are listed below:
\iftechreport
\begin{enumerate}
\itemsep 0 cm
\item Models are infinite runs of counters systems.
\item Atomic formulae are either propositional variables in $\varprop$
  or atomic guards.
\item Given an infinite run 
      $\arun = \tuple{\astate_0,\avect_0} \tuple{\astate_1,\avect_1} \ldots$,
      $\arun, i \models \avarprop$ $\equivdef$ $\avarprop \in \alabelling(\astate_i)$
      and $\arun, i \models \aguard$ $\equivdef$ $\avect_i \models \aguard$. The temporal
operators, fixed point operators and automata-based operators are interpreted then as usual. 
\end{enumerate}
\else
models are infinite runs of counters systems; 
atomic formulae are either propositional variables in $\varprop$
  or atomic guards; 
given an infinite run 
      $\arun := \tuple{\astate_0,\avect_0} \trans  \tuple{\astate_1,\avect_1} \cdots$, we will have $\arun, i \models \avarprop$ $\equivdef$ $\avarprop \in \alabelling(\astate_i)$
      and $\arun, i \models \aguard$ $\equivdef$ $\avect_i \models \aguard$. The temporal
operators, fixed point operators and automata-based operators are interpreted then as usual. 
\fi 
A formula $\aformula$ built over the propositional variables in $at$ and
the atomic guards in $ag_n$ defines a language $\alang(\aformula)$ over  
$\triple{at}{ag_n}{\aalphabet}$ with $\aalphabet = \powerset{at \cup ag_n}$.
\iftechreport
So, $\arun \models \aformula$ 
in the sense of constrained alphabet 
(i.e. there is $\aword \in \alang(\aformula) \subseteq \aalphabet^{\omega}$
such that $\arun \models \aword$) iff \cut{$\arun \models \aformula$ where }$\models$ is the (temporal) satisfaction relation.
\fi
There is no need to recall here the syntax and semantics of ETL, Past LTL and linear $\mu$-calculus
since with their standard definitions and with the  above-mentioned differences, their variants
for counter systems are defined unambiguously (see a lengthy presentation of Past 
LTL for counter systems in~\cite{DemriDharSangnier12}). However, we may recall a few definitions on-the-fly 
if needed.
Herein the size of formulae is understood as the number of subformulae. 
\iftechreport
We just recall the syntax of the formulae\iftechreport for fixing the notations\else , for semantics we use the usual definitions adapted to the differences\fi:
\begin{itemize}
\item ETL:
$
\aformula::= \avarprop \ \mid  \ \aguard \ \mid \
\neg \aformula \ \mid \ \aformula \wedge \aformula \ \mid \
\aautomaton(\aformula, \ldots, \aformula)
$
where $\aautomaton$ is a finite-state automaton,
$\avarprop \in \varprop$ and $\aguard \in \guards(\counters_n)$ for some $n$.
\item $\lmc$:
$
\aformula::= \avarprop \ \mid  \ \aguard \ \mid \
\neg \aformula \ \mid \  \aformula \wedge \aformula \ \mid \
\mynext \aformula \ \mid \ \myprevious \aformula \ \mid \
\mu \avar \cdot \aformula 
$
for some variable $\avar$ that occurs only positively in $\mu \avar \cdot \aformula$ that uses
a least fixed point operator, see e.g.~\cite{Vardi88}. 
\item Past LTL: 
$
\aformula::= \avarprop \ \mid  \ \aguard \ \mid \
\neg \aformula \ \mid \  \aformula \wedge \aformula \ \mid \
\mynext \aformula \ \mid \ \myprevious \aformula \ \mid \
\aformula \until \aformula \ \mid \ \aformula \since \aformula$. 
\end{itemize}
\fi 
\paragraph{Example.}
In \ifconference adjoining figure\else Figure~\ref{figure-example}\fi, we present a \iftechreport simple \fi  
flat counter system with two counters
and with labeling function $\alabelling$ such that $\alabelling(\astate_3)=\set{\avarprop,\avarpropbis}$ and
$\alabelling(\astate_5)=\set{\avarprop}$. We would like  to characterize the set of configurations $\aconf$ with control state $\astate_1$ such that
there is some infinite run from $\aconf$ for which after some position $i$, all future even positions $j$ (i.e. $i \equiv_2 j$)
satisfy that $\avarprop$ holds and the first counter is equal to the second counter.  

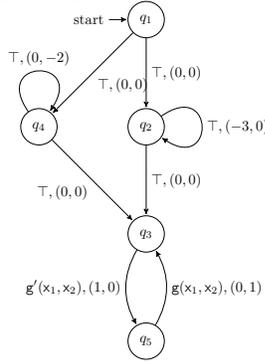
\begin{minipage}{3.8cm}

\scalebox{0.55}{
\begin{tikzpicture}[shorten >=1pt,node distance=1.7cm,auto,>=stealth']

  \node[state,initial] (q_1)                  {$\astate_1$};
  \node[state]         (q_2)  [below=of q_1]  {$\astate_2$};
  \node[state]         (q_3)  [below=of q_2]  {$\astate_3$};
  \node[state]         (q_4)  [left=of q_2]  {$\astate_{4}$};
  \node[state]         (q_5)  [below=of q_3]  {$\astate_5$};
  \path[->] (q_1)   edge              node  {$\top,(0,0)$} (q_2)
                    edge        node  {$\top,(0,0)$} (q_4)
            (q_2)   edge              node  {$\top,(0,0)$} (q_3)
	    	    edge [out=35,in=325,loop]       node  {$\top,(-3,0)$} ()
            (q_3)   edge [swap,bend right] node  {$\aguard'(\acounter_1,\acounter_2),(1,0)$} (q_5)
            (q_5)   edge [swap,bend right] node  {$\aguard(\acounter_1,\acounter_2),(0,1)$} (q_3)
            (q_4)   edge [swap]       node  {$\top,(0,0)$} (q_3)
	    	    edge [out=125,in=55,loop]       node  {$\top,(0,-2)$} ();
\end{tikzpicture}
\label{figure-example}
}
\end{minipage}
\begin{minipage}{7.8cm}
This can be specified in linear $\mu$-calculus using as atomic formulae either propositional
variables or atomic guards. The corresponding formula in linear $\mu$-calculus is: 
$
\mu\avar_1.(\mynext(\nu\avar_2.(\avarprop \wedge (\acounter_1 - \acounter_2 = 0) \wedge \mynext \mynext \avar_2) \vee \mynext \avar_1)
$.
Clearly, such a position $i$ occurs in any run after reaching the control state $\astate_3$ with the same value for both counters. 
Hence, the configurations $\pair{\astate_1}{\avect}$ satisfying these properties have counter values $\avect \in \Nat^2$ verifying the Presburger formula
below:
\end{minipage}
$$
\exists \ \avariablebis \
 (
((\acounter_1 = 3 \avariablebis + \acounter_2)  \wedge 
(\forall \ \avariablebis' \ \aguard( \acounter_2 + \avariablebis',  \acounter_2  + \avariablebis') 
                            \wedge  
                            \aguard'( \acounter_2 + \avariablebis',  \acounter_2  + \avariablebis' + 1)))
\vee
$$
$$ 
((\acounter_2 = 2 \avariablebis + \acounter_1) \wedge
(\forall \ \avariablebis' \ \aguard(\acounter_1  + \avariablebis',  \acounter_1    + \avariablebis') 
                            \wedge  
                            \aguard'(\acounter_1  + \avariablebis',  \acounter_1   + \avariablebis' + 1)))
)
$$
In the paper, we shall establish how to compute systematically such formulae (even without universal quantifications)
\iftechreport
for linear $\mu$-calculus or related specification languages.
\else
for different specification languages.
\fi
\cut{
Below, we just recall the grammars for formulae and 
pinpoint a few features:
\begin{itemize}
\item ETL formulae are of the form below
$$
\aformula::= \avarprop \ \mid  \ \aguard \ \mid \
\neg \aformula \ \mid \ \aformula \wedge \aformula \ \mid \
\aautomaton(\aformula, \ldots, \aformula)
$$
where $\aautomaton$ is a finite-state automaton. 
Each $\aautomaton$ with $m$ letters defines a temporal operator of arity $m$, see e.g.~\cite{Wolper83}. 
\item $\lmc$ formulae are of the form below:
$$
\aformula::= \avarprop \ \mid  \ \aguard \ \mid \
\neg \aformula \ \mid \  \aformula \wedge \aformula \ \mid \
\mynext \aformula \ \mid \ \myprevious \aformula \ \mid \
\mu \avar \cdot \aformula 
$$
for some variable $\avar$ that occurs only positively in $\mu \avar \cdot \aformula$ that uses
a least fixed point operator, see e.g.~\cite{Vardi88}. 
\item 
Past LTL formula are of the form below:
$$
\aformula::= \avarprop \ \mid  \ \aguard \ \mid \
\neg \aformula \ \mid \  \aformula \wedge \aformula \ \mid \
\mynext \aformula \ \mid \ \myprevious \aformula \ \mid \
\aformula \until \aformula \ \mid \ \aformula \since \aformula \
$$
where $\avarprop \in \varprop$ and $\aguard \in \guards(\counters_n)$ for some $n$.
\end{itemize}
}

\cut{
  \subsection{Parikh Path Schemas}

We define minimal path schemas in the line of the respective definitions in~\cite{DemriDharSangnier12}.  Let
$\asys=\tuple{\states,\edges,\alabelling}$ be a Kripke structure, where $\states$ is the set of nodes, $\edges\subseteq\states\times\states$ is
the transition relation and $\alabelling:\states\rightarrow\powerset{\varprop}$ is the labelling fuction. For a transition
$\anedge=\tuple{\astate_1,\astate_2}$, we denote as $\target{\anedge}=\astate_1$ and $\target{\anedge}=\astate_2$. A \defstyle{path segment} $\aseg$
in $\asys$ is a finite sequence of transitions from $\edges$ such that $\target{\aseg(i)}=\source{\aseg(i+1)}$ for all $0 \leq i <
\length{\aseg}-1$. We write $\first{\aseg}$ [resp.  $\last{\aseg}$] to denote the first [resp. last] control state of a path segment, in other words
$\first{\aseg}=\source{\aseg(0)}$ and $\last{\aseg}=\target{\aseg(\length{\aseg}-1)}$. A path segment $\aseg$ is said to be \defstyle{simple} if
$\length{\aseg} >0$ and for all $0 \leq i,j < \length{\aseg}$, $\aseg(i)=\aseg(j)$ implies $i=j$ (no repetition of transitions).  A \defstyle{loop} is
a simple path segment $\aseg$ such that $\first{\aseg}=\last{\aseg}$.  If a path segment is not a loop it is called a \emph{non-loop segment}.  A
\defstyle{path schema} $\aschema$ is an $\omega$-regular expression built over the alphabet of transitions.
More precisely, a path schema $\aschema$ is of the form $\aseg_1 \aloop_1^+ \aseg_2 \aloop_2^+ \ldots \aseg_k \aloop_k^{\omega}$ verifying the following conditions:
\begin{enumerate}
\itemsep 0 cm
\item $\aloop_1$, \ldots, $\aloop_k$ are loops,
\item $\aseg_1 \aloop_1 \aseg_2 \aloop_2 \ldots \aseg_k \aloop_k$ 
is a path segment.
\end{enumerate}

We write $\lengthpathschema{\aschema}$ to denote $\length{\aseg_1 \aloop_1 \aseg_2 \aloop_2 \ldots \aseg_k \aloop_k}$ and $\nbloops{\aschema}$ its
number $k$ of loops.  Let $\languageof{\aschema}$ denote the set of infinite words in $\edges^\omega$ which belong to the language defined by
$\aschema$. Given $\aword \in \languageof{\aschema}$, we write $\loopsof{\aschema}{\aword}$ to denote the unique tuple in $(\nat \setminus \set{0})^{k-1}$ such that
$\aword=\aseg_1 \aloop_1^{\loopsof{\aschema}{\aword}[1]}\aseg_2 \aloop_2^{\loopsof{\aschema}{\aword}[2]}\ldots \aseg_k\aloop_k^\omega$. So, for every
$i \in \interval{1}{k-1}$, $\loopsof{\aschema}{\aword}[i]$ is the number of times the loop $\aloop_i$ is taken. 
So far, a Kripke structure may have an infinite set of path schemas.  To see this, it is sufficient to unroll loops in path segments.  However, we
can impose minimality conditions on path schemas without sacrificing completeness.  A path schema $\aseg_1 \aloop_1^+ \aseg_2 \aloop_2^+ \ldots
\aseg_k \aloop_k^{\omega}$ is \defstyle{minimal} whenever
\begin{enumerate}
\itemsep 0 cm
\item $\aseg_1 \cdots \aseg_k$ is  either the empty word or a simple non-loop segment,
\item $\aloop_1$, \ldots, $\aloop_k$ are loops with disjoint sets of transitions.
\end{enumerate} 

\begin{lemma}\cite{DemriDharSangnier12}
\label{lemma-schemata-finite}
Given a Kripke structure $\asys=\tuple{\states,\edges,\alabelling}$, the total number of minimal path schemas of $\asys$ is finite and
is smaller than $\card{\edges}^{(2 \times \card{\edges})}$.
\end{lemma}
This lemma is a simple consequence of the fact that in a minimal path schema, each transition occurs at most twice.
In Figure~\ref{figure-example-minimal-path-schemas}, we present a Kripke structure $\asys$ with a unique counter and one of its minimal path
schemas. The minimal path schema shown in Figure~\ref{figure-example-minimal-path-schemas} corresponds to the $\omega$-regular expression
$\atransition_1 (\atransition_2 \atransition_3)^+ \atransition_4 \atransition_5 (\atransition_6 \atransition_5)^{\omega}$.  In order to avoid
confusions between path schemas and flat counter systems that look like path schemas, simple loops in the representation are labelled by
\fbox{$\omega$} or \fbox{$\geq 1$} depending whether the simple loop is the last one or not.  Note that in the representation of path schemas, a state
may occur several times, as it is the case for $\astate_3$ (this cannot occur in the representation of counter systems).
\begin{figure}
\begin{center}
\scalebox{0.8}{
\begin{tikzpicture}[->,>=stealth',shorten >=1pt,
node distance=1cm, thick,auto,bend angle=60]

  \tikzstyle{every state}=[fill=white,draw=black,text=black]
  \node[state] (q0) [below]    {$\astate_0$};
  \node[state] (q1) [right= of q0]    {$\astate_1$};
  \node[state] (q2) [above= of q1]    {$\astate_2$};
  \node[state] (q3) [right=2cm of q1]    {$\astate_3$};
  \node[state] (q4) [above= of q3]    {$\astate_4$};
  \node[state] (q02) [right=2.2cm of q3]    {$\astate_0$};
  \node[state] (q12) [right= of q02]    {$\astate_1$};
  \node[state] (q22) [above= of q12]    {$\astate_2$};
  \node[state] (q32) [right= of q12]    {$\astate_3$};
  \node[state] (q42) [right=1.4cm of q32]    {$\astate_4$};
  \node[state] (q52) [above= of q42]    {$\astate_3$};
  \node (q62) [above=0.3cm of q12] {{\small \fbox{$\mathbf{\geq 1}$}}};
  \node (q72) [above=0.3cm of q42] {{\small \fbox{$\mathbf{\omega}$}}};

  \path[->] (q0) edge node {} (q1);
  \path[->] (q1) edge [bend left] node [right] {} (q2);
  \path[->] (q2) edge [bend left] node  [right] {} (q1);
  \path[->] (q1) edge node  [above] {} (q3);
  \path[->] (q3) edge  [bend left] node  [left] {} (q4);
  \path[->] (q4) edge  [bend left] node  [right] {} (q3);     
  \path[->] (q02) edge node {} (q12);
  \path[->] (q12) edge  [bend left] node {} (q22);
  \path[->] (q22) edge  [bend left] node {} (q12);
  \path[->] (q12) edge  node [below] {} (q32);
  \path[->] (q32) edge node  [above] {} (q42);
  \path[->] (q42) edge  [bend left] node  [left] {} (q52);
  \path[->] (q52) edge  [bend left] node  [right] {} (q42);

\end{tikzpicture}
}
\end{center}
\caption{A Kripke structure and one of its minimal path schemas}
\label{figure-example-minimal-path-schemas}
\end{figure}
Minimal path schemas play a crucial role in the sequel, mainly because of the properties stated below.%
\begin{lemma} \cite{DemriDharSangnier12}
Let $\aschema$ be a path schema. There is a minimal path schema $\aschema'$ such that every run respecting $\aschema$ respects $\aschema'$ too.
\end{lemma}
The proof of the above lemma is by an easy verification. Indeed, whenever a maximal number of copies of a simple loop is identified as a factor of
$\aseg_1 \aloop_1 \cdots \aseg_k \aloop_k$, this factor is replaced by the simple loop unless it is already present in the path schema.

Recall that a run $\arun$ starting from a state $\astate_0$ in $\asys$ is an infinite path
in the Kripke structure denoted as:
$$
\arun := \astate_0 \labtrans{\anedge_0} \cdots \labtrans{\anedge_{m-1}} 
\astate_{m} \labtrans{\anedge_{m}} \cdots 
$$ 
where $\astate_i \in \states$ and $\anedge_i \in \edges$ for all $i \in \nat$.  Let $\wordof{\arun}$ be the $\omega$-word $\anedge_0
\anedge_1 \ldots$ and $\statesof{\arun}=\astate_0\astate_1\ldots$.
Finally, we say that a run $\arun$ \defstyle{respects} a path schema $\aschema$
if $\wordof{\arun} \in \languageof{\aschema}$ and for such a run, we write $\loopsof{\aschema}{\arun}$ to denote
$\loopsof{\aschema}{\wordof{\arun}}$.  Note that by definition, if $\arun$ respects $\aschema$, then each loop $\aloop_i$ is visited at least once,
and the last one infinitely.

A \defstyle{Parikh Path Schema} is defined as the tuple $\apps=\tuple{\aschema,\aconstraintsystem}$, where $\aschema$ is a minimal path schema and
$\aconstraintsystem$ is a quatifier-free Presburger formula over the variables $y_1,y_2,\cdots,y_{k-1}$ where $k$ is the number of loops in
$\aschema$. A run $\arun$ in $\apps$ is an infinite run such that $\arun$ respects $\aschema$ and $\loopsof{\aschema}{\arun}$ is a solution to
$\aconstraintsystem$.

\subsection{Model-checking problem}
In this work, we wish to study model-checking problems for Parikh path schema 
taking into account formal specification which can speak about the
atomic propositions present in the seen state. 
In our framework, a specification will be given either in
form of a logic formula or of an automata. More formally to a specification $A$, we will associate $\languageof{A}$ a set
of finite or infinite words over the alphabet $2^{\varprop}$. 

The verification problem we are interested in is then the model-checking problem for different formal specification over counter systems. We write this problem
$\mc{\logicfrag}{\pps}$, where $\logicfrag$ is a formalism to define finite infinite runs and $\pps$ is the class of Parikh path schema.
$\mc{\logicfrag}{\pps}$ can be defined formally as follows:
\begin{description}
\itemsep 0 cm 
\item[Input:] A Parikh path schema $\apps \in \pps$, a state $\astate_0$ and a specification $A \in \logicfrag$;
\item[Output:] Does there exist a infinite (or finite) run $\arun$ in $\apps$ starting from $\astate_0$ such that
  $\statesof{\arun} \in\languageof{A}$?
\end{description}
If the answer is "yes", we will write $\apps,\astate_0 \models A$. 

It is known from \cite{DemriDharSangnier12}, that the model-checking problem for LTL with past on Parikh path schema is \np-complete.  Although the
decidability and complexity classification for many other logics remained open.

}
%
%
%
\cut{
\subsection{Characterizing Runs by Quantifier-free Presburger Formulae}
We will show and important technique of constructing a quantifier-free Presburger formula which characterizes exactly the runs respecting a given path
schema from a counter system \textit{without disjunction} in the guards. Given a path schema
$\aschema=\aseg_1\aloop_1^+\cdots\aseg_k\aloop_k^{\omega}$ in a counter system $\asys$ without disjunction in the guards, we can construct a
constraint system $\aconstraintsystem$ consisting of conjunction of several quantifier-free Presburger formula with $k-1$ variables. The variables in
$\aconstraintsystem$ corresponds to each of the loops in $\aschema$. Thus, a solution $(m_1,m_2,\ldots,m_{k-1})$ to $\aconstraintsystem$, actually
represents a run $\arun$ such that $\loopsof{\aschema}{\arun}=(m_1,m_2,\ldots,m_{k-1})$. From \cite{DemriDharSangnier12}, we get the following
theorem,
\begin{theorem}\cite{DemriDharSangnier12}
\label{thm:constraint}
Let $\aschema$ be a path schema from a flat counter system $\asys$, then there exists a constraint system $\aconstraintsystem$ satisfying the following properties
\begin{itemize}
  \item $k-1$ variables.
  \item at most $\length{\aschema}\times 2 \times\size{\aschema}^2$ conjuncts
  \item the greatest absolute value of the constants are bounded by $2n\times K(K+2)\times(\length{\aschema}+1)$, where $K$ is the maximum absolute
    value of the constants occurring in $\aschema$.
  \item $(m_1,m_2,\ldots,m_{k-1})$ is a solution iff there exists a run $\arun$ respecting $\aschema$, the guards and the counter updates, such
    that $\loopsof{\aschema}{\arun}=(m_1,m_2,\ldots,m_{k-1})$.
\end{itemize}
\end{theorem}
As described in \cite{DemriDharSangnier12} constructing each of the conjuncts in $\aconstraintsystem$ requires atmost one pass over $\aschema$. Thus, the construction of $\aconstraintsystem$
can be done in polynomial time in $\size{\aschema}$. Also, the conditions in Theorem~\ref{thm:constraint} ensures that the size of
$\aconstraintsystem$ is polynomial in $\size{\aschema}$.

\subsection{On how to deal with arithmetical specifications and disjunctions}
\label{subsec:arithmetic-to-classic}
As said before, we consider in this work formal specifications dealing with both atomic propositions and arithmetical constraints and we will see here that it is possible to consider the model-checking problem without interpreting the cardinal constraints by adding to the label of the states of the considered counter systems  cardinal constraints. For all the classes of specification we will consider in our work, we will have the following property: given a specification $A$ defining a language over $\powerset{\varprop}\times \powerset{\guards(\counters_n)}$, there will be a finite subset of arithmetical constraints $\guards(A)$ such that $\languageof{A}$ will be defined over $\powerset{\varprop}\times \powerset{\guards(A)}$. Following a construction proposed in \cite{DemriDharSangnier12}, this allows us from a path schema with labels in $\powerset{\varprop}$ to build a set of path schemas with labels in $\powerset{\varprop} \times \powerset{\guards(A)}$ and thus to answer the model-checking problem without interpreting the constraints of the formula. Furthermore, this construction also allows to remove the disjunctions present in the guards of the path schema which is crucial to apply the technique introduced in the previous subsection.
More precisely, given a specification $A$ with cardinal constraints and  a path schema $\aschema$ from a flat counter system $\asys=\tuple{\states,\counters_n,\edges,\alabelling}$ where $\alabelling:\states\rightarrow\powerset{AP}$, it is possible to transform $\aschema$ to a set of path
schemas $Y_{\aschema}$ such that every $\aschema'\in Y_{\aschema}$ is a path schema from a flat counter system
$\asys'=\tuple{\states',\counters_n,\edges',\alabelling'}$ where $\alabelling':\states\rightarrow \powerset{\varprop}\times \powerset{\guards(A)}$ and for every
$\anedge'=\tuple{\astate,\aguard,\anupdate,\astate'}$ in $\edges'$, $\aguard$ is a conjunction of arithmetical constraints. In fact, from \cite{DemriDharSangnier12} we have the following result.
\begin{theorem}\cite{DemriDharSangnier12}
\label{thm:disjunction}
Given a specification $A$ with $\languageof{A}$ over $\powerset{\varprop} \times \powerset{\guards}$, a path schema $\aschema$ in a flat counter system $\asys$ and a configuration $\tuple{\astate_0,\avect_0}$, one can construct in polynomial time (in the sizes of $\aschema$  and $\guards(A)$) a set $Y_{\aschema}$ such that:
\begin{itemize}
\item The control states present in the path schemas of  $Y_{\aschema}$ have their associated label in $\powerset{\varprop} \times \powerset{\guards(A)}$.
\item No path schema in $Y_{\aschema}$ contains guards with disjunction in it.
\item For every path schema $\aschema'\in Y_{\aschema}$, its length is polynomial in the length of $\aschema$ and in the size of $\guards(A)$.
\item For every run $\arun$ starting in $\tuple{\astate_0,\avect_0}$ and respecting $\aschema$, there is a run $\arun'$ respecting some $\aschema'\in Y_{\aschema}$ and starting with values $\avect_0$ such that
  $\statesof{\arun} \in\languageof{\Interp{A}}$ iff $\blabelling(\statesof{\arun'}) \in \languageof{A}$.
\item For every run $\arun'$ respecting some $\aschema' \in Y_{\aschema}$ starting with values $\avect_0$, there is a run $\arun$ starting in $\tuple{\astate_0,\avect_0}$ such that
  $\statesof{\arun} \in\languageof{\Interp{A}}$ iff $\blabelling(\statesof{\arun'}) \in \languageof{A}$.
\end{itemize}
\end{theorem}

The consequence of this theorem is that we can consider the set of atomic propositions to be $\varprop \cup \guards(A)$ for a given specification $A$ and the model-checking can then be thought of as without counter constraints and  be performed as classical model-checking with specification over a finite alphabet.  Thus, in this way we can transform an arithmetical
specification over a path schema from a flat counter system to another equivalent problem of checking classical specification (without guards) over
path schema from flat counter system. Finally we point the attention on the fact that Theorem~\ref{thm:disjunction} gives us a way to eliminate disjunction from the guards of path
schema. This is useful for characterizing runs using constraint system as explained before.

}

\section{Constrained Path Schemas}
\iftechreport \subsection{Abstracting Path Schemas} \fi
In~\cite{DemriDharSangnier12} we introduced minimal path schemas for 
flat counter systems. 
\iftechreport
A path schema is an 
$\omega$-regular expression over the alphabet of transitions of the form $\aseg_1 (\aloop_1)^* \cdots \aseg_{k-1} (\aloop_{k-1})^* \aseg_k (\aloop_k)^{\omega}$ 
where each $\aseg_i$ is a finite path and each $\aloop_i$ is a simple cycle of the given flat counter system. A path schema is said to be minimal when no transition of the counter systems is used more than twice in the expression. Minimal path schemas enjoy the following properties: for a given flat counter system there are a finite (exponential) number of minimal paths and the word of transitions of each run of a flat counter system is accepted by a minimal path schema. 
Furthermore, we can associate to each path schema an arithmetical constraint characterizing how many times each loop is taken and thus represent in a 
finite manner all the possible runs of a flat counter system. 
\fi
Now, we introduce \defstyle{constrained path schemas} that are more abstract than
path schemas. 
%
%
%
A \defstyle{constrained path schema} $\acps$ is  a pair 
$\pair{\aseg_1 (\aloop_1)^* \cdots \aseg_{k-1} (\aloop_{k-1})^* \aseg_k (\aloop_k)^{\omega}}{\linebreak[0]\aformula(\avariable_1,\linebreak[0] 
\ldots, \avariable_{k-1})}$
where the first component is an $\omega$-regular expression over a constrained alphabet 
$\tuple{at,ag_n,\aalphabet}$ with $\aseg_i, \aloop_i$'s in $\aalphabet^*$, and $\aformula(\avariable_1, \ldots, \avariable_{k-1}) \linebreak[0]
\in \guards(\counters_{k-1})$.
\iftechreport
Each constrained path schema defines a language 
$\alang(\acps) \subseteq \aalphabet^{\omega}$ that is a subset 
of the language defined by the $\omega$-regular expression by taking into account the constraints on the repetition of the
$l_i$'s. More specifically:
$$
\alang(\acps) \egdef \set{
\aseg_1 (\aloop_1)^{n_1} \cdots \aseg_{k-1} (\aloop_{k-1})^{n_{k-1}} \aseg_k (\aloop_k)^{\omega}: 
\aformula(n_1, \ldots, n_{k-1}) \ {\rm holds \ true}
}
$$ 
\else
Each constrained path schema defines a language 
$\alang(\acps) \subseteq \aalphabet^{\omega}$ given by
$\alang(\acps) \egdef \set{
\aseg_1 (\aloop_1)^{n_1} \cdots \aseg_{k-1} (\aloop_{k-1})^{n_{k-1}} \aseg_k (\aloop_k)^{\omega}: 
\aformula(n_1, \ldots, n_{k-1}) \ {\rm holds \ true}
}
$. 
\fi 
The size of $\acps$, written $\size{\acps}$, is equal to
$2k + \length{\aseg_1 \aloop_1 \cdots \aseg_{k-1} \aloop_{k-1} \aseg_k \aloop_k} +
\size{\aformula(\avariable_1,\ldots, \linebreak[0] \avariable_{k-1})}$. 
\iftechreport
Observe that in general constrained path schemas
are defined under constrained alphabet and all the decision problems we  consider
also assume that the specifications are over constrained alphabet unless stated otherwise.
\else
Observe that in general constrained path schemas
are defined under constrained alphabet and so will the associated specifications unless stated otherwise. 
\fi
%
%

\iftechreport \subsection{Decision Problems} \label{subsection:dec-prob} \fi
Let us consider below the three decision problems on constrained path schemas that are useful in the rest of the paper. 
\iftechreport
\emph{Consistency problem} is defined as follows:\\
\iftechreport
\prob{A constrained path schema $\acps$}{Is $\alang(\acps) \neq \emptyset$?}
Note that consistency amounts to check that the second argument of the constrained path schema is satisfiable. Let us first recall simple consequences
of the classical result~\cite[Corollary 1]{Pottier91} which will be useful to state the complexity of the consistency problem but also that we will
take advantage of later.
\else
{\bf Input:}A constrained path schema $\acps$\\
{\bf Output:}Is $\alang(\acps) \neq \emptyset$?\\
\fi
\else
\defstyle{Consistency problem} checks whether $\alang(\acps)$ is non-empty.
It amounts to verify the satisfiability status of the second component. 
Let us recall the result below. 
\fi 
\begin{theorem}~\cite{Pottier91}
\label{theorem-pottier91}
\iftechreport
There exist polynomials $\mathtt{pol}_1(\cdot)$, $\mathtt{pol}_2(\cdot)$ and $\mathtt{pol}_3(\cdot)$ such that for every guard $\aguard$, say in
$\guards(\counters_n)$, of size $N$, we have
\begin{description}
\itemsep 0 cm 
\item[(I)] there exist $\basis \subseteq  \interval{0}{2^{\mathtt{pol}_1(N)}}^n$ and 
           $\aperiod_1, \ldots, \aperiod_{\alpha} \in \interval{0}{2^{\mathtt{pol}_1(N)}}^n$ with 
            $\alpha \leq 2^{\mathtt{pol}_2(N)}$ such that for every $\vect{y} \in \Nat^n$, $\vect{y} \models \aguard$
           iff there are $\abasis \in \basis$ and $\vect{a} \in \Nat^{\alpha}$ such that
           $\vect{y} = \abasis + \vect{a}[1] \aperiod_1 + \cdots + \vect{a}[\alpha] \aperiod_{\alpha}$;
\item[(II)] if $\aguard$ is satisfiable, then there is  $\vect{y} \in \interval{0}{2^{\mathtt{pol}_3(N)}}^n$ such that
             $\vect{y} \models \aguard$. 
\end{description}
\else
There are polynomials $\mathtt{pol}_1(\cdot)$, $\mathtt{pol}_2(\cdot)$ and $\mathtt{pol}_3(\cdot)$ 
such that for every guard $\aguard$, say in
$\guards(\counters_n)$, of size $N$, we have
(I) there exist $\basis \subseteq  \interval{0}{2^{\mathtt{pol}_1(N)}}^n$ and 
           $\aperiod_1, \ldots, \aperiod_{\alpha} \in \interval{0}{2^{\mathtt{pol}_1(N)}}^n$ with 
            $\alpha \leq 2^{\mathtt{pol}_2(N)}$ such that for every $\vect{y} \in \Nat^n$, $\vect{y} \models \aguard$
           iff there are $\abasis \in \basis$ and $\vect{a} \in \Nat^{\alpha}$ such that
           $\vect{y} = \abasis + \vect{a}[1] \aperiod_1 + \cdots + \vect{a}[\alpha] \aperiod_{\alpha}$;
(II) if $\aguard$ is satisfiable, then there is  $\vect{y} \in \interval{0}{2^{\mathtt{pol}_3(N)}}^n$ s.t.
             $\vect{y} \models \aguard$. 
\fi
\end{theorem}
\iftechreport
Note that (II) is an immediate consequence of (I). 
Now, we can obtain the \np~upper bound for the consistency problem (the lower bound being obtained directly by reducing SAT for instance), 
indeed, a constrained path schema  defines a non-empty language iff its formula can be satisfied by a tuple 
in $\interval{0}{2^{\mathtt{pol}_3(N)}}^{k-1}$ where $N$ is its size. This allows us to state the next lemma.
\begin{lemma} Consistency problem is \np-complete.
\end{lemma}
\else
Consequently, the consistency problem is \np-complete (the hardness being obtained by reducing SAT). 
\fi 
\iftechreport
The \defstyle{intersection non-emptiness problem}, defined below and clearly related to model-checking problem,
is defined as follows:
\prob{A constrained path schema $\acps$ and a 
specification $\aspecification \in \mathcal{L}$}{Is $\alang(\acps) \cap \alang(\aspecification) \neq \emptyset$ ? }
\else
The \defstyle{intersection non-emptiness problem}, clearly related to model-checking problem,
takes as  input a constrained path schema $\acps$ and a specification $\aspecification \in \mathcal{L}$
and asks whether $\alang(\acps) \cap \alang(\aspecification) \neq \emptyset$. 
\fi 
\cut{
The main problem of interest for us since as we shall see it is strongly connected with the model-checking problem is the
 \defstyle{intersection non-emptiness problem} for a specification language $\mathcal{L}$ over a finite alphabet $\Sigma$:\\
\iftechreport
\prob{A constrained path schema $\acps$ and a 
specification $\aspecification \in \mathcal{L}$}{Is $\alang(\acps) \cap \alang(\aspecification) \neq \emptyset$ ? }
\else
{\bf Input:} A constrained path schema $\acps$ and a specification $\aspecification \in \mathcal{L}$\\
{\bf Output:} Is $\alang(\acps) \cap \alang(\aspecification) \neq \emptyset$?
\fi
}
Typically, for several specification languages  $\mathcal{L}$\iftechreport(first-order logic, B\"uchi automata, etc.)\fi, we
establish the existence of a computable map $\amap_{\aspeclanguage}$ (at most exponential) such that 
whenever  $\alang(\acps) \cap \alang(\aspecification) \neq \emptyset$ there is 
$\aseg_1 (\aloop_1)^{n_1} \cdots \aseg_{k-1} (\aloop_{k-1})^{n_{k-1}} \aseg_k (\aloop_k)^{\omega}$ belonging
to the intersection and for which each $n_i$ is bounded by $\amap_{\mathcal{L}}(\aspecification, \acps)$. 
\iftechreport 
A way to solve an instance of the intersection non-emptiness problem is to guess $n_1, \ldots, n_{k-1}$ bounded by 
 $\amap_{\mathcal{L}}(\aspecification, \acps)$, assuming that it is sufficient, and then to test that 
$\aseg_1 (\aloop_1)^{n_1} \cdots \aseg_{k-1} (\aloop_{k-1})^{n_{k-1}} \aseg_k (\aloop_k)^{\omega}$  indeed belongs to 
$\alang(\acps) \cap \alang(\aspecification)$.\fi 
\iftechreport
This motivates the introduction of a last decision problem over constrained path schemas namely \emph{the 
membership problem} for a specification language $\mathcal{L}$  over a finite alphabet $\Sigma$:\\
\iftechreport
\prob{A constrained path schema $\acps$, a specification $\aspecification \in \mathcal{L}$ 
       and $n_1, \ldots, n_{k-1} 
      \in \Nat$}{Is 
      $\aseg_1 (\aloop_1)^{n_1} \cdots \aseg_{k-1} (\aloop_{k-1})^{n_{k-1}} \aseg_k (\aloop_k)^{\omega}
      \in  \alang(\aspecification)$ ? }
\else
{\bf Input:} A constrained path schema $\acps$,  $\aspecification \in \mathcal{L}$ 
       and $n_1,\ldots, n_{k-1} 
      \in \Nat$\\
{\bf Output:} Is 
      $\aseg_1 (\aloop_1)^{n_1} \cdots \aseg_{k-1} (\aloop_{k-1})^{n_{k-1}} \aseg_k (\aloop_k)^{\omega}
      \in  \alang(\aspecification)$?
\fi
\else
This motivates the introduction of the \defstyle{membership problem} for $\aspeclanguage$ 
that takes as input a constrained path schema $\acps$, a specification $\aspecification \in \aspeclanguage$ 
       and $n_1,\ldots, n_{k-1} 
      \in \Nat$ and checks whether 
      $\aseg_1 (\aloop_1)^{n_1} \cdots \linebreak[0] \aseg_{k-1} (\aloop_{k-1})^{n_{k-1}} \aseg_k (\aloop_k)^{\omega}
      \in  \alang(\aspecification)$.
\fi 
Here the $n_i$'s are understood to be encoded in binary and we do not require them to satisfy
the constraint of the path schema\iftechreport (it is just irrelevant for our future developments)\fi. 
\cut{
Complexity of the membership problem with standard specification languages will be established
in the rest of the paper. For example, we know from~\cite{DemriDharSangnier12} that 
the membership problem with Past LTL
can be solved in polynomial time whereas Corollary~\ref{corollary-intersection-fo}
states that the membership problem for first-order logic is \pspace-complete.
By  contrast, the membership problem with B\"uchi automata will be shown in \ptime. 
}

\iftechreport
\subsection{Reduction From Model-Checking to Intersection Non-Emptiness}
\label{section-main-reduction}
\fi 

Since constrained path schemas are abstractions of path schemas used in~\cite{DemriDharSangnier12}, from this work we can show that runs from flat counter systems can be represented
by a finite set of constrained path schemas as stated below.


%
%
%
\begin{theorem}
\label{theorem-reduction} 
Let 
$at$ be a finite set of atomic propositions,
$ag_n$ be a finite set of atomic guards from  $\guards(\counters_n)$,
$\asys$ be a flat counter system whose atomic propositions and atomic guards are from $at \cup ag_n$ and
$\aconf_0 = \pair{\astate_0}{\avect_0}$ be an initial configuration. 
One can construct in exponential time a set $\aset$ of constrained path schemas s.t.:
\iftechreport
\begin{itemize}
\itemsep 0 cm 
\item Each constrained path schema $\acps$ in $\aset$ has an alphabet of the form 
      $\triple{at}{ag_n}{\aalphabet}$ ($\aalphabet$ may vary) and 
     $\acps$ is of polynomial size.
\item Checking  whether a constrained path schema belongs to $\aset$ can be done in polynomial time. 
\item For every run $\arun$ from $\aconf_0$, there is a constrained path schema $\acps$ in $\aset$ and $\aword \in \alang(\acps)$ such that
      $\arun \models \aword$. 
\item  For every constrained path schema $\acps$ in $\aset$ and for every $\aword \in \alang(\acps)$,  
       there is a run $\arun$ from $\aconf_0$ such that $\arun \models \aword$. 
\end{itemize}
\else
(I) Each constrained path schema $\acps$ in $\aset$ has an alphabet of the form 
      $\triple{at}{ag_n}{\aalphabet}$ ($\aalphabet$ may vary) and 
     $\acps$ is of polynomial size.
(II) Checking  whether a constrained path schema belongs to $\aset$ can be done in polynomial time. 
(III) For every run $\arun$ from $\aconf_0$, there is a constrained path schema $\acps$ in $\aset$ and $\aword \in \alang(\acps)$ such that
      $\arun \models \aword$. 
(IV) For every constrained path schema $\acps$ in $\aset$ and for every $\aword \in \alang(\acps)$,  
       there is a run $\arun$ from $\aconf_0$ such that $\arun \models \aword$. 
\fi 
\end{theorem}
%
%
\iftechreport
Below, we provide the main steps of the proof, details can be found in~\cite{DemriDharSangnier12}. 

\begin{proof} (sketch)
Let us explain how to build the set $\aset$.
\begin{enumerate}
\itemsep 0 cm
\item Given a flat counter system $\asys$ and a state $\astate$ from $\aconf$, 
      there is at most an exponential number of minimal path schemas
      starting at $\astate$ in the sense of~\cite[Lemma 4]{DemriDharSangnier12}. 
      Let $\asetbis_1$ be this set of minimal path schemas.
\item For each path schema $P$ in $\asetbis_1$, there is a set of path schemas $\asetbis_{P}$  
      such that the path schemas in $\asetbis_{P}$   have no disjunctions in guards and satisfaction
      of guards can be concluded from the states, see~\cite[Theorem 14]{DemriDharSangnier12}. 
      Let $\asetbis_2$ be this set of unfolded path schemas
      and it is of cardinality at most exponential.
\item Following~\cite[Lemma 12]{DemriDharSangnier12}, every path schema from $\asetbis_2$ is equivalent 
      to a constrained path schema. The set $\aset$ is precisely the set of constrained path schemas
      obtained from all the unfolded path schemas from $\asetbis_2$.
\end{enumerate}
Completeness of the set $\aset$ is a consequence of~\cite[Lemma 12]{DemriDharSangnier12} 
and~\cite[Theorem 14(4--6)]{DemriDharSangnier12}. Satisfaction of the size constraints is
a consequence of~\cite[Lemma 12]{DemriDharSangnier12} 
and~\cite[Theorem 14(2--3)]{DemriDharSangnier12}.
\qed
\end{proof}
\fi
%
%
In order to  take advantage of Theorem~\ref{theorem-reduction}
for the verification of flat counter systems, we need to introduce an additional property: 
$\aspeclanguage$  has the 
\defstyle{nice subalphabet property} iff for all 
specifications $\aspecification \in \aspeclanguage$ over $\triple{at}{ag_n}{\aalphabet}$ and
for all constrained alphabets $\triple{at}{ag_n}{\aalphabet'}$,  one can build a specification
$\aspecification'$ over $\triple{at}{ag_n}{\aalphabet'}$ in polynomial time in the sizes of 
$\aspecification$ and
$\triple{at}{ag_n}{\aalphabet'}$ such that 
\iftechreport
$$\alang(\aspecification) \cap (\aalphabet')^{\omega} = \alang(\aspecification')$$ 
Let us justify this definition. In this paper, we will have two ways of building specification languages for flat counter systems from standard 
specification languages (like first-order logic or B\"uchi automata). The first way consists in representing each letter from a constrained alphabet $\triple{at}{ag_n}{\aalphabet}$ explicitly (i.e. letters will belong to $\powerset{at \cup ag_n}$). The second way consists in representing set of letters by a Boolean combination of atomic propositions and atomic guards. In the latter case, it is not always obvious that the specification language will have the nice subalphabet property however we
\else
$\alang(\aspecification) \cap (\aalphabet')^{\omega} = \alang(\aspecification')$.  We
\fi 
need this property to build from $A$ and a constraint path schema over $\triple{at}{ag_n}{\aalphabet'}$, the specification $A'$. This property will also be used to transform a specification over $\triple{at}{ag_n}{\aalphabet}$ into a specification over the finite alphabet $\aalphabet'$.

\iftechreport
\begin{lemma} \label{lemma-nice-subalphabet-property}
For every specification language $\aspeclanguage$ among BA, ABA, $\lmc$, ETL or Past LTL, 
the language $\aspeclanguage$ has the nice subalphabet property.
\end{lemma}
\fi
\ifconference
\begin{lemma} \label{lemma-nice-subalphabet-property}
BA, ABA, $\lmc$, ETL, Past LTL have the nice subalphabet property.
\end{lemma}
\fi
\iftechreport
\begin{proof} 
Let $\aspecification = \triple{Q}{E}{q_0,F}$ be a specification in BA. 
over the alphabet $\triple{at}{ag_n}{\aalphabet}$ and
$\aalphabet' \subseteq \aalphabet$. The specification $\aspecification' =  \triple{Q}{E'}{q_0,F}$
such that $\alang(\aspecification') = \alang(\aspecification) \cap (\aalphabet')^{\omega}$
is defined as follows: for every $q \step{\aformulabis} q' \in E$, we include in
$E'$ the edge 
$q \step{(\bigvee_{\aletter \in \aalphabet'} \ \aformulabis_{\aletter}) \wedge \aformulabis} q'$
where  $\aformulater_{\aletter}$ is defined as a conjunction made of positive
literals from $\aletter$ and negative literals from $(at \cup ag_n) \setminus \aletter$. 
A similar transformation can be performed with specifications in ABA.

Let $\aformula$ be a formula for $\aspeclanguage$ among linear $\mu$-calculus, ETL or Past LTL
built over atomic formulae in $at \cup ag_n$ and  $\triple{at}{ag_n}{\aalphabet'}$ be a constrained alphabet.
The formulae $\aformula'$ such that $\alang(\aformula') = \alang(\aformula) \cap (\aalphabet')^{\omega}$
is obtained from $\aformula$ by replacing every atomic formula $\aformulabis$ by
$\bigvee_{\set{\aletter \in \aalphabet' \mid \aformulabis \in \aletter}} \aformulater_{\aletter}$.
\qed
\end{proof}
\fi
\iftechreport
\begin{algorithm}
{\footnotesize
\caption{Solving   $\mcpb{\aspeclanguage}{\flatcs}$}
\label{algorithm:abstact}
\begin{algorithmic}[1]
\REQUIRE A flat counter system $\asys$, a configuration $\aconf_0$ and a specification $A$ from $\mathcal{L}$
\STATE Guess $\acps$ in $\aset$ with an alphabet of the form 
      $\triple{at}{ag_n}{\aalphabet'}$ [see Theorem \ref{theorem-reduction}]
\STATE Build $\aspecification'$ such that $\alang(\aspecification) \cap (\aalphabet')^{\omega} = \alang(\aspecification')$
\STATE Return $\alang(\acps)\cap\alang(\aspecification') \neq \emptyset$
\end{algorithmic}
}
\end{algorithm}
Abstract algorithm~\ref{algorithm:abstact} 
is the one used to solve $\mcpb{\aspeclanguage}{\flatcs}$ when $\aspeclanguage$ has the nice
subalphabet property. As one can notice it take fully advantage of Theorem \ref{theorem-reduction} in the first guess of a constrained path schema and of the nice subalphabet property of the specification language to build the specification $A'$. 
\else
The abstract Algorithm 1 which performs the following steps (1) to (3) takes as input  $\asys$, 
a configuration $\aconf_0$ and $\aspecification \in \aspeclanguage$ and solves  $\mcpb{\aspeclanguage}{\flatcs}$:
(1) Guess $\acps$ over  $\triple{at}{ag_n}{\aalphabet'}$ in $\aset$;
(2) Build $\aspecification'$ such that $\alang(\aspecification) \cap (\aalphabet')^{\omega} = \alang(\aspecification')$;
(3) Return $\alang(\acps)\cap\alang(\aspecification') \neq \emptyset$.
\fi 
\ifconference 
\else
We now show that this Algorithm is correct.
\begin{lemma}
Algorithm \ref{algorithm:abstact} on inputs $\asys$, $\aconf_0$ and $A$ has an accepting run iff there exists a run $\arun$ of $\asys$ starting at $\aconf_0$ such that $\arun \models A$.
\end{lemma}
\begin{proof}
First assume there exists a run $\arun$ of $\asys$ starting at $\aconf_0$ such that $\arun \models A$. By Theorem \ref{theorem-reduction},  there is a constrained path schema $\acps$ with an alphabet of the form 
      $\triple{at}{ag_n}{\aalphabet'}$ in $\aset$ and $\aword \in \alang(\acps)$ such that $\arun \models \aword$. Consequently we deduce that  $\aword \in \alang(A)$ and that $\alang(\acps) \cap \alang(A) \neq \emptyset$. Since $\alang(\acps) \subseteq (\aalphabet')^{\omega}$ and since $\alang(\aspecification) \cap (\aalphabet')^{\omega} = \alang(\aspecification')$, we deduce that $\alang(\acps)\cap\alang(\aspecification') \neq \emptyset$. Hence the Algorithm has an accepting run. 

Now if the Algorithm~\ref{algorithm:abstact} 
has an accepting run, we deduce that there exists a constrained path schema $\acps$ with an alphabet of the form $\triple{at}{ag_n}{\aalphabet'}$ in $\aset$ such that there exists a word $\aword$ in $\alang(\acps) \cap \alang(A')$. Using the nice subalphabet property we deduce that $\aword \in \alang(A)$ and by the last point of Theorem \ref{theorem-reduction}, we know that there exists a run $\arun$ from $\asys$ starting at $\aconf_0$ such that $\arun \models \aword$. This allows us to conclude that $\arun \models A$. \qed
\end{proof}
\fi
Thanks to Theorem~\ref{theorem-reduction},  
the first guess 
\iftechreport in Algorithm~\ref{algorithm:abstact}
\else 
\fi  can be performed in polynomial time and with the nice subalphabet property,  we can build $\aspecification'$ 
in polynomial time too. 
\iftechreport
This together with the previous lemma allows us to conclude the following lemma which make the link between the model-checking problem and the intersection non-emptiness problem for specification languages having the nice subalphabet property.
\else This allows us to conclude the following lemma which is a consequence of the correctness of the above algorithm (Appendix~\ref{section-proof-correctness-algo1}).
%
\fi 
\begin{lemma}
\label{lemma-complexity-subalphabet}
If $\mathcal{L}$ has the nice subalphabet property and its intersection non-emptiness problem is in \np [resp. \pspace], then
 $\mcpb{\mathcal{L}}{\flatcs}$ is in \np [resp. \pspace]
\end{lemma}
We know that the membership problem for Past LTL 
is in \ptime  \
and the intersection non-emptiness problem is in \np~
(as a consequence of~\cite[Theorem 3]{DemriDharSangnier12}). 
By Lemma~\ref{lemma-complexity-subalphabet}, we are able to conclude the main result
from~\cite{DemriDharSangnier12}: $\mcpb{{\rm Past LTL}}{\flatcs}$ is in \np. 
This is not surprising at all since in this paper we present a general 
method for different specification languages that rests on Theorem~\ref{theorem-reduction} 
(a consequence of technical developments from~\cite{DemriDharSangnier12}). 
\section{Taming First-Order Logic and Flat Counter Systems}

In this section, we consider first-order logic as a specification language. 
\ifconference
By Kamp's  Theorem,
\else Using Kamp's  Theorem~\cite{Kamp68}, 
\fi 
first-order logic 
has the same expressive power as Past LTL and hence model-checking first-order logic over flat counter systems is decidable 
too~\cite{DemriDharSangnier12}. However this does not provide us an optimal upper bound for the
 model-checking problem. In fact, it is known that
the satisfiability problem for first-order logic formulae is 
non-elementary  and consequently the 
translation into Past LTL  leads to a significant blow-up in the size of the formula.

%
%
\subsection{First-Order Logic in a Nutshell}
\label{section-FO-definition} 

For defining first-order logic formulae, we consider a countably infinite set of variables $\varset$ 
and a finite (unconstrained) alphabet $\aalphabet$. 
The syntax of \defstyle{first-order logic} over
atomic propositions FO$_{\aalphabet}$ is then given by the following grammar:
\iftechreport
$$
\begin{array}{lcl}
\aformula & ::= &  
\aletter(\avar)~\mid~S(\avar,\avar')~\mid~ \avar < \avar' \mid~ \avar = \avar' \mid~  \neg \aformula~ \mid~ 
\aformula \wedge \aformula'~ \mid~ 
\exists\avar~\aformula(\avar) 
\end{array}
$$
\else
$
\begin{array}{lcl}
\aformula & ::= &  
\aletter(\avar)~\mid~S(\avar,\avar')~\mid~ \avar < \avar' \mid~ \avar = \avar' \mid~  \neg \aformula~ \mid~ 
\aformula \wedge \aformula'~ \mid~ 
\exists\avar~\aformula(\avar) 
\end{array}
$
\fi 
where 
$\aletter \in \aalphabet$ 
and $\avar,\avar' \in \varset$. For a formula $\aformula$, we will denote by $free(\aformula)$ its set of free variables defined as usual.
%
%
A formula with no free variable
is called a \defstyle{sentence}. As usual, we define the \defstyle{quantifier height} $\qheight{\aformula}$ of a formula $\aformula$ as
the maximum nesting depth of the operators $\exists$ in $\aformula$.
Models for FO$_{\aalphabet}$  are  $\omega$-words over the alphabet 
$\aalphabet$ 
and variables are interpreted by positions in the word. A \defstyle{position assignment} is a 
partial function $\amap:\varset\rightarrow\nat$.
Given a model $\aword \in \aalphabet^{\omega}$, 
a FO$_{\aalphabet}$ formula $\aformula$ and a position assignment $\amap$ such that 
$\amap(\avar)\in\nat$ for every variable $\avar\in free(\aformula)$, the
satisfaction relation $\models_{\amap}$ is defined as \ifconference usual. \else follows:
$$
\begin{array}{lcl}
\aword\models_{\amap} \aletter(\avar) & \mbox{iff} & \aletter = \aword(\amap(\avar))\\
\aword\models_{\amap} S(\avar,\avar') & \mbox{iff} & \amap(\avar')= \amap(\avar)+1\\
\aword\models_{\amap} \avar<\avar' & \mbox{iff} & \amap(\avar) < \amap(\avar')\\
\aword\models_{\amap} \avar=\avar'& \mbox{iff} & \amap (\avar) = \amap(\avar') \\
\aword\models_{\amap} \neg \aformula & \mbox{iff} & \aword \nvDash_{\amap} \aformula\\
\aword\models_{\amap} \aformula \wedge \aformula' & \mbox{iff} & \aword\models_{\amap} \aformula\mbox{ and }\aword \models_{\amap} \aformula'\\
\aword\models_{\amap} \exists\avar~\aformula(\avar) & \mbox{iff} &\mbox{ there exists }j \in\nat\mbox{ such that } \aword 
\models_{\amap[\avar\rightarrow j]} \aformula(\avar)\\
\end{array}
$$ 
\fi
Given a  FO$_{\aalphabet}$ sentence $\aformula$, we write $\aword \models \aformula$ when $\aword\models_{\amap} \aformula$ for an arbitrary position 
assignment $\amap$. 
The language  of $\omega$-words $\aword$ over 
$\aalphabet$ associated to a sentence $\aformula$ is then $\languageof{\aformula}= 
\{\aword\in \aalphabet^{\omega}~\mid~\aword\models\aformula\}$. For $n \in \nat$, we  define the equivalence relation 
$\equivrel{n}$ between $\omega$-words over 
$\aalphabet$
as:
 $\astruct\equivrel{n}\bstruct$ when for every sentence
$\aformula$ with $\qheight{\aformula}\leq n$, $\astruct\models\aformula$ iff $\bstruct\models\aformula$.

\paragraph{FO on CS.}
FO formulae interpreted over infinite runs of counter systems are defined as FO formulae over a finite alphabet
except that atomic formulae of the form $\aletter(\avar)$ are replaced by atomic formulae of the form 
$\avarprop(\avar)$  or $\aguard(\avar)$ where $\avarprop$ is an atomic formula or $\aguard$ 
is an atomic guard from $\guards(\counters_n)$. Hence, a formula $\aformula$ built over atomic formulae
from a finite set $at$ of atomic propositions and from a finite set $ag_n$ of atomic guards
from $\guards(\counters_n)$ defines a specification for the constrained
alphabet $\triple{at}{at_n}{\powerset{at \cup ag_n}}$. Note that the alphabet can be of exponential size
in the size of $\aformula$ and $\avarprop(\avar)$ actually corresponds to a disjunction 
$\bigvee_{\avarprop \in \aletter} \aletter(\avar)$.
\begin{lemma}
\label{lemma-fo-subalphabet-property}
FO  has the nice subalphabet property.
\end{lemma}

\iftechreport
\begin{proof}
Consider a FO formula $\aformula$ that defines a specification
over the  constrained alphabet $\triple{at}{ag_n}{\aalphabet}$ with $\aalphabet = \powerset{at \cup ag_n}$.
Consider a subalphabet $\aalphabet'\subseteq\aalphabet$.
Let $\aformula''$ be the formula obtained from $\aformula$
by replacing  every occurrence of  $\avarprop(\avar)$ by 
      $\bigvee_{\set{\aletter \in \aalphabet' \mid \avarprop \in \aletter}} \aletter(\avar)$
and  every occurrence of  $\aguard(\avar)$ is replaced by 
      $\bigvee_{\set{\aletter \in \aalphabet' \mid \aguard \in \aletter }} \aletter(\avar)$.
It is easy to see that, by construction, $\alang(\aformula'')= \alang(\aformula')\cap(\aalphabet)^{\omega}$
\qed
\end{proof}

\fi
We have taken time to properly define first-order logic
for counter systems (whose models are runs of counter systems, see also Section~\ref{section-model-checking-problem})
but below, we will mainly operate with  FO$_{\aalphabet}$ over a standard (unconstrained) alphabet.

\iftechreport  

\subsection{Ehrenfeucht-Fra\"iss\'e Games}
\label{section-definition-efgames}
Ehrenfeucht-Fra\"iss\'e (EF) game is a well known technique to determine whether two structures are equivalent with respect to a set of formulae. We recall here the definition of a EF game adapted to our context. Given $N \in \nat$ and two $\omega$-words $\aword,\aword'$ over 
$\aalphabet$, 
the main idea of the corresponding EF game is that two players, the Spoiler and the Duplicator, plays in a turn based manner. 
The Spoiler begins by choosing a word between $\aword$ and $\aword'$ and a position in this word, then the Duplicator aims at finding a position in the other word which is \emph{similar} and this during $N$ rounds. At the end, the Duplicator wins if the set of chosen positions respects some 
isomorphism. We now move to the formal definition of such a game.

Let $\aword$ and $\aword'$ be two $\omega$-words over 
$\aalphabet$. 
We define a play as a finite sequence of triples 
$(p_1,a_1,b_1)(p_2,a_2,b_2)\cdots(p_i,a_i,b_i)$ in $(\{0,1\}\times\nat^2)^{*}$ where for each triple the first element describes 
which word has been chosen by the Spoiler ($0$ for the word $\aword$), then the second element corresponds to the position chosen in $\aword$ and the third element the position chosen in $\aword'$ by the Spoiler or the Duplicator according to the word chosen by the Spoiler. For instance if $p_1=1$, this means that at the first turn Spoiler has chosen the position $b_1$ in $\aword'$ and Duplicator the position $a_1$ in $\aword$. A play of size $i \in \nat$ is called an 
$i$-round play (a $0$-round play being an empty sequence). A strategy for the Spoiler is a mapping $\astrategy_{S}:
(\{0,1\}\times\nat^2)^{*}\rightarrow\{0,1\}\times\nat$ which takes as input a play and outputs $0$ or $1$ for words $\astruct$ or
$\bstruct$ respectively and a position in the word. Similarly, a strategy for the Duplicator is a  mapping $\astrategy_{D}:
(\{0,1\}\times\nat^2)^{*} \times (\{0,1\}\times\nat)\rightarrow\nat$ with the difference being that Duplicator takes into account the position played by the Spoiler in the current round. For all $i \in \nat$, a strategy $\astrategy_S$ for the Spoiler and a strategy $\astrategy_D$ for the Duplicator, the $i$-round play over $\aword$ and $\aword'$ following $\astrategy_S$ and $\astrategy_D$ is defined inductively as follows: $\play{i}{\astrategy_S,\astrategy_D}{\aword,\aword'}= \play{i-1}{\astrategy_S,\astrategy_D}{\aword,\aword'}(p,a,b)$ where if $p=0$, $(0,a)=\sigma_D(\play{i-1}{\astrategy_S,\astrategy_D}{\aword,\aword'})$ and $b=\sigma_S(\play{i-1}{\astrategy_S,\astrategy_D}{\aword,\aword'},(0,a))$  and if $p=1$, $(1,b)=\sigma_D(\play{i-1}{\astrategy_S,\astrategy_D}{\aword,\aword'})$ and $a=\sigma_S(\play{i-1}{\astrategy_S,\astrategy_D}{\aword,\aword'},(1,b))$.

For $N \in \nat$, a $N$-round play $(p_1,a_1,b_1)(p_2,a_2,b_2)\cdots(p_N,a_N,b_N)$ over $\aword$ and $\aword'$ is winning for Duplicator iff the following conditions are satisfied for all $i,j \in \interval{1}{N}$:
\begin{itemize}
\itemsep 0 cm
\item $a_i=a_j$ iff $b_i=b_j$,
\item $a_{i}+1=a_{j}$ iff $b_{i}+1=b_{j}$,
\item $a_{i}<a_{j}$ iff $b_{i}<b_{j}$,
\item $\astruct(a_i) = \bstruct(b_i)$.
\end{itemize}
A $N$-round EF game over the $\omega$-words $\astruct,\bstruct$, denoted as $EF_N(\astruct,\bstruct)$, is said to be winning if there exists a strategy
$\astrategy_{D}$ for Duplicator such that for all strategies $\astrategy_{S}$ of spoiler, the play $\play{N}{\astrategy_S,\astrategy_D}{\aword,\aword'}$ is winning for Duplicator.  We write $\astruct\equivEF{N}\bstruct$ iff the game $EF_N(\astruct,\bstruct)$ is winning.
Theorem~\ref{theorem:EF} below states that   two $\omega$-words, the $N$-round game is winning iff these two $\omega$-words satisfy the same set of 
first-order formulae of quantifier height smaller than $N$.
\begin{theorem}[EF Theorem, see e.g.~\cite{Libkin04}]
\label{theorem:EF}
For any two $\omega$-words $\astruct,\bstruct$ over 
$\aalphabet$, 
$\astruct\equivEF{N}\bstruct$ iff $\astruct\equivrel{N}\bstruct$.
\end{theorem}
We will use EF games for FO$_{\aalphabet}$
formulae to prove a \emph{stuttering theorem} which will allow us to bound the number of times each loop needs to be taken in a path schema in 
order to satisfy a 
 FO$_{\aalphabet}$ 
formula. Note that in \cite{Etessami&Wilke00}, EF games have  been introduced for the specific case of LTL specifications here also to show 
some small model properties. 
 \subsection{Stuttering Theorem for FO$_{\aalphabet}$}
\label{section-stuttering-fo}

In this section, we prove that if in an $\omega$-sequence $\aword$, a subword $\mathbf{s}$ is repeated consecutively a large number of times, then 
this $\omega$-word and other $\omega$-words obtained by removing some of the repetitions of $\mathbf{s}$ satisfy the same set of 
FO$_{\aalphabet}$ sentences, this is what we call the stuttering theorem for 
FO$_{\aalphabet}$. 
Such a result will allow us to bound the repetition of iteration of loops in path schema and thus to obtain a model-checking 
algorithm for the logic 
FO$_{\aalphabet}$ optimal in complexity. In order to prove the stuttering theorem, we will use EF games.

In the sequel we consider a natural $N \geq 1$ and two $\omega$-words over 
$\aalphabet$ of the following form $\aword = \aword_1 \mathbf{s}^M \aword_2, \aword'= \aword_1 \mathbf{s}^{M+1} \aword_2 \in \aalphabet^{\omega}$ with 
$M > 2^{N+1}$, $\aword_1 \in \aalphabet^*,  \mathbf{s} \in \aalphabet^+$ and $\aword_2 \in \aalphabet^{\omega}$.  
We will now show that the game $EF_N(\astruct,\bstruct)$ is winning. The strategy for Duplicator 
will work as follows: at the $i$-th round (for $i \leq N$), if the point chosen by the Spoiler is close to another previously 
chosen position then the Duplicator will choose a point in the other word at the exact same distance from the corresponding position and if the point 
is far  from any other position  then in the other word the Duplicator will chose a position also far 
away from any other position. 

Before providing a winning strategy for the Duplicator we  define some invariants on any $i$-round play (with $i \leq N$) that will be maintained by 
the Duplicator's strategy. In order to define this invariant and the Duplicator's strategy, let introduce a few notations:
\begin{itemize}
\itemsep 0 cm
\item $a_{-3} = b_{-3} = 0$; $a_{-2} = b_{-2} = \length{\aword_1}$;
\item      $a_{-1}= \length{\aword_1\mathbf{s}^{M}}$ and $b_{-1}= \length{\aword_1\mathbf{s}^{M+1}}$;
\item $a_{0} = b_{0} = \omega$.  
\end{itemize}
We extend the substraction and addition operations in order to deal with $\Nat \cup \set{-\omega,\omega}$ such that:
$\alpha - \omega =  -\omega$, $\omega - \alpha = \omega$ and $\omega + \alpha = \omega$ if $\alpha \in \Nat$ 
(no need to define the other cases for what follows). 
The relation $<$ on $\Nat \cup \set{-\omega,\omega}$ is extended in the obvious way. 
Given a $i$-round play $\aplay_i=(p_1,a_1,b_1)(p_2,a_2,b_2) \cdots (p_i,a_i,b_i)$, we  say that  $\aplay_i$ respects the invariant 
$\aninv$ iff the following conditions are satisfied for all $j,k \in \interval{-3}{i}$:
\begin{enumerate}
\itemsep  0 cm
\item $a_{j} \leq a_{k}$ iff $b_{j} \leq b_{k}$,
\item $\mid a_j - a_k \mid <  2^{N+1-i} \length{\mathbf{s}}$ iff  $\mid b_j - b_k \mid <  2^{N+1-i} \length{\mathbf{s}}$,
\item $\mid a_j - a_k \mid <  2^{N+1-i} \length{\mathbf{s}}$ implies $a_j - a_k = b_j - b_k$,
\item $a_j \leq a_{-2}$ or $b_j \leq b_{-2}$  implies $b_j = a_j$,
\item $a_j \geq a_{-1}$ or $b_j \geq b_{-1}$ implies $b_j = a_j + \length{\mathbf{s}}$,
\item $a_{-2} < a_j < a_{-1}$ or $b_{-2} < b_j < b_{-1}$ implies
      $\mid a_j - b_j \mid = 0 \ {\rm mod} \ \length{\mathbf{s}}$. 
\end{enumerate}
First we remark the invariant $\aninv$ is a sufficient condition for a play to be winning as stated by the following lemma.
\begin{lemma}
\label{lem-inv}
If a $N$-round play over $\aword$ and $\aword'$ respects $\aninv$, then it is a winning play for the Duplicator.
\end{lemma}

\begin{proof}
Let $(p_1,a_1,b_1)(p_2,a_2,b_2)\cdots(p_N,a_N,b_N)$ be a $N$-round play over $\aword$ and $\aword'$ respecting $\aninv$.
 Let $i,j \in \interval{1}{N}$. It is easy to see that satisfaction of $\aninv$ implies that 
 $a_i=a_j$ iff $b_i=b_j$,  $a_i < a_j$ iff $b_i < b_j$, and  $a_i+1=a_j$ iff $b_i+1=b_j$. Moreover, Condition
 $\aninv$(4--6) obviously guarantees that $\astruct(a_j) = \bstruct(b_j)$. \qed
\end{proof}

Given an $(i-1)$-round play $\aplay_{i-1}=(p_1,a_1,b_1)(p_2,a_2,b_2) \cdots (p_{i-1},a_{i-1},b_{i-1})$ and $a_i \in \Nat$ such that
$a_i \not \in \set{a_{-3},a_{-2}, \ldots, a_{i-2}, a_{i-1}}$, we define $\lef{a_i}=\max(a_k \mid k \in \interval{-3}{i-1} \mbox{ and } a_k < a_i)$ and 
$\rig{a_i}=\min(a_k \mid k \in \interval{-3}{i-1} \mbox{ and } a_i< a_k)$ (i.e.~$\lef{a_i}$ and $\rig{a_i}$ are the closest neighbor of $a_i$). 
We define similarly $\lef{b_i}$ and $\rig{b_i}$.

We define now a strategy $\hat\astrategy_D$ for the Duplicator that respects at each round the invariant $\aninv$ and this no matter what the 
Spoiler plays. By Lemma~\ref{lem-inv}, we can conclude that this strategy is winning  for the Duplicator.
Let $i \in \interval{1}{N}$ and $\aplay_{i-1}=(p_1,a_1,b_1)(p_2,a_2,b_2)\cdots(p_{i-1},a_{i-1},b_{i-1})$ be a $(i-1)$-round play. First, we define 
$b_i=\astrategy_D(\aplay_{i-1},\tuple{0,a_i})$ that is what Duplicator answers if the Spoiler chooses position $a_i$ in the $\omega$-word 
$\aword$. We have $b_i=\hat\astrategy_D(\aplay_{i-1},\tuple{0,a_i})$ defined as follows:
\begin{itemize}
\itemsep 0 cm
\item If $a_i = a_j$ for some $j \in \interval{-3}{i-1}$, then $b_i \egdef b_j$;
\item Otherwise, let $a_l=\lef{a_i}$ and $a_r=\rig{a_i}$:
\begin{itemize}
\itemsep 0 cm 
\item  If $a_i-a_l \leq a_r-a_i$, we have $b_i \egdef b_l+(a_i-a_l)$
\item  If $a_r-a_i < a_i-a_l$, we have $b_i \egdef b_r-(a_r-a_i)$
\end{itemize}
\end{itemize}
Similarly we have $a_i=\hat\astrategy_D(\aplay_{i-1},\tuple{1,b_i})$ defined as follows:
\begin{itemize}
\itemsep 0 cm
\item If $b_i = b_j$ for some $j \in \interval{-3}{i-1}$, then $a_i \egdef a_j$;
\item Otherwise, let $b_l=\lef{b_i}$ and $b_r=\rig{b_i}$:
\begin{itemize}
\itemsep 0 cm 
\item  If $b_i-b_l \leq b_r-b_i$, we have $a_i \egdef a_l+(b_i-b_l)$
\item  If $b_r-b_i < b_i-b_l$, we have $a_i \egdef a_r-(b_r-b_i)$
\end{itemize}
\end{itemize}

\begin{lemma}
\label{lem-strat-win}
For any Spoiler's strategy $\astrategy_S$ and for all $i \in \interval{0}{N}$, we have that $\play{i}{\astrategy_S,\hat\astrategy_D}{\aword,\aword'}$ respects $\aninv$.
\end{lemma}


\begin{proof}
The proof proceeds by induction on $i$. The base case for $i=0$ is obvious since the empty play respects $\aninv$. 
However, we need to use the fact that $M > 2^{N+1}$ (otherwise condition $\aninv.2$ might not hold). 

Let $\astrategy_S$ be a Spoiler's strategy and for every $i \in \interval{1}{N-1}$, we assume that 
$\play{i-1}{\astrategy_S,\hat\astrategy_D}{\aword,\aword'}$ respects $\aninv$. Suppose that 
$\astrategy_S(\play{i-1}{\astrategy_S,\hat\astrategy_D}{\aword,\aword'})=\tuple{0,a_i}$ and let 
$b_i=\hat\astrategy_D(\play{i-1}{\astrategy_S,\hat\astrategy_D}{\aword,\aword'},\tuple{0,a_i})$. 

\noindent 
(1.) Let $j,k \in \interval{-3}{i}$. If  $j,k \in \interval{-3}{i-1}$, then by the induction hypothesis, 
$a_{j} \leq a_{k}$ iff $b_{j} \leq b_{k}$. Otherwise, let suppose $j= i$ and $k \neq i$ (only remaining interesting case). If
$a_i = a_{j'}$ for some $j' \in \interval{-3}{i-1}$, then $b_i = b_{j'}$ and therefore
$a_{i} \leq a_{k}$ iff $b_{i} \leq b_{k}$ by the induction hypothesis. Otherwise, $a_l < a_i < a_r$ and  $b_l < b_i < b_r$
which entails that $a_{i} \leq a_{k}$ iff $b_{i} \leq b_{k}$.

\noindent
(4.)  The case $j \in \interval{-3}{i-1}$ is immediate from the induction hypothesis. 
Now, suppose that $a_{-3} \leq a_i \leq a_{-2}$. 
If $a_i = a_{j}$ for some $j \in \interval{-3}{i-1}$, then $b_i = b_{j}$  and  $a_{-3} \leq a_j \leq a_{-2}$.
By induction hypothesis, $b_i = b_j = a_j = a_i$. Otherwise,  $a_l < a_i < a_r$ and $a_l = b_l$ and $a_r = b_r$ 
by induction hypothesis.  Either 
 $a_i-a_l \leq a_r-a_i$ or $a_r-a_i < a_i-a_l$ implies that $b_i = a_i$.

\noindent
(5.)  
 The case $j \in \interval{-3}{i-1}$ is immediate from the induction hypothesis. 
Now, suppose that $a_{-1} \leq a_i$. 
If $a_i = a_{j}$ for some $j \in \interval{-3}{i-1}$, then $b_i = b_{j}$  and  $a_{-1} \leq a_j$.
By induction hypothesis, $b_i = b_j = a_j + \length{\mathbf{s}} = a_i + \length{\mathbf{s}}$. 
Otherwise,  $a_l < a_i$ and $b_l = a_l  + \length{\mathbf{s}}$ 
by induction hypothesis.  
Since $a_r = \omega$, we have $b_i = b_l + (a_i - a_l) = a_i  + \length{\mathbf{s}}$. 

\noindent
(6.)  The case $j \in \interval{-3}{i-1}$ is immediate from the induction hypothesis. 
Now, let us deal with $j = i$. 
Satisfaction of (4.) and (5.) implies that 
 $a_{-2} < a_i < a_{-1}$  iff $b_{-2} < b_i < b_{-1}$.
Suppose that $a_{-2} < a_i < a_{-1}$. So,  $a_{l} < a_i < a_{r}$
and by induction hypothesis $\mid a_l - b_l \mid = 0 \ {\rm mod} \ \length{\mathbf{s}}$
and  $\mid a_r - b_r \mid = 0 \ {\rm mod} \ \length{\mathbf{s}}$.
If  $a_i-a_l \leq a_r-a_i$, then $b_i = b_l + (a_i -a_l)$ and 
 $\mid a_i - b_i \mid = \mid a_l - b_l \mid$, whence $\mid a_i - b_i \mid = 0 \ {\rm mod} \ \length{\mathbf{s}}$.
Similarly, if  $a_r-a_i < a_i-a_l$, then $b_i = b_r - (a_r -a_i)$
and $\mid a_i - b_i \mid = \mid a_r - b_r \mid$, whence $\mid a_i - b_i \mid = 0 \ {\rm mod} \ \length{\mathbf{s}}$.

\noindent 
(2.--3.)  Let $j,k \in \interval{-3}{i}$. If  $j,k \in \interval{-3}{i-1}$, then by the induction hypothesis, it is easy to verify
that
\begin{itemize}
\itemsep 0 cm 
\item $\mid a_j - a_k \mid <  2^{N+1-i} \length{\mathbf{s}}$ iff  $\mid b_j - b_k \mid <  2^{N+1-i} \length{\mathbf{s}}$,
\item $\mid a_j - a_k \mid <  2^{N+1-i} \length{\mathbf{s}}$ implies $a_j - a_k = b_j - b_k$.
\end{itemize}
Indeed, it is a consequence of the stronger properties below satisfied by induction hypothesis:
\begin{itemize}
\itemsep 0 cm 
\item $\mid a_j - a_k \mid <  2^{N+2-i} \length{\mathbf{s}}$ iff  $\mid b_j - b_k \mid <  2^{N+2-i} \length{\mathbf{s}}$,
\item $\mid a_j - a_k \mid <  2^{N+2-i} \length{\mathbf{s}}$ implies $a_j - a_k = b_j - b_k$.
\end{itemize}
Otherwise, let suppose $j= i$ and $k \neq i$ (only remaining interesting case).
If $a_i = a_{j'}$ for some $j' \in \interval{-3}{i-1}$, then by the induction hypothesis, we have
\begin{itemize}
\itemsep 0 cm 
\item $\mid a_{j'} - a_k \mid <  2^{N+2-i} \length{\mathbf{s}}$ iff  $\mid b_{j'} - b_k \mid <  2^{N+2-i} \length{\mathbf{s}}$,
\item $\mid a_{j'} - a_k \mid <  2^{N+2-i} \length{\mathbf{s}}$ implies $a_{j'} - a_k = b_{j'} - b_k$.
\end{itemize}
Again, this implies that 
\begin{itemize}
\itemsep 0 cm 
\item[(I)] $\mid a_i - a_k \mid <  2^{N+1-i} \length{\mathbf{s}}$ iff  $\mid b_i - b_k \mid <  2^{N+1-i} \length{\mathbf{s}}$,
\item[(II)] $\mid a_i - a_k \mid <  2^{N+1-i} \length{\mathbf{s}}$ implies $a_i - a_k = b_i - b_k$.
\end{itemize}
Now, suppose that there is no $j' \in \interval{-3}{i-1}$ such that $a_i = a_{j'}$.

\noindent
{\em Case 1}:  $a_i-a_l \leq a_r-a_i$ and  $b_i = b_l+(a_i-a_l)$.

\noindent
{\em Case 1.1}: $a_k \leq a_l$. 
\begin{itemize}
\itemsep 0 cm  
\item $a_k = a_l$:  $\mid a_i - a_k \mid  =  \mid b_i - b_k \mid$  and therefore (I)-(II) holds.
\item $a_k < a_l$:  If $\mid a_l - a_k \mid \geq 2^{N+1-i} \length{\mathbf{s}}$, by induction hypothesis 
                     $\mid b_l - b_k \mid \geq 2^{N+1-i} \length{\mathbf{s}}$ and $b_k < b_l$.
                   So,  $\mid a_i - a_k \mid \geq  2^{N+1-i} \length{\mathbf{s}}$ and  $\mid b_i - b_k \mid \geq  2^{N+1-i}$.

                    If $\mid a_i - a_l \mid \geq 2^{N+1-i} \length{\mathbf{s}}$, by definition of $b_l$,
                     $\mid a_i - a_l \mid = \mid b_i - b_l \mid \geq 2^{N+1-i} \length{\mathbf{s}}$ and $b_k < b_l$.
                   So,  $\mid a_i - a_k \mid \geq  2^{N+1-i} \length{\mathbf{s}}$ and  $\mid b_i - b_k \mid \geq  2^{N+1-i}$.

                   If  $\mid a_l - a_k \mid \leq 2^{N+1-i} \length{\mathbf{s}}$ and $\mid a_i - a_l \mid \leq 2^{N+1-i} \length{\mathbf{s}}$,
                   then by induction hypothesis  $\mid b_l - b_k \mid = \mid a_l - a_k \mid$, 
                    $\mid a_i - a_l \mid = \mid b_i - b_l \mid$, $a_k < a_l < a_i$ and $b_k < b_l < b_i$. 
                   So  $\mid a_i - a_k \mid = \mid b_i - b_k \mid$, whence (I)--(II)
                   holds. 
\end{itemize}

\noindent
{\em Case 1.2}: $a_r \leq a_k$. 

\begin{itemize}
\item $\mid a_k - a_r \mid \geq 2^{N+1-i} \length{\mathbf{s}}$: By induction hypothesis, $\mid b_k - b_r \mid \geq 2^{N+1-i} \length{\mathbf{s}}$.
      So,  $\mid a_k - a_i \mid \geq 2^{N+1-i} \length{\mathbf{s}}$ and $\mid b_k - b_i \mid \geq 2^{N+1-i} \length{\mathbf{s}}$ since
      $a_i < a_r \leq a_k$ and $b_i < b_r \leq b_k$.
\item $\mid a_k - a_r \mid \leq 2^{N+1-i} \length{\mathbf{s}}$: By induction hypothesis, $\mid b_k - b_r \mid = \mid a_k - a_r \mid$. 
      
      \noindent
      {\em Case 1.2.1.} $\mid a_r - a_l \mid \leq 2^{N+2-i} \length{\mathbf{s}}$. By induction hypothesis, 
       $\mid a_r - a_l \mid = \mid b_r - b_l \mid$ and therefore  $\mid b_r - b_i \mid = \mid a_r - a_i \mid$.
      Whence,  $\mid b_k - b_i \mid = \mid a_k - a_i \mid$ so (I)-(II) holds.

      \noindent
      {\em Case 1.2.2}  $\mid a_r - a_l \mid \geq 2^{N+2-i} \length{\mathbf{s}}$. By induction hypothesis, 
       $\mid b_r - b_l \mid \geq 2^{N+2-i} \length{\mathbf{s}}$. Moreover, since  $a_i-a_l \leq a_r-a_i$, 
      $a_r - a_i \geq 2^{N+1-i} \length{\mathbf{s}}$.  Since $b_i - b_l = a_i -b_l$, we have $b_r - b_i \geq 2^{N+1-i} \length{\mathbf{s}}$ too.
      So,  $a_k - a_i \geq 2^{N+1-i} \length{\mathbf{s}}$ and  $b_k - b_i \geq 2^{N+1-i} \length{\mathbf{s}}$, which 
      guarantees (I)--(II). 
 \end{itemize}

\noindent
{\em Case 2}: $a_r-a_i < a_i-a_l$ and  $b_i = b_r-(a_r-a_i)$. 

\noindent
{\em Case 2.1}: $a_k \geq a_r$. Similar to Case 1.1  by replacing $a_l$ by $a_r$, $b_l$ by $b_r$ and, by permuting '$<$' by '$>$' and '$\leq$' by '$\geq$'  
about positions.

\noindent
{\em Case 2.2}: $a_k \leq a_l$. Similar to Case 1.2 by replacing $a_r$ by $a_l$, $b_r$ by $b_l$ and, by permuting '$<$' by '$>$' and '$\leq$' by '$\geq$' 
about positions.  \qed
\end{proof}
%

Using Lemma~\ref{lem-inv} and~\ref{lem-strat-win}, we deduce that Duplicator has a winning strategy against any strategy of the Spoiler in $EF_N(\astruct,\bstruct)$, so by Theorem \ref{theorem:EF}, we can conclude \ifconference \textbf{Theorem~\ref{theorem-stuttering-fo} [Stuttering Theorem]}.\else the main theorem of this subsection.

\begin{theorem}[Stuttering Theorem]
\label{theorem-stuttering-fo}
Let $\aword = \aword_1 \mathbf{s}^M \aword_2, \aword'= \aword_1 \mathbf{s}^{M+1} \aword_2 \in \aalphabet^{\omega}$  such that $N \geq 1$, 
$M> 2^{N+1}$ and $\mathbf{s} \in \aalphabet^+$. Then $\astruct\equivrel{N}\bstruct$.
\end{theorem}
\fi

\else
 \iftechreport
Ehrenfeucht-Fra\"iss\'e (EF) game is a well known technique to determine whether two structures are equivalent with respect to a set of formulae. We recall here the definition of a EF game adapted to our context. Given $N \in \nat$ and two $\omega$-words $\aword,\aword'$ over 
$\aalphabet$, 
the main idea of the corresponding EF game is that two players, the Spoiler and the Duplicator, plays in a turn based manner. 
The Spoiler begins by choosing a word between $\aword$ and $\aword'$ and a position in this word, then the Duplicator aims at finding a position in the other word which is \emph{similar} and this during $N$ rounds. At the end, the Duplicator wins if the set of chosen positions respects some 
isomorphism.
Using EF games we can prove the stuttering theorem for FO$_{\aalphabet}$.
\else
Let us state our first result about $\FO_{\aalphabet}$ which  allows us to bound the number of times each loop is taken 
in a constrained path schema in order to satisfy a formula. We provide a stuttering theorem equivalent for $FO_{\aalphabet}$ formulas as is done in \cite{DemriDharSangnier12} for PLTL and in \cite{Kucera&Strejcek05} for LTL. The lengthy proof of 
Theorem~\ref{theorem-stuttering-fo}  uses Ehrenfeuch-Fra\"iss\'e game (Appendix~\ref{appendix:ef-games}).
\fi
\begin{theorem}[Stuttering Theorem]
\label{theorem-stuttering-fo}
Let $\aword = \aword_1 \mathbf{s}^M \aword_2, \aword'= \aword_1 \mathbf{s}^{M+1} \aword_2 \in \aalphabet^{\omega}$  such that $N \geq 1$, 
$M> 2^{N+1}$ and $\mathbf{s} \in \aalphabet^+$. Then $\astruct\equivrel{N}\bstruct$.
\end{theorem}

\fi
%
\subsection{Model-Checking Flat Counter Systems with FO}
\label{section-pspace-fo}
Let us characterize the complexity of $\mcpb{\FO}{\flatcs}$. 
First, we will state the complexity of the intersection non-emptiness problem. Given  
a constrained path schema $\acps$ and 
a $\FO$ sentence $\aformulabis$, Theorem~\ref{theorem-pottier91} provides two polynomials 
$\mathtt{pol}_1$ and $\mathtt{pol}_2$ to represent succinctly the solutions of the guard in $\acps$.
Theorem~\ref{theorem-stuttering-fo} allows us to bound the number of times loops are visited. 
Consequently, we can compute a value $\amap_{\FO}(\aformulabis,\acps)$ exponential in the size of 
$\aformulabis$ and $\acps$, as explained  earlier,  which allows us to find a witness for the intersection 
non-emptiness problem where each loop is taken a number of times smaller than $\amap_{\FO}(\aformulabis,\acps)$.

\begin{lemma}\label{lemma:fo-cps-small-loop}
Let $\acps$ be a constrained path schema and $\aformulabis$ be a FO$_{\aalphabet}$ sentence. Then 
$\alang(\acps) \cap \alang(\aformulabis)$ is non-empty iff there is an $\omega$-word in 
$\alang(\acps) \cap \alang(\aformulabis)$ in which each loop is taken at most
$
2^{(\qheight{\aformulabis} + 2) + \mathtt{pol}_1(\size{\acps}) + \mathtt{pol}_2(\size{\acps})}$ times.
\end{lemma}

Hence $\amap_{{\rm FO}}(\aformulabis,\acps)$ has the value $2^{(\qheight{\aformulabis}+2) 
+ (\mathtt{pol}_1+\mathtt{pol}_2)(\size{\acps})}$. 
\iftechreport
\begin{proof} Let $\acps = \pair{\aseg_1 (\aloop_1)^* \cdots \aseg_{k-1} (\aloop_{k-1})^* \aseg_k (\aloop_k)^{\omega}}{\aformula(\avariable_1, 
\ldots, \avariable_{k-1})}$ be a constrained path schema and $\aformulabis$ be a first-order sentence. 
Suppose that 
$$
\aseg_1 (\aloop_1)^{\vect{n}[1]} \cdots \aseg_{k-1} (\aloop_{k-1})^{\vect{n}[k-1]} \aseg_k (\aloop_k)^{\omega} 
\in \alang(\acps) \cap \alang(\aformulabis).
$$
Let  $\basis \subseteq  \interval{0}{2^{p_1(\size{\acps})}}^{k-1}$ and 
           $\aperiod_1, \ldots, \aperiod_{\alpha} \in \interval{0}{2^{p_1(\size{\acps})}}^{k-1}$ defined for the guard $\aformula$
following Theorem~\ref{theorem-pottier91}. 
Since $\vect{n} \models \aformula$,
there are $\abasis \in \basis$ and $\vect{a} \in \Nat^{\alpha}$ such that
           $\vect{n} = \abasis + \vect{a}[1] \aperiod_1 + \cdots + \vect{a}[\alpha] \aperiod_{\alpha}$.
Let  $\vect{a}' \in \Nat^{\alpha}$ defined from $\vect{a}$ such that
$\vect{a}'[i] = \vect{a}[i]$ if $\vect{a}[i] \leq 2^{\qheight{\aformulabis} +1} + 1$ otherwise 
$\vect{a}'[i] = 2^{\qheight{\aformulabis} +1} + 1$. Note that 
$\vect{n}'  = \abasis + \vect{a}'[1] \aperiod_1 + \cdots + \vect{a}'[\alpha] \aperiod_{\alpha}$ still satisfies
$\aformula$ and for every loop $i \in \interval{1}{k-1}$, 
$\vect{n}[i] >  2^{\qheight{\aformulabis} +1}$ iff $\vect{n}'[i] >  2^{\qheight{\aformulabis} +1}$.
By  Theorem~\ref{theorem-stuttering-fo}, 
$\aseg_1 (\aloop_1)^{\vect{n}'[1]} \cdots \aseg_{k-1} (\aloop_{k-1})^{\vect{n}'[k-1]} \aseg_k (\aloop_k)^{\omega} 
\in \alang(\aformulabis)$. 

Now, let us bound the values in  $\vect{n}'$. 
\begin{itemize}
\itemsep 0 cm
\item There are at most $2^{p_2(\size{\acps})}$ periods.
\item Each basis or period has values in $\interval{0}{2^{p_1(\size{\acps})}}$.
\item Each period in $\vect{n}'$ is taken at most $2^{\qheight{\aformulabis} +1} + 1$ times.
\end{itemize}
Consequently, each $\vect{n}'[i]$ is bounded by
$$
2^{p_1(\size{\acps})} + (2^{\qheight{\aformulabis} +1} + 1) 2^{p_2(\size{\acps})} \times 2^{p_2(\size{\acps})}
$$
which is itself bounded by $2^{(\qheight{\aformulabis} + 2) + p_1(\size{\acps}) + p_2(\size{\acps})}$. 
\qed
\end{proof}
\fi
Furthermore checking whether  $\alang(\acps) \cap \alang(\aformulabis)$  is non-empty amounts to guess some
$\vect{n} \in \interval{0}{2^{(\qheight{\aformulabis} + 2) + \mathtt{pol}_1(\size{\acps}) + \mathtt{pol}_2(\size{\acps})}}^{k-1}$ and verify
whether 
\iftechreport
$$\aword = \aseg_1 (\aloop_1)^{\vect{n}[1]} \cdots \aseg_{k-1} (\aloop_{k-1})^{\vect{n}[k-1]} \aseg_k (\aloop_k)^{\omega} \in 
\alang(\acps) \cap \alang(\aformulabis).
$$  
\else
$\aword = \aseg_1 (\aloop_1)^{\vect{n}[1]}\linebreak[0] \cdots\linebreak[0] \aseg_{k-1} (\aloop_{k-1})^{\vect{n}[k-1]} \aseg_k (\aloop_k)^{\omega} \in 
\alang(\acps) \cap \alang(\aformulabis).
$ 
\fi
Checking if $\aword \in \alang(\acps)$ can be done in polynomial time
in $(\qheight{\aformulabis} + 2) + \mathtt{pol}_1(\size{\acps}) + \mathtt{pol}_2(\size{\acps})$ 
(and therefore in polynomial time in $\size{\aformulabis}
+ \size{\acps}$) since this amounts to verify whether $\vect{n} \models \aformula$. 
Checking whether $\aword \in \alang(\aformulabis)$ can be done in exponential space in
$\size{\aformulabis}
+ \size{\acps}$ by using~\cite[Proposition 4.2]{Markey&Schnoebelen03}.
\iftechreport  
Indeed, checking whether $u \cdot (v)^{\omega} \models \aformula$ (where $u \cdot (v)^{\omega}$ is an ultimately
periodic word and $\aformula$ is a first-order sentence) can be done in space $\mathcal{O}(\size{uv} \size{\aformula}^2)$. 
When the ultimately periodic word is $\aword$ this provides a space bound in
$\mathcal{O}(\size{\acps} \times 2^{(\qheight{\aformulabis} + 2) + \mathtt{pol}_1(\size{\acps}) + \mathtt{pol}_2(\size{\acps})} \times  \size{\aformulabis}^2)$.
\fi
\iftechreport
Hence, this leads to a nondeterministic exponential space decision procedure for the intersection non-emptiness
problem but it is possible to get down to
nondeterministic polynomial space as explained below. 

Now, we show that the membership problem with first-order logic  can be solved in
polynomial space in $\size{\acps} + \size{\aformulabis}$. 
Let $\acps$,  $\aformulabis$ and $\vect{n} \in \Nat^{k-1}$ be an instance of the problem.
For $i \in \interval{1}{k-1}$, let $\vect{n}'[i] \egdef \mathtt{min}(\vect{n}[i],2^{\qheight{\aformulabis}+1} +1)$. 
By  Theorem~\ref{theorem-stuttering-fo},  the propositions below are equivalent:
\begin{itemize}
\item  $\aseg_1 (\aloop_1)^{\vect{n}[1]} \cdots \aseg_{k-1} (\aloop_{k-1})^{\vect{n}[k-1]} \aseg_k (\aloop_k)^{\omega} \in 
\alang(\aformulabis)$,
\item $\aseg_1 (\aloop_1)^{\vect{n}'[1]} \cdots \aseg_{k-1} (\aloop_{k-1})^{\vect{n}'[k-1]} \aseg_k (\aloop_k)^{\omega} \in 
\alang(\aformulabis)$.
\end{itemize}
Without any loss of generality, let us assume then that 
$\vect{n} \in \interval{0}{2^{\qheight{\aformulabis}+1} +1}^{k-1}$.

Let us decompose $\aword = \aseg_1 (\aloop_1)^{\vect{n}[1]} \cdots \aseg_{k-1} 
(\aloop_{k-1})^{\vect{n}[k-1]} \aseg_k (\aloop_k)^{\omega}$ as
$u \cdot (v)^{\omega}$ where $u = \aseg_1 (\aloop_1)^{\vect{n}[1]} \cdots \aseg_{k-1} (\aloop_{k-1})^{\vect{n}[k-1]} \aseg_k$
and $v = \aloop_k$. 
Note that the length of $u$ is exponential in the size of the instance.  
We write $\hat{\aformulabis}$ to denote the formula $\aformulabis$ in which every existential quantification is
relativized to positions less than $\length{u} + \length{v} \times 2^{\qheight{\aformulabis}}$. 
This means that every quantification '$\exists \ \avariable \ \cdots$' is replaced by
'$\exists \ \avariable <  (\length{u} + \length{v} \times 2^{\qheight{\aformulabis}}) \ \cdots$'. 
By~\cite{Markey&Schnoebelen03}, we know that $\aword \models \aformulabis$ iff 
$\aword \models \hat{\aformulabis}$. Now, checking $\aword \models \hat{\aformulabis}$ can be done in polynomial space 
by using a standard first-order model-checking algorithm by restricting ourselves to positions in
$\interval{0}{\length{u} + \length{v} \times 2^{\qheight{\aformulabis}}}$ for existential quantifications.
Such positions can be obviously encoded in polynomial space. Moreover, note that given 
$i \in \interval{0}{\length{u} + \length{v} \times 2^{\qheight{\aformulabis}}}$, one can check in polynomial time
what is the $i$th letter of $\aword$. Details are standard and omitted here. By way of example, the $i$th letter of $\aword$
is the first letter of $l_k$ iff 
$i \geq \alpha$ and $(i-\alpha) = 0 \ {\rm mod} \ \length{\aloop_k})$
with $\alpha = (\Sigma_{j \in \interval{1}{k-1}} (\length{p_j} + \length{l_j} \times \vect{n}[j])) + \length{p_j}$.  

\begin{algorithm}
{\footnotesize
\caption{FOSAT$(\acps,\vect{n},\aformulabis,\amap)$}
\label{algorithm:FOSAT}
\begin{algorithmic}[1]
\IF{$\aformula = \aletter(\avar)$}
   \STATE Calculate the $\amap(\avar)$th letter $\aletterbis$ of $\aword$ and  return $\aletter = \aletterbis$.
\ELSIF{$\aformulabis$ is of the form $\neg \aformulabis'$}
      \STATE return not FOSAT$(\acps,\vect{n}, \aformulabis',\amap)$ 
\ELSIF{$\aformulabis = \aformulabis_1 \wedge \aformulabis_2$}
   \STATE return FOSAT$(\acps,\vect{n},\aformulabis_1,\amap)$ and FOSAT$(\acps,\vect{n},\aformulabis_2,\amap)$. 
\ELSIF{$\aformulabis$ is of the form $\exists \avar < m \  \aformulabis'$}
   \STATE guess a position $k\in\interval{0}{m-1}$.
   \STATE return FOSAT$(\acps,\vect{n},\aformulabis',\amap[\avar \mapsto k])$.
\ELSIF{$\aformula$ is of the form $R(\avar,\avar')$ for some $R \in \set{=,<,S}$}
   \STATE return $R(\amap(\avar),\amap(\avar'))$.
\ENDIF
\end{algorithmic}
}
\end{algorithm}

Polynomial space algorithm is obtained by computing
FOSAT$(\acps,\vect{n}, \aformulabis, \amap_0)$ with the algorithm FOSAT defined below ($\amap_0$ is a zero assignment
function). Note that the polynomial space bound is obtained since the recursion depth is linear in
$\size{\aformulabis}$ and positions in $\interval{0}{\length{u} + \length{v} \times 2^{\qheight{\aformulabis}}}$
can be encoded in polynomial space in $\size{\acps} + \size{\aformulabis}$. Furthermore, since model-checking ultimately periodic words with first-order logic is \pspace-hard~\cite{Markey&Schnoebelen03}, we deduce directly the lower bound for the membership problem with FO. This allows us to state the following lemma for the lower bound being a consequence of model-checking ultimately periodic words with first-order logic which is \pspace-hard\cite{Markey&Schnoebelen03}).

\begin{lemma} 
\label{lemma:fo-membership}
\iftechreport
Membership problem with first-order logic is \pspace-complete and it can be solved in
           polynomial space in $\size{\acps} + \size{\aformula}$.
\else
Membership problem for FO$_{\aalphabet}$ is \pspace-complete.
\fi 
\end{lemma}
Note that the membership problem for FO$_{\aalphabet}$ is for unconstrained alphabet, but due to the nice subalphabet property of FO and the conversion of instance of membership problem for FO formulas to FO$_{\aalphabet}$, the same holds true for constrained alphabet.
Thanks to this lemma and to the result stated in Lemma \ref{lemma:fo-cps-small-loop}, we  obtain the point of following theorem .
%

\begin{theorem} \label{theorem-intersection-fo}
The intersection non-emptiness problem with first-order logic is \pspace-complete. 
\end{theorem}  

By Lemma~\ref{lemma-complexity-subalphabet}, Lemma~\ref{lemma-fo-subalphabet-property} and  Theorem~\ref{theorem-intersection-fo}, 
we get the following result. 

\begin{theorem} \label{theorem-fo-pspace-complete}
$\mcpb{\FO}{\flatcs}$ is \pspace-complete. 
\end{theorem}

\pspace-hardness is obtained by noting that 
 model-checking FO over finite flat Kripke
structures with two states is known to be \pspace-hard (see also~\cite{Markey&Schnoebelen03}). 
Furthermore, since flat Kripke structures  form a subclass of flat counter systems, we conclude this new result for flat Kripke structures
(the lower bound being a direct consequence of Theorem~\ref{theorem-intersection-fo}).

\begin{corollary}
Model-checking flat Kripke structures over FO formulae is \pspace-complete.
\end{corollary}


\else
Hence, this leads to a nondeterministic exponential space decision procedure for the intersection non-emptiness
problem but it is possible to get down to
nondeterministic polynomial space using the succinct representation of constrained path schema as stated by 
Lemma~\ref{lemma:fo-membership} below for which the lower bound is deduced by the fact that model-checking 
ultimately periodic words with first-order logic is \pspace-hard~\cite{Markey&Schnoebelen03}. 

\begin{lemma}
\label{lemma:fo-membership}
\iftechreport
 Membership problem with first-order logic is \pspace-complete and it can be solved in
           polynomial space in $\size{\acps} + \size{\aformula}$.
\else
Membership problem with $\FO_{\aalphabet}$ is \pspace-complete.
\fi 
\end{lemma}
Note that the membership problem for FO is for unconstrained alphabet, but due to the nice subalphabet property of FO, the same holds for constrained alphabet since given a FO  formula over $\triple{at}{ag_n}{\aalphabet}$, we can build in polynomial time a FO formula over $\triple{at}{ag_n}{\aalphabet'}$ from which we can build also in polynomial time a formula of  $\FO_{\aalphabet'}$ (where $\aalphabet'$ is for instance the alphabet labeling a constrained path schema). We can now state the main results concerning FO.
%

\begin{theorem} \label{theorem-fo}
\iftechreport
\begin{description}
\item[(I)]
The intersection non-emptiness problem with FO is \pspace-complete.
\item[(II)] $\mcpb{\FO}{\flatcs}$ is \pspace-complete.
\item[(III)] Model-checking flat Kripke structures with FO is \pspace-complete.
\end{description} 
\else
(I) The intersection non-emptiness problem with FO is \pspace-complete.
(II) $\mcpb{\FO}{\flatcs}$ is \pspace-complete.
(III) Model-checking flat Krip\-ke structures with FO is \pspace-complete.
\fi 
\end{theorem}

\begin{proof}
(I) is a consequence of Lemma~\ref{lemma:fo-cps-small-loop} and Lemma~\ref{lemma:fo-membership}. 
We obtain (II) from (I) by applying Lemma \ref{lemma-complexity-subalphabet} and  
Lemma~\ref{lemma-fo-subalphabet-property}. (III) is obtained by observing that flat Kripke structures  
form a subclass of flat counter systems. To obtain the lower bound, we use that model-checking 
ultimately periodic words with first-order logic is \pspace-hard~\cite{Markey&Schnoebelen03}.\qed
\end{proof}
\fi

\iftechreport 
\section{Taming Linear $\mu$-calculus and Other Languages}
In this section, we consider several specification languages defining $\omega$-regular properties
on atomic propositions and arithmetical constraints. First, we deal with BA by mainly establishing 
Theorem~\ref{thm:stuttering-buchi}
and then we draw the consequences for ABA, ETL and linear $\mu$-caculus.
\cut{
In this section, we are going to provide the complexity bounds for formulas in Linear $\mu$-calculus. In fact, we will also provide the complexity
bounds for other specification languages like BA, ETL, ABA. We will provide a general treatment for the upper bound for all these
specification. Since, this is achieved by transforming each of the specification to \buchi~automata, we will first prove a property for \buchi~automata.
}
\subsection{No Need to Visit Loops Too Often}

We denote a \buchi~automaton as $\bauto=\tuple{\states,\aalphabet,\astate_0,\edges,F}$ where as usual $\states$ is the finite set of
states, $\aalphabet$ are the alphabets, $\astate_0\in\states$ is the initial state and $F\subseteq \states$ is the set of final states respectively. The set
of transitions is defined as $\edges\subseteq\states\times\aalphabet\times\states$. An accepting run $\arun$ for a word $\aword$ in $\bauto$ is
defined as an infinite sequence of states $\arun\in\states^{\omega}$ such that for every $i\in\nat$, $\tuple{\arun(i),\aword(i),\arun(i+1)}\in\edges$
and $\arun(j)\in F$ for infinitely many $j$'s.

In this section we will establish the function $\amap_{{\rm \auto}}$ as described in Section 3.3 for specifications given in the form of
\buchi~automaton. For any alphabet $\aalphabet$, we will consider a \buchi~automaton $\bauto=\tuple{\states,\aalphabet,\astate_0,\edges,F}$ and a
constrained path schema $\acps=\pair{\cpschema}{\aconstraintsystem}$. In the sequel we will denote $x=\size{\aconstraintsystem}$. 
\begin{lemma}
\label{lemma-buchi-aux}
Let $\aword\in\alang(\bauto)$, such that $\aword=\aword_1.u^{2.|\states|^k}.\aword_2$ for some $k$, then there exist an integer $K\in[1,|\states|]$
such that for all $N\in[1,|\states|^{k-2}]$, $\aword_1.u^{2.|\states|^k-(K\times N)}.\aword_2\linebreak[0]\in\alang(\bauto)$.
\end{lemma}

\ifconference
\else
\begin{proof}
Let $\bauto=\tuple{\states,\aalphabet,\astate_0,\edges,F}$. Since $\aword=\aword_1.u^{2.\card{\states}^k}.\aword_2\in\alang(\bauto)$, there exists an
accepting run $\arun\in\states^{\omega}$ for $\aword$. We will construct an accepting run for $\aword'=\aword_1.u^{2.\card{\states}^k-(K\times
  N)}.\aword_2$ in $\bauto$ using $\arun$. In $\aword$, $u$ is repeated $2.\card{\states}^k$ times. Consider the first $\card{\states}+1$ iterations of $u$. Let
the positions where the iterations of $u$ starts be $m_1,m_2,\cdots,m_{\card{\states}+1}$. By pigeon-hole principle, there exists some states
$\astate\in\states$ such that for some $i<j\in[1,\card{\states}+1]$, $\arun(m_i)=\arun(m_j)=\astate$. Let $\alpha_1=j-i+1$. We consider $\card{\states}+1$
iterations of $u$ after $m_j$. We proceed as before to obtain $\alpha_2$ and so on. Since $u$ is repeated $2.\card{\states}^k$ times, we will obtain at
least $\card{\states}^{k-1}$ (possibly different) values as $\alpha_1,\alpha_2,\cdots,\alpha_{\card{\states}^{k-1}}\in[1,\card{\states}]$ because
$\card{\states}^{k-1}\times(\card{\states}+1)\leq 2.\card{\states}^{k}$. Again, by pigeon-hole principle, we know that there exists
$j_1,j_2,\ldots,j_{\card{\states}^{k-2}}\in[1,\card{\states}^{k-1}]$ such that $\alpha_{j_1}=\alpha_{j_2}=\ldots=\alpha_{j_{\card{\states}^{k-2}}}=K$ for some
$K\in[1,\card{\states}]$ because $K\times\card{\states}^{k-2}\leq \card{\states}^{k-1}$. Note that for each such $\alpha_j$,
$j\in\{j_1,j_2,\ldots,j_{\card{\states}^{k-2}}\}$, we have a corresponding different loop structure in $\arun$ where we have positions $a$ and $b$ in
$\aword$ such that $\aword[a,b]=(u)^{K}$ and $\arun(a)=\arun(b)=\astate$ for some $\astate\in\states$ as shown in Figure~\ref{fig:run}. Hence, the run
$\arun(1)\ldots\arun(a).\arun(b)\ldots$ is still an accepting run in $\bauto$ for $\aword_1.u^{2.\card{\states}^k-(K)}.\aword_2$. Since, there are
$\card{\states}^{k-2}$ such loops, it is easy to see that for every $N\in[1,\card{\states}^{k-2}]$, we can remove the loops corresponding to
$\alpha_{j_1},\alpha_{j_2},\ldots,\alpha_{j_{N}}$ and have an accepting run for the word $\aword_1.u^{2.\card{\states}^k-(K\times N)}.\aword_2$ in
$\bauto$.  \qed
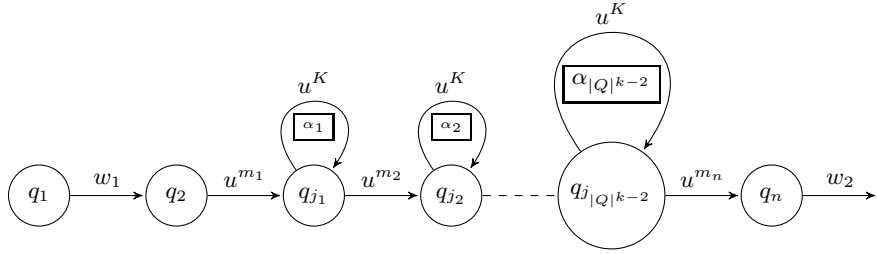
\begin{figure}
  \begin{center}
    \begin{tikzpicture}[shorten >=1pt,node distance=1cm,auto,>=stealth']

  \node[state] (q_1)                  {$\astate_1$};
  \node[state]         (q_2)  [right=of q_1]  {$\astate_{2}$};
  \node[state]         (q_3)  [right=of q_2]  {$\astate_{j_1}$};
  \node (a3)[above=0.2cm of q_3]  {\fbox{{\bf \tiny  $\alpha_1$}}};
  \node[state]         (q_4)  [right=of q_3]  {$\astate_{j_2}$};
  \node (a4)[above=0.2cm of q_4]  {\fbox{{\bf \tiny  $\alpha_2$}}};
  \node[state]         (q_6)  [right=of q_4]  {$\astate_{j_{|\states|^{k-2}}}$};
  \node (a5)[above=0.4cm of q_6]  {\fbox{{\bf \small $\alpha_{|\states|^{k-2}}$}}};
  \node[state]         (q_n)  [right=of q_6]  {$\astate_n$};
  \node                (q_f)  [right=of q_n]  {};
  \path[->] (q_1)   edge              node  {$\aword_1$} (q_2)
            (q_2)   edge              node  {$u^{m_1}$} (q_3)
            (q_3)   edge              node  {$u^{m_2}$} (q_4)
	    	    edge [out=125,in=55,loop] node  {$u^K$} ()
            (q_4)   edge [out=125,in=55,loop] node  {$u^K$} ()
            (q_6)   edge              node  {$u^{m_{n}}$} (q_n)
	    	    edge [out=125,in=55,loop] node  {$u^K$} ()
            (q_n)   edge              node  {$\aword_2$} (q_f);
  \draw[dashed] (5.9,0) -- (6.9,0);
\end{tikzpicture}
  \end{center}
  \caption{Shape of a sample run}
  \label{fig:run}
\end{figure}
\end{proof}

\fi

\begin{theorem}
\label{thm:stuttering-buchi}
Given a constrained path schema $\acps=\pair{\cpschema}{\aconstraintsystem}$, and a \buchi~automaton $\bauto$,
$\alang(\acps)\cap\alang(\bauto)\neq \emptyset$ iff there exists $\vect{y}\in[0,\powerset{\mathtt{pol}_1(x)}+2.|\bauto|^{k}\times
  \powerset{\mathtt{pol}_1(x)}\times \powerset{\mathtt{pol}_2(x)}]^{k-1}$ such that
$\aseg_1(\aloop_1)^{\vect{y}[1]}\ldots\aseg_{k-1} \linebreak[0] (\aloop_{k-1})^{\vect{y}[k-1]}\aseg_k\aloop_k^{\omega}\in\alang(\bauto)\cap\alang(\acps)$.
\end{theorem}

\ifconference
\else
\begin{proof} [Theorem~\ref{thm:stuttering-buchi}]
Since $\alang(\acps)\cap\alang(\bauto)\neq \emptyset$, there exist an infinite word $\aword$ such that $\aword\in\alang(\bauto)\cap\alang(\acps)$. Let
$\vect{y}\in\nat^{k-1}$ be the vector such that
$\aword=\aseg_1(\aloop_1)^{\vect{y}[1]}\ldots\aseg_{k-1}\linebreak[0](\aloop_{k-1})^{\vect{y}[k-1]}\aseg_k\allowbreak\aloop_k^{\omega}$. We will now prove that either
$\vect{y}\in[0,\powerset{\mathit{pol}_1(\size{\acps})}+2.\size{\bauto}^{k}\times \powerset{\mathtt{pol}_1(\size{\acps})+\mathtt{pol}_2(\size{\acps})}]^{k-1}$ or we can
construct another word 
$$\aword'=\aseg_1(\aloop_1)^{\vect{y}'[1]}\ldots\aseg_{k-1}(\aloop_{k-1})^{\vect{y}'[k-1]}\linebreak[0]
\aseg_k\aloop_k^{\omega}$$ 
such that
$\vect{y}'\in[0,\powerset{\mathit{pol}_1(\size{\acps})}+2.\size{\bauto}^{k}\times \powerset{\mathtt{pol}_1(\size{\acps})+\mathtt{pol}_2(\size{\acps})}]^{k-1}$ and
$\aword'\in\alang(\acps)\cap\alang(\bauto)$. Since, $\vect{y}\models\aconstraintsystem$ and $\aconstraintsystem$ is a quantifier-free Presburger
formula, we know that there exist $\abasis,\aperiod_1,\aperiod_2\cdots\aperiod_{\alpha}\in[0,\powerset{\mathtt{pol}_1(\size{\acps})}]^{k-1}$ and $\alpha\leq
\powerset{\mathit{pol}_2(\size{\acps})}$ such that $\vect{y}=\abasis + \Sigma_{i\in[1,\alpha]}a_i.\aperiod_i$ for some $(a_1,a_2,\ldots,a_{\alpha})\in\nat^{\alpha}$. Let us
assume that $\vect{y}\notin[0,\powerset{\mathtt{pol}_1(\size{\acps})}+2.\size{\bauto}^{k}\times \powerset{\mathtt{pol}_1(\size{\acps})+\mathtt{pol}_2(\size{\acps})}]^{k-1}$ and hence there exists some $a_j$, $j\in[1,\alpha]$ such that $a_j>2.\card{\states}^{k}$. We would like to find $a$
such that $\vect{y}'=\abasis +\Sigma_{i\in[1,j-1]}a_i.\aperiod_i + (a_j-a)\aperiod_j + \Sigma_{i\in[j+1,\alpha]}a_i.\aperiod_i$ and
$\aword'=\aseg_1(\aloop_1)^{\vect{y'}[1]}\ldots\aseg_{k-1}\linebreak[0](\aloop_{k-1})^{\vect{y}'[k-1]}\aseg_k\aloop_k^{\omega}\in\alang(\bauto)\cap\alang(\acps)$.
\begin{enumerate}
\itemsep 0 cm 
  \item For any $a\leq a_j$, we have that, $\aword'$ with $\vect{y}'=\abasis +\Sigma_{i\in[1,j-1]}a_i.\aperiod_i + (a_j-a)\linebreak\aperiod_j +
    \Sigma_{i\in[j+1,\alpha]}a_i.\aperiod_i$, $\aword'\in\alang(\acps)$. Indeed by selecting any $a\in[a_j,\linebreak a_j-2.\card{\states}^k]$, we will obtain
    $\vect{y}'\in[0,\powerset{\mathtt{pol}_1(\size{\acps})}+2.\size{\bauto}^{k}\times\linebreak
      \powerset{\mathtt{pol}_1(\size{\acps})+\mathtt{pol}_2(\size{\acps})}]^{k-1}$ where $\aword'\in\alang(\acps)$.
  \item For showing that there exists a value for $a$ such that $\vect{y}'\in[0,\powerset{\mathtt{pol}_1(\size{\acps})}+2.\size{\bauto}^{k}\times
    \powerset{\mathtt{pol}_1(\size{\acps})+\mathtt{pol}_2(\size{\acps})}]^{k-1}$ and $\aword'\in\alang(\bauto)$ we will use Lemma\\\ref{lemma-buchi-aux}. For
    each $m\in[1,k-1]$, we take $r_m=(\aloop_m)^{\aperiod_j[m]}$ i.e. $r_m$ is $\aperiod_j[m]$ copies of $\aloop_m$. Note that by our assumption, for
    each $m\in[1,k-1]$ we can factor $\aword$ as $\aword=\aword^m_1.(r_m)^{a_m}.\aword^m_2$ where $a_m\geq 2.\card{\states}^k$. Thus, applying
    Lemma~\ref{lemma-buchi-aux}, we get that there exist $K_m\in[1,\card{\states}]$ such that for any $N_m\in[1,\card{\states}^{k-2}]$,
    $\aword''=\aword^m_1.(r_m)^{a_m-(N_m\times K_m)}.\aword^m_2\in\alang(\bauto)$. For each $m\in[1,k]$ we take $N_m=K_1\times K_2\cdots K_{m-1}\times
    K_{m+1}\cdots K_{k-1}$ which is less than or equal to $\card{\states}^{k-2}$. It is clear that for each $m\in[1,k-1]$, the number of iteration of $r_m$
    we reduce is $N_m\times K_m$ and is same for all $m$, $N_m\times K_m=K_1\times K_2\cdots K_{k-1}$. Combining the result of Lemma~\ref{lemma-buchi-aux}
    for every loop $\aloop_m$, $m\in[1,k-1]$ and taking $a=K_1\times K_2\cdots K_{k-1}$ in $\vect{y}'$ we obtain $\aword'$ such that
    $\aword'\in\alang(\bauto)$. 
\end{enumerate}
We can continue the process to obtain $\vect{y}'\in[0,\powerset{\mathtt{pol}_1(\size{\acps})}+2.\size{\bauto}^{k}\times \powerset{\mathtt{pol}_1(\size{\acps})+\mathtt{pol}_2(\size{\acps})}]^{k-1}$ and $\aword'\in\alang(\bauto)$.  
\qed
\end{proof}

\fi 

Thus, we obtain the exponential map as $\amap_{{\rm
    \auto}}(\bauto,\acps)=\powerset{\mathtt{pol}_1(\size{\acps})}+2.\size{\bauto}^{\size{\acps}}\times\powerset{\mathtt{pol}_1(\size{\acps})+
  \mathit{pol}_2(\size{\acps})}$, where $\mathtt{pol}_1$ and $\mathtt{pol_2}$ are polynomials in accordance with Theorem~\ref{theorem-pottier91}.
%
\subsection{\np~Algorithm for \buchi~automata and BA}
\label{buchi}
The goal of this section is to characterize the complexity of the model-checking problem
over flat counter systems with BA. Thanks to Theorem~\ref{thm:stuttering-buchi},
we have already established the existence of $\amap_{{\rm \auto}}$ for the membership problem
with \buchi~automata and $\BA$.


For some \buchi~automata $\bauto$, checking  $\alang(\acps) \cap \alang(\bauto)\neq\emptyset$  amounts to guess some
$\vect{n} \in \interval{0}{\powerset{\mathtt{pol}_1(\size{\acps})}+2.\size{\bauto}^{\size{\acps}}\times \powerset{\mathtt{pol}_1(\size{\acps})+\mathtt{pol}_2(\size{\acps})}}^{k-1}$ and check
whether
\iftechreport 
$$\aword = \aseg_1 (\aloop_1)^{\vect{n}[1]} \cdots \aseg_{k-1} (\aloop_{k-1})^{\vect{n}[k-1]} \aseg_k (\aloop_k)^{\omega} \in 
\alang(\acps) \cap \alang(\bauto).
$$
\else
$\aword = \aseg_1 (\aloop_1)^{\vect{n}[1]} \cdots \aseg_{k-1} (\aloop_{k-1})^{\vect{n}[k-1]} \aseg_k (\aloop_k)^{\omega} \in 
\alang(\acps) \cap \alang(\bauto).
$
\fi
 Checking whether $\aword \in \alang(\acps)$ can be done in time which is at most polynomial in
$\size{\acps}log(2.\size{\bauto})+2.\mathtt{pol}_1(\size{\acps})+\mathtt{pol}_2(\size{\acps})$ (and therefore in polynomial time in $log(\size{\bauto}) + \size{\acps}$) since
this amounts to verify whether $\vect{n} \models \aformula$.  Checking whether $\aword \in \alang(\bauto)$ can be done in polynomial time in
$\size{\bauto} + \size{\acps}$ using the results from \cite{Markey&Schnoebelen03}. Hence, membership problem with \buchi~automata is in polynomial time.
\begin{theorem}\label{theorem-ba-ptime}
Membership problem with $\BA$  is \ptime.
\end{theorem}
The proof involves using both nice subalphabet property of BA and using the membership algorithm for \buchi~automata. Details can be found in Appendix.

\iftechreport
For deciding the intersection non-emptiness problem with BA,
we guess a polynomial size vector $\vect{n}$ (thanks to
Theorem~\ref{thm:stuttering-buchi}) to obtain a path in $\alang(\acps)$ which then is used as an instance of membership problem. 
 This
\fi
\ifconference
The construction of a polynomial size \buchi~automaton from a specification in BA and a $\acps$,
\fi
 gives us a
\np~procedure for the intersection non-emptiness problem. \np-hardness follows from the fact that there is an easy reduction from solving a
quantifier-free Presburger formula to the intersection problem.
\begin{corollary} \label{corollary-intersection-ba}
The intersection non-emptiness problem with BA is \np-complete. 
\end{corollary}

By  Lemma~\ref{lemma-nice-subalphabet-property},
Lemma~\ref{lemma-complexity-subalphabet} and Corollary~\ref{corollary-intersection-ba},
we get the following result.
\begin{theorem} \label{theorem-ba-np-complete}
$\mcpb{\BA}{\flatcs}$ is \np-complete. 
\end{theorem}

\np-hardness comes from an easy reduction from the satisfiability
problem over propositional logic, similar to
\cite{DemriDharSangnier12}. We encode the satisfiability problem in
the guard of a flat counter system and use a \buchi~automata to
express the reachability of a state.
\subsection{PSPACE Upper Bound Thanks to Nice BA Property}
\label{general}
In the sequel, we will deal with different type of specifications that express $\omega$-regular properties like extended temporal logic, 
alternating \buchi~automata and linear $\mu$ calculus. Since the treatment for all these specifications are same, we will first provide a 
solution for a general specification language $\mathcal{L}$ and then present the procedures for each of the specifications underlining 
the differences for each.

Let us consider a specification language $\mathcal{L}$ built over the set of atomic propositions $\aalphabet$. We also assume the following properties for
instances of $\mathcal{L}$.

A specification language $\aspeclanguage$ has the \defstyle{nice \buchi~ property} iff for every
specification $\aspecification$ from  $\aspeclanguage$, we can build a B\"uchi automaton 
$\bauto_\aspecification$ such that 
\iftechreport
\begin{enumerate}
\itemsep 0 cm 
  \item $\alang(\aspecification)=\alang(\bauto_\aspecification)$, 
  \item each state of $\bauto_{\aspecification}$ is polynomial size, 
  \item it can be checked if a state is initial [resp. accepting] in polynomial space,
  \item the transition relation can be decided in polynomial space too. 
\end{enumerate}
Note that this means that $\bauto_\aspecification$ can be built in exponential time
in the size of $\aspecification$. This property is not new at all and it is often used to
establish that  satisfiability problem is in \pspace \ for a linear-time temporal logic;
indeed non-emptiness of $\bauto_\aspecification$ can be checked then on-the-fly in nondeterministic
polynomial space. 
\else
$\alang(\aspecification)=\alang(\bauto_\aspecification)$, 
each state of $\bauto_{\aspecification}$ is polynomial size, 
it can be checked if a state is initial [resp. accepting] in polynomial space
and the transition relation can be decided in polynomial space too. 
\fi 

\cut{
\begin{property}
\label{prop}
Given an instance $\aspecification$ of the specification language $\aspeclanguage$, we can construct a \buchi~automata $\bauto_\aspecification$ of size at most
$\powerset{\mathtt{pol}_1(\size{\aspecification})}$ for some polynomial $\mathtt{pol}_1$, such that:
\begin{enumerate}
\itemsep 0 cm 
  \item $\alang(\bauto_\aspecification)$ is exactly the set of models for $\aspecification$ ($\alang(\aspecification)=\alang(\bauto_\aspecification)$).
  \item the size of each state of $\bauto_{\aspecification}$ is polynomial in the size of $\aspecification$,
  \item it can be checked if a state is initial in space polynomial in the size of $\aspecification$,
  \item it can be checked if a state is designated in space polynomial in the size of $\aspecification$,
  \item each transition of $\bauto_{\aspecification}$ can be checked in space polynomial in the size of $\aspecification$.
\end{enumerate}
\end{property}
}

The above property and Theorem~\ref{thm:stuttering-buchi}, allows us to have the following lemma
\begin{lemma}
\label{map-ETL}
Let $\acps$ be a constrained path schema and $\aspecification$ be an instance of specification $\aspeclanguage$.
If $\alang(\acps) \cap \alang(\aspecification)$ is non-empty, then there is an $\omega$-word in 
$\alang(\acps) \cap \alang(\aspecification)$ such that each loop is taken at most a number of times
bounded by $\amap_{{\rm \auto}}(\bauto_{\aspecification},\acps)$.
\end{lemma}

Note that this implies that $\amap_{\mathcal{L}}(\aspecification,\acps)=\amap_{{\rm \auto}}(\bauto_{\aspecification},\acps)$.

So, checking whether  $\alang(\acps) \cap \alang(\aspecification)$  is non-empty amounts to guess some
$\vect{n} \in \interval{0}{\amap_{\mathcal{L}}(\aspecification,\acps)}^{k-1}$ and check
whether 
$$\aword = \aseg_1 (\aloop_1)^{\vect{n}[1]} \cdots \aseg_{k-1} (\aloop_{k-1})^{\vect{n}[k-1]} \aseg_k (\aloop_k)^{\omega} \in 
\alang(\acps) \cap \alang(\aspecification).
$$ 
Checking whether $\aword \in \alang(\acps)$ can be done in polynomial time in $log(\amap_{\mathcal{L}}(\aspecification,\acps))$
(and therefore in polynomial time in $\size{\aspecification} + \size{\acps}$) since this amounts to verify whether $\vect{n} \models \aformula$.
Checking whether $\aword \in \alang(\aspecification)$ can be done in exponential space in $\size{\aspecification} + \size{\acps}$ by converting
$\aspecification$ to an exponential size \buchi~automaton, $\bauto_{\aspecification}$ using 
the nice \buchi~ property
and applying the result of
the membership problem of \buchi~automata. Hence, this leads to an exponential space decision procedure
for the intersection non-emptiness problem but it is possible to get down to nondeterministic polynomial space as explained below.

Clearly, to obtain a nondeterministic polynomial space algorithm we can not construct $\bauto_{\aspecification}$ explicitly since its size could be
exponential. Thus, we would like to do an ``on-the-fly'' construction of $\bauto_{\aspecification}$ where we need to store only polynomial size
information at any point of the algorithm. To facilitate this we use the properties of $\bauto_\aspecification$ described in
the definition of the nice \buchi~ property. 

Let $\acps$,  $\aspecification$ and $\vect{n} \in \Nat^{k-1}$ be an instance of the problem.
For $i \in \interval{1}{k-1}$, let $\vect{n}'[i] \egdef \mathtt{min}(\vect{n}[i],\amap_{\mathcal{L}}(\aspecification,\acps))$. 
By  
the nice \buchi~ property
and Theorem~\ref{thm:stuttering-buchi},  the propositions below are equivalent:
\begin{itemize}
\item  $\aseg_1 (\aloop_1)^{\vect{n}[1]} \cdots \aseg_{k-1} (\aloop_{k-1})^{\vect{n}[k-1]} \aseg_k (\aloop_k)^{\omega} \in 
\alang(\aspecification)$,
\item $\aseg_1 (\aloop_1)^{\vect{n}'[1]} \cdots \aseg_{k-1} (\aloop_{k-1})^{\vect{n}'[k-1]} \aseg_k (\aloop_k)^{\omega} \in 
\alang(\bauto_{\aspecification})$.
\end{itemize}
Without any loss of generality, let us assume then that 
$\vect{n} \in \interval{0}{\amap_{\mathcal{L}}(\aspecification,\acps)}^{k-1}$. 

We would actually check that $\aseg_1 (\aloop_1)^{\vect{n}[1]} \cdots \aseg_{k-1} (\aloop_{k-1})^{\vect{n}[k-1]} \aseg_k (\aloop_k)^{\omega} \in
\alang(\bauto_{\aspecification})$. For this we try to find a ``lasso'' structure in \buchi~automaton. Let $\aword=\aseg_1 (\aloop_1)^{\vect{n}[1]}\linebreak[0] \cdots
\aseg_{k-1} (\aloop_{k-1})^{\vect{n}[k-1]} \aseg_k (\aloop_k)^{\omega}$ and
$\bauto_{\aspecification}=\tuple{\states,\aalphabet,\edges,\states^i,F}$. Since $F$ is finite and in an accepting run we should see infinitly many times
a state from $F$, at least one $\astate_f\in F$ appears infinitely many times in the accepting run of $\aword$ if
$\aword\in\bauto_{\aspecification}$. Again, only $\aloop_k$ is taken infinitly many times in $\aword$ and so there exists at least one position $j$ in
$\aloop_k$ such that the transition from $\astate_f$ reading $\aloop_k(j)$ is taken infinitely many times. For this reason, we try to find a run in
$\bauto_{\aspecification}$ such that $\aseg_1 (\aloop_1)^{\vect{n}[1]} \cdots \aseg_{k-1}(\aloop_{k-1})^{\vect{n}[k-1]} \aseg_k \aloop_k[1,j]$ has a run
from an initial state to $\astate_f$ and $\aloop_k[j+1,\length{\aloop_k}](\aloop_k)^*\aloop_k[1,j]$ has a run from $\astate_f$ to $\astate_f$, this
is done in the following algorithm: 
\begin{algorithm}
\caption{Inputs $\aspecification$, $\acps=\pair{\cpschema}{\aconstraintsystem}$, $\vect{y}\in[0,\amap_{\mathcal{L}}(\aspecification,\acps)]^{k-1}$}
\label{algo:gen-pspace}
\begin{algorithmic}[1]
  \STATE Let $\bauto_{\aspecification}=\tuple{\states,\aalphabet,\edges,\astate_i,F}$
  \FOR{each $\astate_f\in F$ and each $j\in[1,\length{\aloop_k}]$}
     \STATE Construct $\fsa_1=\tuple{\states,\aalphabet,\astate_0,\edges,\{\astate_f\}}$ and
         $\fsa_2=\tuple{\states,\aalphabet,\astate_f,\edges,\{\astate_f\}}$.  
     \STATE Check $\aseg_1(\aloop_1)^{\vect{y}[1]}\ldots\aseg_{k-1}(\aloop_{k-1})^{\vect{y}[k-1]}\aseg_k\aloop_k[1,j]\in\alang(\fsa_1)$
     \STATE Check $\alang(\fsa_2)\cap\alang(\aloop_k[j+1,\length{\aloop_k}]\aloop_k^{*}\aloop_k[1,j])\neq\emptyset$
     \IF {both the checks return true}
        \STATE return TRUE.
     \ENDIF
  \ENDFOR
  \STATE return FALSE.
\end{algorithmic}
\end{algorithm}

\iftechreport
\begin{lemma}
Algorithm~\ref{algo:gen-pspace} returns true iff
$\aword=\aseg_1(\aloop_1)^{\vect{y}[1]}\ldots\aseg_{k-1}(\aloop_{k-1})^{\vect{y}[k-1]}\aseg_k\aloop_k^{\omega}\in\alang(\aspecification)$.
\end{lemma}
\begin{proof}
First observe that by the nice \buchi~ property,
we know that $\aword\in\alang(\aspecification)$ iff $\aword\in\bauto_{\aspecification}$. Thus, in the rest we
will prove that Algorithm~\ref{algo:gen-pspace} returns true iff $\aword\in\alang(\bauto_{\aspecification})$. 

First let us assume that
$\aword=\aseg_1(\aloop_1)^{\vect{y}[1]}\ldots\aseg_{k-1}(\aloop_{k-1})^{\vect{y}[k-1]}\aseg_k\aloop_k^{\omega}\in\alang(\bauto_{\aspecification})$. Thus,
there is an accepting run $\arun\in\edges^{\omega}$ for $\aword$ in $\bauto_{\aspecification}$. According to the \buchi~acceptance condition there exists
a state $\astate_f\in F$ which is visited infinitely often. In $\aword$ only $\aloop_k$ is taken infinitely many times. Thus, $\aloop_k$ being of
finite size, there exists a position $j\in[1,\length{\aloop_k}]$ such that transitions of the form $\astate\labtrans{\aloop_k(j)}\astate_f$ for some
$\astate\in\states$ occurs infinitely many times in $\arun$. Thus, for $\arun$ to be an accepting run, there exists
$\aword'=\alang(\aseg_1(\aloop_1)^{\vect{y}[1]}\ldots\aseg_{k-1}(\aloop_{k-1})^{\vect{y}[k-1]}\aseg_k\aloop_k[1,j])$, which has a run in $\bauto_{\aspecification}$ from
$\astate_i$ to $\astate_f$ and there must exists words $\aword''\in\alang(\aloop_k[j+1,\length{\aloop_k}]\aloop_k^{*}\aloop_k[1,j])$ which has a run
from $\astate_f$ to $\astate_f$. In effect, $\aword'\in\alang(\fsa_1)$ and $\alang(\fsa_{2})\cap\aword''\neq\emptyset$. Thus, there
exists at least one choice of $\astate_f$ and $j$, for which both the checks return true and hence the algorithm returns true.

Now let us assume that the algorithm returns true. Thus, there exists $\astate_f\in F$ and $j\in[1,\length{\aloop_k}]$ such that
$\aword'\in\alang(\fsa_1)$ and $\alang(\fsa_{2})\cap\aword''\neq\emptyset$.
Since, $\alang(\fsa_{2})\cap\aword''\neq\emptyset$, there exists a word
$\aword_1=\aloop_k[j+1,\length{\aloop_k}]\aloop_k^{n}\aloop_k[1,j]\in\alang(\fsa_{loop})$ for some $n$. Consider the word
$\aword_2=\aword'.(\aword_1)^{\omega}$. Clearly, $\aword_2\in\alang(\acps)$. Also, note that $\aword'$ has a run $\arun'$ starting from $\astate_i$ to
$\astate_f$ and $\aword_1$ has a run $\arun''$ starting from $\astate_f$ to $\astate_f$. Thus, we can compose $\arun=\arun'\arun''^{\omega}$ to obtain
a run for $\aword_2$. Note that $\arun$ follows transitions from $\edges$ and visits $\astate_f\in F$ infinitely many times. Thus,
$\aword=\aseg_1(\aloop_1)^{\vect{y}[1]}\ldots\aseg_{k-1}(\aloop_{k-1})^{\vect{y}[k-1]}\aseg_k\aloop_k^{\omega}\in\alang(\bauto_{\aspecification})$.  
\qed
\end{proof}
\else
The proof of correctness of the algorithm can be found in the Appendix.
\fi

\begin{lemma}
Algorithm~\ref{algo:gen-pspace} runs in nondeterministic polynomial space.
\end{lemma}
\begin{proof}
The proof is standard and we recall standard arguments in the presence
of the nice \buchi~ property.
Indeed, by the nice \buchi~ property, 
$\bauto_{\aspecification}$ can be of exponential size in the size of the input, so we cannot  
construct the transition relation of
$\bauto_{\aspecification}$ explicitly, instead we do it on-the-fly.
Let us consider each step of the algorithm and show that each iteration of the {\bf for} loop 
can be done in polynomial space.
\begin{enumerate}
\itemsep 0 cm 
  \item $\fsa_1$ and $\fsa_2$ are essentially copies of $\bauto_{\aspecification}$ and hence their transition 
  relations are also not constructed
    explicitly. But, 
    by the nice \buchi~ property, 
    their states can be represented in polynomial space.
  \item Checking $\aword'\in\alang(\fsa_1)$ can be done by simulating $\fsa_1$ on $\aword'$. Note that for simulating $\fsa_1$, at any position we
    only need to store the previous state and the letter 
    at current position to obtain the next state of $\fsa_1$. Thus, this can be performed in
    polynomial space in $\size{\aspecification}$ and $\size{\acps}+\size{\vect{y}}$.
  \item Checking $\alang(\fsa_2)\cap\alang(\aword'')\neq\emptyset$ can be done by constructing a finite state automata $\fsa_{loop}$ for
    $\alang(\aword'')$ and checking for reachability of final state in the automaton 
     $\fsa_2\times\fsa_{loop}$. Note that $\size{\fsa_{loop}}$ is
    polynomial, but since $\size{\fsa_2}$ is can be of exponential size, 
    $\size{\fsa_2\times\fsa_{loop}}$ can be also of exponential magnitude. 
    However, the graph accessibility problem (GAP) is in \nlogspace, so
    $\alang(\fsa_2)\cap\alang(\aword'')\neq\emptyset$ can also be done in nondeterministic 
    polynomial space.
\end{enumerate}
Thus, the whole procedure can be completed in nondeterministic polynomial space.
\qed
\end{proof}

\begin{lemma} 
\label{lemma-membership-gen}
Let $\aspeclanguage$ be a specification language satisfying the nice \buchi~ property.
Membership problem for $\aspeclanguage$ is in \pspace.
\end{lemma}
\begin{corollary} \label{corollary-intersection-gen}
Let $\aspeclanguage$ be a specification language satisfying the nice \buchi~ property. 
Then, the intersection non-emptiness problem for $\aspeclanguage$ is in \pspace.
\end{corollary}
\iftechreport
Let us recall consequences of results from the literature.
\begin{itemize}
\itemsep 0 cm
\item ETL has the nice \buchi~ property by~\cite{Vardi&Wolper94}.
\item Linear $\mu$-calculus has the nice \buchi~ property by~\cite{Vardi88}.
      Note that the construction in~\cite{Vardi88} uses a (one-way) B\"uchi automaton 
     of exponential size to check for the satisfaction of the formula in the model whereas two B\"uchi automata  check the 
     non-foundedness of the model. A close inspection reveals
     that all of these steps can be done  on-the-fly and in fact this is used in~\cite{Vardi88} 
     to obtain the \pspace \ upper bound for model-checking Kripke structures with $\lmc$~\cite[Corollary 4.6]{Vardi88}. 
\item ABA has the nice \buchi~ property by~\cite{MiyanoHayashi84}.
\end{itemize}
The results for ETL and ABA can be also obtained thanks to existing translations into  linear $\mu$-calculus.
\else
Let us recall consequences of results from the literature.
ETL has the nice BA property by~\cite{Vardi&Wolper94}, 
linear $\mu$-calculus has the nice BA property by~\cite{Vardi88}
and  ABA has the nice BA property by~\cite{MiyanoHayashi84}.
Note that the results for ETL and ABA can be also obtained thanks to translations into  linear $\mu$-calculus.
\fi 
By Lemma~\ref{lemma-membership-gen}, Lemma~\ref{lemma-complexity-subalphabet} and the above-mentioned results, we
obtain the following results. 
\begin{theorem} \label{theorem-pspace-all}
$\mcpb{\ABA}{\flatcs}$, $\mcpb{\ETL}{\flatcs}$ and $\mcpb{\lmc}{\flatcs}$ are in \pspace.
\end{theorem}
It is also possible to get lower bounds.
\begin{theorem} \label{theorem-hardness}
(I) The intersection non-emptiness problem for $\ABA$ [resp. $\lmc$] is \pspace-hard.
(II) $\mcpb{\ABA}{\flatcs}$ and $\mcpb{\lmc}{\flatcs}$ are \pspace-hard.
\end{theorem}
\ifconference
See Appendix~\ref{section-proof-theorem-hardness}. 
According to the proof of Theorem~\ref{theorem-hardness},
\pspace-hardness already holds true for a fixed Kripke structure, that is actually a simple
path schema. Hence for linear $\mu$-caluclus, there is a complexity gap between model-checking
unconstrained path schemas with two loops (in UP$\cap$co-UP~\cite{Jurdzinski98}) and model-checking rudimentary unconstrained path 
schemas (Kripke structures) made of  
a single loop, which is in contrast to Past LTL, where model-checking unconstrained path schemas with a bounded number of loops 
can be done in polynomial time~\cite[Theorem 9]{DemriDharSangnier12}.
\else

\begin{proof} (II) is  direct consequence of (I).

First, we recall that an \defstyle{alternating finite automaton} is a structure of the form $\aautomaton = (Q,\aalphabet, \transitions, q_0,F)$ 
such that $Q$ and
$\aalphabet$ are finite nonempty sets, $\transitions: Q \times \aalphabet \rightarrow \Boolplus(Q)$ 
is the transition function 
($\Boolplus(Q)$ is the set of positive Boolean formulae built over $Q$), $q_0 \in Q$ and $F \subseteq Q$. The \defstyle{acceptance predicate} 
$Acc \subseteq Q \times
\aalphabet^*$ is defined by induction on the length of the second component so that (1) $\pair{q_f}{\varepsilon} \in Acc$ whenever $q_f \in F$ and (2)
$\pair{q}{\aletter \cdot \aword} \in Acc$ iff $v \models \transitions(q,\aletter)$ where $v$ is the Boolean assignment such that $v(q') = \top$ iff
$\pair{q'}{\aword} \in Acc$.  We write $\alang(\aautomaton)$ to denote the language $\set{\aword \in \aalphabet^*: \pair{q_0}{\aword} \in Acc}$ and
more generally, $\alang(\aautomaton,q) \egdef \set{\aword \in \aalphabet^*: \pair{q}{\aword} \in Acc}$.  It has been shown in~\cite{JancarSawa07} that
checking whether an alternating finite automaton $\aautomaton$ with a singleton alphabet has a non-empty language $\alang(\aautomaton)$
is \pspace-hard.
Without loss of generality, we can assume that ($\star$) $q_0 \not \in F$,
($\star \star$) for every $q_f \in F$, $\transitions(q_f,\aletter) = \perp$
and ($\star \star \star$) for every $q \in Q$,  $\transitions(q,\set{\aletter}) \neq \top$ , assuming that
$\aletter$ is the only letter, and still preserves \pspace-hardness.
Indeed, let $\aautomaton = (Q,\set{\aletter}, \transitions, q_0,F)$ be an alternating finite automaton
and  $\aautomaton' = (Q',\set{\aletter}, \transitions', q_0^{\rm new},\set{q_f^{\rm new}})$ be its variant such that
$Q' = Q \uplus \set{ q_0^{\rm new}, q_f^{\rm new}}$, 
$\transitions'(q_0^{\rm new},\aletter) = q_0$,
$\transitions'(q_f^{\rm new}, \aletter) = \perp$ and for every $q  \in Q$, $\transitions'(q,\aletter)$
is obtained from $\transitions(q,\aletter)$ by simultaneously replacing every
occurrence of $q_f \in F$ by $(q_f \vee q_f^{\rm new})$. In the case,  $\transitions(q,\aletter) = \top$ with $q \in Q$,
$\transitions'(q,\aletter)$ is defined as $q \vee q_f^{\rm new}$. 
It is clear from the construction that $\aautomaton'$ follows the conditions of our assumption
and $\alang(\aautomaton') = \aletter \cdot \alang(\aautomaton)$; whence 
$\alang(\aautomaton)$ is non-empty iff $\alang(\aautomaton')$
is non-empty. 

In order to prove the result for $\ABA$,
it is sufficient to observe that given an alternating finite automaton $\aautomaton$ built over the singleton alphabet
$\set{\aletter}$, one can build in logarithmic space an  alternating B\"uchi automaton $\aautomaton'$ 
over the alphabet $\set{\aletter,\aletterbis}$
such that $\alang(\aautomaton') = \alang(\aautomaton) \cdot \set{\aletterbis}^{\omega}$. 
Roughly speaking, the reduction consists in taking the accepting states of $\aautomaton$ and in
letting them accept $\set{\aletterbis}^{\omega}$ in $\aautomaton'$. \pspace-hardness 
of the intersection non-emptiness problem for $\ABA$ is obtained by noting that 
$\alang(\aautomaton)$ is non-empty iff $\alang(\pair{\aletter^{*} \cdot \aletterbis^{\omega}}{\top}) \cap \alang(\aautomaton') \neq \emptyset$.

Now, let us deal with $\lmc$.
The \pspace-hardness is essentially obtained by reducing  nonemptiness problem for
alternating finite automata with a singleton alphabet (see e.g.~\cite{JancarSawa07})
into the vectorial linear $\mu$-calculus with a fixed simple constrained path schema.
Reduction in polynomial-time into linear $\mu$-calculus is then possible
when formula sizes are measured in terms of numbers of subformulae. 
This is a standard type of reduction (see e.g.~\cite[Section 5.4]{Walukiewicz01}); 
we provide details below
not only to be self-contained but also because we need a limited number of resources:
no greatest fixed point operator (e.g. no negation of least fixed
point operator) and we then use a simple path schema. 
In the sequel, for ease of presentation, we consider this latter class of alternating finite automata
and we present a logarithmic-space reduction into the intersection non-emptiness problem with linear $\mu$-calculus.
More precisely, for every alternating finite automaton $\aautomaton$ built over the
singleton alphabet $\set{ \set{\avarprop}}$, we build a formula $\aformula_{\aautomaton}$ in the linear $\mu$-calculus
 (without $\myprevious$ and the greatest fixed-point operator $\nu$) such that 
 $\alang(\aautomaton)$ is non-empty iff
there is $\set{\avarprop} \cdot \set{\avarprop}^{n_1} \cdot
\emptyset^{\omega}$ in $\alang(\acps)$ with the constrained path schema $\acps = \pair{\set{\avarprop} \cdot \set{\avarprop}^* \cdot \emptyset^{\omega}}{\top}$
and $\set{\avarprop} \cdot \set{\avarprop}^{n_1} \cdot
\emptyset^{\omega} \models \aformula_{\aautomaton}$.
In order to define $\aformula_{\aautomaton}$, we build first an intermediate formula
in the vectorial version of the linear $\mu$-calculus, see e.g. similar developments 
in~\cite[Section 5.4]{Walukiewicz01},
and then we translate it into an equivalent formula in the linear $\mu$-calculus by using the well-known
\defstyle{Beki\v{c}'s Principle}.

Let $\aautomaton = (Q,\set{\set{\avarprop}}, \transitions, q_0,F)$ be
a alternating  finite automaton with a singleton alphabet such that $q_0 \not \in F$,
and for every $q_f \in F$, $\transitions(q_f,\set{\avarprop}) = \perp$. 
We order the states of $Q \setminus F$ with $q_1, \ldots, q_{\alpha}$ such that
$q_1$ is the initial state.

We define the formulae in the vectorial version of linear $\mu$-calculus
$\aformulabis_1^0$, \ldots, $\aformulabis_{\alpha}^0$, 
$\aformulabis_1^1$, \ldots, $\aformulabis_{\alpha-1}^1$,
\ldots,
$\aformulabis_1^i$, \ldots, $\aformulabis_{\alpha-i}^i$,
\ldots,
$\aformulabis_{1}^{\alpha-1}$ and such that 
$ \mu \ \avar_1 \ \cdot \aformulabis_{1}^{\alpha-1}$ belongs to the (standard) linear $\mu$-calculus. 
Such formulae will satisfy the following conditions. 
\begin{itemize}
\itemsep 0 cm
\item[(I)] For all $n \geq 1$, 
      $\set{\avarprop}^n \in \alang(\aautomaton)$ iff
      $\set{\avarprop}^n \cdot \emptyset^{\omega} \models \mu \mytuple{\avar_1}{\avar_{\alpha}}
      \ \mytuple{\aformulabis_1^0}{\aformulabis_{\alpha}^0} \cdot \avar_1$.
\item[(II)] For all $j \in \interval{0}{\alpha-1}$,
      $\mu \mytuple{\avar_1}{\avar_{\alpha-j}}
      \ \mytuple{\aformulabis_1^j}{\aformulabis_{\alpha-j}^j} \cdot \avar_1$ is equivalent to
     $\mu \mytuple{\avar_1}{\avar_{\alpha-j-1}}
      \ \mytuple{\aformulabis_1^{j+1}}{\aformulabis_{\alpha-j-1}^{j+1}} \cdot \avar_1$.
\item[(III)] Consequently, 
       for all $n \geq 1$, 
      $\set{\avarprop}^n \in \alang(\aautomaton)$ iff
      $\set{\avarprop}^n \cdot \emptyset^{\omega} \models  \mu \ \avar_1 \ \cdot \aformulabis_{1}^{\alpha-1}$
      and we pose $\aformula_{\aautomaton} = \mu \ \avar_1 \ \cdot \aformulabis_{1}^{\alpha-1}$.
\end{itemize}

Let us define below the formulae: the substitutions are simple and done  hierarchically. 
\begin{description}
\itemsep 0 cm 
\item[(init)] For every $i \in \interval{1}{\alpha}$, $\aformulabis_i^0$ is obtained from
$\transitions(q_i, \set{\avarprop})$ by substituting each $q_j \in Q \setminus F$ by $\mynext \avar_j$
and each $q_f \in F$ by $\mynext \neg \avarprop$, and then by taking the conjunction with $\avarprop$.
So,  $\aformulabis_i^0$ can be written schematically 
as $\avarprop \wedge \transitions(q_i, \set{\avarprop})[q_j \leftarrow \mynext \avar_j,
q_f \leftarrow \mynext \neg \avarprop]$. 
\item[(ind)] For every $j \in \interval{1}{\alpha-1}$, for every $i \in \interval{1}{\alpha-j}$,
              $\aformulabis_i^{j}$ is obtained from  $\aformulabis_i^{j-1}$ by substituting every occurrence of
            $\avar_{\alpha - j + 1}$ by $\mu \avar_{\alpha -j + 1} \ \aformulabis_{\alpha - j +1}^{j-1}$. 
\end{description}
Note that $\mu \ \avar_1 \ \cdot \aformulabis_{1}^{\alpha-1}$ can be built in logarithmic space in
the size of $\aautomaton$ since formulae are represented as DAGs (their size is the number of subformulae)
and for all $j \in \interval{1}{\alpha-1}$ and $i \in \interval{1}{\alpha-j}$,
$\aformulabis_i^j$ has no free occurrences of $\avar_{\alpha -j+1}$, \ldots, $\avar_{\alpha}$. 

It remains to check that (I)--(III) hold. First, observe that (III) is a direct consequence of
(I) and (II). By \defstyle{Beki\v{c}'s Principle}, see e.g.~\cite[Section 1.4.2]{Arnold&Niwinski01}, 
$\mu \mytuple{\avar_1}{\avar_{j}}
      \ \mytuple{\aformulater_1(\avar_1, \ldots,\avar_j)}{\aformulater_j(\avar_1, \ldots,\avar_j)} \cdot \avar_1$ is equivalent to
     $$\mu \mytuple{\avar_1}{\avar_{j-1}}
      \ \mytuple{\aformulater_1(\avar_1, \ldots,\avar_{j-1}, \aformulater')}{
      \aformulater_{j-1}(\avar_1, \ldots, \avar_{j-1},\aformulater')} \cdot \avar_1$$
where $\aformulater' = \mu \avar_j \ \aformulater_j(\avar_1, \ldots,\avar_j)$. 
Note that the substitution performed to build the formula follows exactly the same principle.
For every $j \in \interval{1}{\alpha-1}$, we obtain 
$\mytuple{\aformulabis_1^{j}}{\aformulabis_{\alpha-j}^{j}}$ by replacing $\avar_{\alpha - j+1}$ by 
$\mu \avar_{\alpha -j+1} \ \aformulabis_{\alpha - j+1}^{j-1}$ in $\mytuple{\aformulabis_1^{j-1}}{\aformulabis_{\alpha-j+1}^{j-1}}$. 
Thus by Beki\v{c}'s Principle,  
      $$\mu \mytuple{\avar_1}{\avar_{\alpha-j}}
      \ \mytuple{\aformulabis_1^j}{\aformulabis_{\alpha-j}^j} \cdot \avar_1 \Leftrightarrow \mu \mytuple{\avar_1}{\avar_{\alpha-j-1}}
      \ \mytuple{\aformulabis_1^{j+1}}{\aformulabis_{\alpha-j-1}^{j+1}}\cdot \avar_1$$
is valid for  all $j \in \interval{0}{\alpha-1}$. It remains to verify that (I) holds true. 

In vectorial linear $\mu$-calculus, formulae with outermost fixed-point operators are of the form
$\mu \mytuple{\avar_1}{\avar_{\beta}} \mytuple{\aformula_1}{\aformula_{\beta}} \cdot \avar_j$ with $j \in \interval{1}{\beta}$. 
Whereas fixed points in linear $\mu$-calculus are considered for monotone functions over
the complete lattice $\pair{\powerset{\Nat}}{\subseteq}$, 
fixed points in vectorial linear 
$\mu$-calculus are considered for monotone functions over
the complete lattice $\pair{(\powerset{\Nat})^{\beta}}{\subseteq}$,
where $\mytuple{\asetbis_1}{\asetbis_{\beta}} \subseteq  \mytuple{\asetbis_1'}{\asetbis_{\beta}'}$
iff for every $i \in \interval{1}{\beta}$, we have $\asetbis_i \subseteq \asetbis'_i$. 
So, the satisfaction relation is defined as follows.
Given a model $\sigma \in (\powerset{\varprop})^{\omega}$, $\sigma, i \models_{\amap} 
\mu \mytuple{\avar_1}{\avar_{\beta}} \mytuple{\aformula_1}{\aformula_{\beta}} \cdot \avar_j$ (assuming that the
variables $\avar_k$ occurs positively in the $\aformula_l$'s) iff
$i \in \asetter_j^{\mu}$ where $\mytuple{\asetter_1^{\mu}}{\asetter_{\beta}^{\mu}}$
is the least fixed point of the monotone function $\mathcal{F}_{\amap, \sigma}: (\powerset{\Nat})^{\beta} \rightarrow (\powerset{\Nat})^{\beta}$
defined by $\mathcal{F}_{\amap, \sigma}(\mathcal{Y}_1, \ldots,\mathcal{Y}_{\beta}) =  \mytuple{\mathcal{Y}_1'}{\mathcal{Y}_{\beta}'}$
where 
$$
\mathcal{Y}_{l}'  \egdef \set{i' \in \Nat:
\sigma, i' \models_{\amap[\avar_1 \leftarrow \mathcal{Y}_1, \ldots,\avar_{\beta} \leftarrow \mathcal{Y}_{\beta}]} \aformula_l
}
$$
It is well-known that the least fixed point $\mytuple{\asetter_1^{\mu}}{\asetter_{\beta}^{\mu}}$
can be obtained by an iterative process:
$\mytuple{\asetter_1^{0}}{\asetter_{\beta}^{0}} \egdef \mytuple{\emptyset}{\emptyset}$,
$\mytuple{\asetter_1^{i+1}}{\asetter_{\beta}^{i+1}} \egdef \mathcal{F}_{\amap, \sigma}(\asetter_1^{i},\ldots,\asetter_{\beta}^{i})$
for all $i \geq 0$ 
and, 
$\mytuple{\asetter_1^{\mu}}{\asetter_{\beta}^{\mu}} = \bigcup_{i} \mytuple{\asetter_1^{i}}{\asetter_{\beta}^{i}}$. 

Let $\sigma_n$ be the model $\set{\avarprop}^{n} \cdot \emptyset^{\omega}$ with $n > 0$, $\amap_{\emptyset}$ be the constant assignment equal to
$\emptyset$ everywhere and $\mathcal{F}_{\amap_{\emptyset}, \sigma_n}$ be the monotone function 
$\mathcal{F}_{\amap_{\emptyset}, \sigma_n}: (\powerset{\Nat})^{\alpha} \rightarrow (\powerset{\Nat})^{\alpha}$
defined from $ \mu \mytuple{\avar_1}{\avar_{\alpha}} \ \mytuple{\aformulabis_1^0}{\aformulabis_{\alpha}^0} \cdot \avar_1$.  

Let us show  by induction that for every $i \in \interval{1}{n}$,
the $i$th iterated tuple $\mytuple{\asetter_1^{i}}{\asetter_{\alpha}^{i}}$ verifies that for every $l \in \interval{1}{\alpha}$, 
 $u \in \asetter_l^{i}$ iff $u \in \interval{n-i}{n-1}$ and $\set{\avarprop}^{n-u} \in
\alang(\aautomaton,q_l)$. 

\noindent
{\em Base Case}: $i = 1$. The propositions below are equivalent ($l \in \interval{1}{\alpha}$):
\begin{itemize}
\itemsep 0 cm
\item $u \in  \asetter_l^1$, 
\item $\sigma_n, u \models_{\amap_{\emptyset}[\avar_1 \leftarrow \emptyset, \ldots,\avar_{\alpha} \leftarrow \emptyset]}   \aformulabis_l^0$
     (by definition of $\mathcal{F}_{\amap_{\emptyset}, \sigma_n}$),
\item  $\sigma_n, u \models_{\amap_{\emptyset}[\avar_1 \leftarrow \emptyset, \ldots,\avar_{\alpha} \leftarrow \emptyset]}  
        \avarprop \wedge \transitions(q_i, \set{\avarprop})[q_j \leftarrow \mynext \avar_j,
        q_f \leftarrow \mynext \neg \avarprop]$
       (by definition of $\aformulabis_l^0$),
\item  $\sigma_n, u \models 
        \avarprop \wedge \transitions(q_l, \set{\avarprop})[q_j \leftarrow \perp,
        q_f \leftarrow \mynext \neg \avarprop]$
       (by definition of $\models$),
\item $\sigma_n, u \models \avarprop$ and there is a Boolean valuation $v:Q \rightarrow \set{\perp,\top}$
      such that for every $q \in (Q \setminus F)$, we have $v(q) = \perp$ and
      $v \models \transitions(q_l, \set{\avarprop})$,
\item  $\sigma_n, u \models \avarprop$, $\sigma_n, u+1 \models \neg \avarprop$ and 
       $\pair{q_l}{\set{\avarprop}} \in Acc$ (by definition of $Acc$ and by assumption ($\star \star \star$)),
\item  $u = n-1$ and $\set{\avarprop}^{n-u} \in \alang(\aautomaton,q_l)$ (by definition of $\sigma_n$ and 
$\alang(\aautomaton,q_l)$). 
\end{itemize}

Before proving the induction, we observe that we can also show by induction, that 
for all $l \in \interval{1}{\alpha}$ and $i$, $\asetter_l^i \subseteq \interval{n-i}{n-1}$ ($\dag$). 

\noindent
{\em Induction Step}: Now let us assume that for some $i \in \interval{1}{n}$, the $i$th
iterated tuple $\mytuple{\asetter_1^{i}}{\asetter_{\alpha}^{i}}$ verifies that for every $l \in \interval{1}{\alpha}$, 
$u \in \asetter_l^{i}$ iff $u \in \interval{n-i}{n-1}$ and $\set{\avarprop}^{n-u} \in
\alang(\aautomaton,q_l)$. 
We will show that the same holds true for $(i+1)$th iteration $\mytuple{\asetter_1^{i+1}}{\asetter_{\alpha}^{i+1}}$. 
Since $\asetter_l^{i} \subseteq \asetter_l^{i+1}$ (monotonicity), 
for every $u \in \asetter_l^{i+1} \cap \asetter_l^{i}$, we have 
$u \in \interval{n-i-1}{n-1}$ and $\set{\avarprop}^{n-u} \in
\alang(\aautomaton,q_l)$ (since $\interval{n-i}{n-1} \subseteq \interval{n-i-1}{n-1}$). 
Similarly, if  $u \in \interval{n-i}{n-1}$ and $\set{\avarprop}^{n-u} \in
\alang(\aautomaton,q_l)$, then $u \in \asetter_l^{i}$ by induction hypothesis and therefore $u \in \asetter_l^{i+1}$.
Hence, it remains to show that $u \in  (\asetter_l^{i+1} \setminus 
 \asetter_l^{i})$ iff  $u = n-i-1$ and $\set{\avarprop}^{n-u} \in
\alang(\aautomaton,q_l)$ (i.e. $\set{\avarprop}^{i+1} \in
\alang(\aautomaton,q_l)$). 
By ($\dag$), it is sufficient to show that $(n-i-1) \in \asetter_l^{i+1}$ iff 
$\set{\avarprop}^{i+1} \in
\alang(\aautomaton,q_l)$.

The propositions below are equivalent ($l \in \interval{1}{\alpha}$, $i \geq 1$, $n-i-1 \geq 0$):
\begin{itemize}
\itemsep 0 cm
\item $(n-i-1) \in  \asetter_l^{i+1}$, 
\item $\sigma_n, n-i-1 \models_{\amap_{\emptyset}[\avar_1 \leftarrow \asetter_1^{i}, \ldots,\avar_{\alpha} \leftarrow \asetter_1^{i}]}   \aformulabis_l^0$
     (by definition of $\mathcal{F}_{\amap_{\emptyset}, \sigma_n}^{i+1}$),
\item  $\sigma_n,  n-i-1 \models_{\amap_{\emptyset}[\avar_1 \leftarrow \asetter_1^{i}, \ldots,\avar_{\alpha} \leftarrow \asetter_1^{i}]}  
        \avarprop \wedge \transitions(q_i, \set{\avarprop})[q_j \leftarrow \mynext \avar_j,
        q_f \leftarrow \mynext \neg \avarprop]$
       (by definition of $\aformulabis_l^0$),
\item $\sigma_n,  n-i-1 \models \avarprop$ and there is a Boolean valuation $v:Q \rightarrow \set{\perp,\top}$
      such that
      \begin{enumerate}
      \itemsep 0 cm 
      \item for every $q_{l'} \in (Q \setminus F)$, we have $v(q_{l'}) = \top$ iff $n-i \in  \asetter_{l'}^{i}$,
      \item for every $q_f \in F$, $v(q_f) = \perp$,
      \end{enumerate}
       $v \models \transitions(q_l, \set{\avarprop})$ (by definition of $\models$ and $i \geq 1$),
\item there is  $v:Q \rightarrow \set{\perp,\top}$
      such that 
      \begin{enumerate}
      \itemsep 0 cm 
      \item for every $q_{l'} \in (Q \setminus F)$, we have $v(q_{l'}) = \top$ iff 
            $n-i \in \interval{n-i}{n-1}$ and $\pair{q_{l'}}{\set{\avarprop}^{i}} \in Acc$,
      \item for every $q_f \in F$, $v(q_f) = \perp$,
      \end{enumerate}
       and $v \models \transitions(q_l, \set{\avarprop})$ (by induction hypothesis and since $n-i-1 \in \interval{0}{n-1}$),
\item there is  $v:Q \rightarrow \set{\perp,\top}$
      such that 
      \begin{enumerate}
      \itemsep 0 cm 
      \item for every $q_{l'} \in (Q \setminus F)$, we have $v(q_{l'}) = \top$ iff 
            $\pair{q_{l'}}{\set{\avarprop}^{i}} \in Acc$,
      \item for every $q_f \in F$, $v(q_f) = \perp$ 
      \end{enumerate}
       and $v \models \transitions(q_l, \set{\avarprop})$ (by propositional reasoning),
\item there is  $v:Q \rightarrow \set{\perp,\top}$
      such that 
      \begin{enumerate}
      \itemsep 0 cm 
      \item for every $q_{l'} \in (Q \setminus F)$, we have $v(q_{l'}) = \top$ iff 
            $\pair{q_{l'}}{\set{\avarprop}^{i}} \in Acc$,
      \item for every $q_f \in F$, $v(q_f) = \top$  iff $\pair{q_f}{\set{\avarprop}^{i}} \in Acc$,
      \end{enumerate}
       and $v \models \transitions(q_l, \set{\avarprop})$ (since $i \geq 1$, $\transitions(q_f, \set{\avarprop}) = \perp$ 
       and $\pair{q_f}{\set{\avarprop}^{i}} \not \in Acc$),
\item $\pair{q_l}{\set{\avarprop}^{i+1}} \in Acc$ (by definition of $Acc$),
\item $\set{\avarprop}^{i+1} \in \alang(\aautomaton,q_l)$. 
\end{itemize}

Thus, for every $i \in \interval{1}{n}$,
the $i$th iterated tuple $\mytuple{\asetter_1^{i}}{\asetter_{\alpha}^{i}}$ verifies that for every $l \in \interval{1}{\alpha}$,  
$u \in \asetter_l^{i}$ iff $u \in \interval{n-i}{n-1}$ and $\set{\avarprop}^{n-u} \in
\alang(\aautomaton,q_l)$.
So, $\set{\avarprop}^{n} \in \alang(\aautomaton)$ iff $0 \in \asetter_1^{n}$.
Since $\mytuple{\asetter_1^{\mu}}{\asetter_{\alpha}^{\mu}}$ is precisely equal 
to  $\mytuple{\asetter_1^{n}}{\asetter_{\alpha}^{n}}$ because of the simple structure of
 $\sigma_n$ (see ($\dag$), we conclude that
$\set{\avarprop}^{n} \in \alang(\aautomaton)$ iff
$\sigma_n, 0 \models  \mu \mytuple{\avar_1}{\avar_{\alpha}}
      \ \mytuple{\aformulabis_1^0}{\aformulabis_{\alpha}^0} \cdot \avar_1$, whence (I) holds. 

From (III), we conclude that $\alang(\aautomaton)$ is non-empty iff
there is $\set{\avarprop} \cdot \set{\avarprop}^{n_1} \cdot
\emptyset^{\omega}$ in $\alang(\acps)$ with $\acps = \pair{\set{\avarprop} \cdot \set{\avarprop}^* \cdot \emptyset^{\omega}}{\top}$
such that $\set{\avarprop} \cdot \set{\avarprop}^{n_1} \cdot
\emptyset^{\omega} \models \aformula_{\aautomaton}$. Since $\acps$ and $\aformula_{\aautomaton}$ can be computed in logarithmic space in
the size of $\aautomaton$, this provides a reduction from the nonemptiness problem for  alternating
finite automata with a singleton
alphabet to the intersection non-emptiness problem with linear $\mu$-calculus. 
Hence, the intersection non-emptiness problem is \pspace-hard (we use only  
 a fixed constrained 
path schema and a formula without past-time operators and without greatest fixed-point operator).
\qed
\end{proof}

According to the proof of Theorem~\ref{theorem-hardness},
\pspace-hardness already holds true for a fixed Kripke structure, that is actually a simple
path schema. This contrasts with the fact that model-checking loops with linear $\mu$-calculus
is equivalent to model-checking finite graphs with modal $\mu$-calculus~\cite{Markey&Schnoebelen06}, as far as 
worst-case complexity is concerned (in UP$\cap$co-UP~\cite{Jurdzinski98}).
Hence, for linear $\mu$-calculus, there is already a complexity gap between model-checking
unconstrained path schemas with two loops and model-checking rudimentary unconstrained path 
schemas (Kripke structures) made of  
a single loop. This illustrates also an interesting difference with Past LTL for which
model-checking unconstrained path schemas  (from Kripke structures) with a bounded number of loops 
can be done in polynomial time~\cite[Theorem 9]{DemriDharSangnier12}. 
\fi

As an additional corollary, we can solve the global model-checking problem with existential
Presburger formulae. We knew that Presburger formulae exist for global model-checking~\cite{demri-model-10} 
for Past LTL (and therefore for FO) but we can conclude that they are structurally simple and we provide an 
alternative proof. Moreover, the question has been open for $\lmc$  since the decidability status of 
$\mcpb{\lmc}{\flatcs}$ has been only resolved in the present work.

\begin{corollary}  \label{corollary-global-model-checking}
Let $\aspeclanguage$ be a specification language among FO, BA, ABA, ETL or $\lmc$.
Given a flat counter system $\asys$, a control state $\astate$ and a specification
$\aspecification$ in  $\mathcal{L}$,  one can effectively build an existential Presburger formula
 $\aformula(\avariableter_1, \ldots,\avariableter_n)$
that represents the initial counter values $\avect$ such that there is
an infinite run $\arun$ starting at $\pair{\astate}{\vec{\avect}}$ and
$\arun \models \aspecification$.
\end{corollary}

\cut{
\subsection{Extended Temporal Logic}

Extended temporal logics(ETL) are temporal logics where we use automata connectives in the formula. ETL formulas as introduced in
\cite{Wolper&Vardi&Sistla83,Piterman00}, is defined in following way,
$$
\begin{array}{lcl}
\aformula & ::= & a \mid~ \neg \aformula~ \mid~ \aformula \wedge \aformula'~ \mid~ \fsa(\aformula_1,\aformula_2,\cdots,\aformula_n)
\end{array}
$$ where $p\in \varprop$ and $\fsa=\tuple{\states,\aalphabet,\edges,\states^i,F}$ is a nondeterministic finite automaton. A run $\arun$ satisfies the
ETL formula $\fsa(\aformula_1,\aformula_2,\cdots,\aformula_n)$ at position $i$ iff there exists a word $\aword\in\alang(\fsa)$ over the alphabet set
$\{b_1,b_2,\cdots,b_n\}$ such that for each position $j\in[1,\length{\aword}]$, if $\aword[j]=b_k$ then $\arun,i+j\models\aformula_k$.
We will now characterize the complexity of the model-checking problem over flat counter systems with ETL formulas. By following the development in
\cite{Wolper&Vardi&Sistla83}, we obtain the following theorem just like the specification language $M$ considered in Section~\ref{general}.
\begin{theorem}\cite{Wolper&Vardi&Sistla83}
\label{thm:ETL}
ETL formulas satisfy Property~\ref{prop}
\end{theorem}
This allows us to apply directly the result from Section~\ref{general} for ETL. Thus, applying Algorithm~\ref{algo:gen-pspace} and using
Theorem~\ref{thm:ETL} and Lemma~\ref{lemma-membership-gen}, we get the following ,

\begin{lemma} 
Membership problem with ETL is in \pspace and it can be solved in polynomial space in $\size{\acps} + \size{\aformulabis}$.
\end{lemma}

\begin{corollary} \label{corollary-intersection-ETL}
The intersection non-emptiness problem with ETL is in \pspace. 
\end{corollary}  

Similarly, following the reasoning and construction in Section~\ref{section-pspace-fo}, we know that ETL satisfies the nice sub-alphabet
property. Applying Lemma~\ref{lemma-complexity-subalphabet} and Corollary~\ref{corollary-intersection-ETL}, we can deduce the following.
\begin{theorem}  \label{theorem-etl-inpspace-complete}
MC(ETL, $\flatcs$) is in \pspace. 
\end{theorem}


\subsection{Alternating \buchi~Automata}

Alternating \buchi~automata generalises the normal \buchi~automata by extending the transitions relation to include conjunction alongwith
disjunction. We use the definition provided in ~\cite{MiyanoHayashi84}. Also, using the construction provided in ~\cite{MiyanoHayashi84}, allows us to have a property of ABA similar to Property~\ref{prop} for the specification $M$ in Section~\ref{general}. Thus, using Corollary~\ref{corollary-intersection-gen} we obtain the following lemma,

\begin{lemma} \label{corollary-intersection-ABA}
The intersection non-emptiness problem with ABA is in \pspace. 
\end{lemma}  

We will now prove the \pspace-hardness of the intersection problem for ABA. \pspace-hardness comes from a reduction from the emptiness problem of
1-letter alternating finite automaton to the membership problem of ABA. Alternating finite automata are defined similarly as the alternating
\buchi~automata except for the acceptance condition where each branch is finite and ends with a final state. We use the following theorem.

\begin{theorem}\cite{JancarSawa07}
Checking the non-emptyness of an alternating finite automata over a singleton alphabet is \pspace-hard.
\end{theorem}
We wil construct in \logspace, an alternating \buchi~automata from a given alternating finite automata over single letter alphabet. Let
$\fsa_f=\tuple{\states,\aalphabet,\edges,\astate_i,F}$ with $\aalphabet=\{a\}$ be an alternating finite automata. We construct an alternating
\buchi~automata $\fsa_B=\tuple{\states',\aalphabet',\edges',\astate_i,F'}$ where $\states'=\states\cup\{\astate_b\}$,
$\aalphabet'=\aalphabet\cup\{b\}$ and $F'=\{\astate_b\}$. We also define an additional set of transitions
$\anedge=\{(\astate_b,b,\astate_b\}\cup\{(\astate,b,\astate_b)|\astate\in F\}$ and $\edges'=\edges\cup\anedge$. Intuitively, we add a final state
$\astate_b$ and let the automata take the transition with $b$ to this state from the last final state in each accepting branch. Once, in $\astate_b$,
it can only loop on $b$. Thus, from construction clearly, $\alang(\tuple{a^*b^{\omega},\true})\cap\alang(\fsa_B)\neq\emptyset$ iff $\alang(\fsa_f)$ is
non-empty. Thus, the membership problem for ABA is \pspace-hard.

\begin{theorem}
The intersection non-emptiness problem with ABA is \pspace-complete.
\end{theorem}

In case of flat counter systems, an alternating \buchi~automaton is defined over the set of alphabet $at\cup ag_n$ whereas the flat counter system is
labelled by the set of alphabet $\aalphabet=2^{at\cup ag_n}$. For an alternating \buchi~automaton $\fsa$ built over the alphabet $at\cup ag_n$, we
construct another alternating \buchi~automaton $\fsa'$ built over $\aalphabet$. For this, we replace each transition of the form
$\tuple{\astate,a,\anedge}$ from $\fsa$ with the set $\{\tuple{\astate,b,\anedge}|\forall b\in\aalphabet, a\in b\}$ in $\fsa'$.  Similarly, following
the reasoning in Section~\ref{buchi} for \buchi~automata, we know that ABA has nice sub-alphabet property. 
{\bf Ste: for each specification language, we need to explain explicitly how are encoded the specifications.
I am not sure the nice subalphabet property for ABA is that obvious. Please elaborate. How do you
code alternating automata?.}
Combining Lemma~\ref{lemma-complexity-subalphabet} and
Corollary~\ref{corollary-intersection-ba} we have the following.
\begin{theorem} \label{theorem-aba-pspace-complete}
MC(ABA, $\flatcs$) is in \pspace-Complete. 
\end{theorem}

The \pspace-hardness comes from the fact that intersection non-emptiness problem with ABA is \pspace-hard.

\subsection{Linear $\mu$-calculus}
For defining linear $\mu$-calculus, we assume a countably infinite set of variables $\varset$. Variables will be typically denoted by
$\avar_1,\avar_2,\ldots$. We also assume the same set of propositional symbols denoted as $\varprop$.  A formula in \defstyle{linear $\mu$-calculus ($\lmc$)} is
constructed as following:
$$
\begin{array}{lcl}
\aformula & ::= & a \mid~ \avar_1 \mid~ \neg \aformula~ \mid~ \aformula \wedge \aformula'~ \mid~ \mynext \aformula~ \mid~ \myprevious\aformula ~\mid~
\mu\avar_1.\aformula(\avar_1)
\end{array}
$$ where $p\in \varprop$ and $\avar_1$ occurs positively in $\aformula(\avar_1)$ in the definition of $\mu\avar_1.\aformula(\avar_1)$ and
$\nu\avar_1.\aformula(\avar_1)$.  The operators $\mu\avar_1.\aformula(\avar_1)$ and $\nu\avar_1.\aformula(\avar_1)$ are called \defstyle{fixpoint
  operator}. For defining the satisfaction relation of a $\lmc$ formula over a word $\aword\in\powerset{\varprop}$, we need a labelling function
$\blabelling:\nat\rightarrow\powerset{\varset}$ which returns the set of propositional variables that are true at the specified position.
$$
\begin{array}{rcl}
\aword,i \models_{\blabelling}~ \avarprop & ~\defeq~ & \avarprop \in \aword(i)\\
\aword,i \models_{\blabelling}~ \avar_1 & ~\defeq~ & \avar_1\in\blabelling(i) \\
\aword,i \models_{\blabelling}~ \mynext \aformula & ~\defeq~ & \aword,i+1 \models_{\blabelling}~ \aformula \\
\aword,i \models_{\blabelling}~ \myprevious \aformula & ~\defeq~ & i>0 \mbox{ and } \aword,i-1 \models_{\blabelling}~ \aformula \\
\end{array}
$$

For defining the satisfaction of the \defstyle{fixpoint operator}, we need to define some more notions about the labelling function.  
For a given fixpoint formula $\mu\avar_1.\aformula(\avar_1)$, we define the function $\mathcal{F}_{\amap, \sigma,\aformula}: \powerset{\Nat} \rightarrow \powerset{\Nat}$
defined by $\mathcal{F}_{\amap, \sigma,\aformula}(\mathcal{Y}_1) =  \mathcal{Y}_1'$ where 
$$
\mathcal{Y}_{1}'  \egdef \set{i' \in \Nat:
\sigma, i' \models_{\amap[\avar_1 \leftarrow \mathcal{Y}_1]} \aformula
}
$$
The least fixed point $\asetter_1^{\mu}$
can be obtained by an iterative process:
$\asetter_1^{0} \egdef \emptyset$,
$\asetter_1^{j+1} \egdef \mathcal{F}_{\amap, \sigma,\aformula}(\asetter_1^{j})$.
for all $j \geq 0$ 
and, 
$\asetter_1^{\mu} = \bigcup_{j} \asetter_1^{j}$. 
We say that 
$$
\begin{array}{rcl}
\aword,i \models_{\blabelling}~ \mu\avar_1.\aformula(\avar_1) & ~\defeq~ & i\in\asetter_1^{\mu}\mbox{ where }\asetter_1^{\mu}\mbox{ is the least fixed point of }\mathcal{F}_{\amap, \sigma,\aformula}\\
\end{array}
$$

Note that for a sentence $\aformula$, the satisfaction of $\aformula$ on $\aword$ at position $i$ does not depend on the labelling function
$\blabelling$. Thus, for a sentence $\aformula$ and a word $\aword$ we say, $\aword,0\models\aformula$ if for some $\blabelling$,
$\aword,0\models_{\blabelling}\aformula$. For an ETL formula $\aformulabis$ we denote by $\size{\aformulabis}$ as
the number of subformulas of $\aformulabis$.




\subsection{\pspace~Algorithm for model-checking of $\lmc$ sentences}

The goal of this section is to characterize the complexity of the model-checking problem over flat counter systems with linear $\mu$ calculus. Thanks to \cite{Vardi88} we can transform a given $\lmc$ sentence $\aformula$ to an equivalent \buchi~automaton using the following
theorem,
\begin{theorem}\cite{Vardi88}
\label{thm:vardi}
Given an $\lmc$ sentence $\aformulabis$, we can construct a \buchi~automaton $\bauto_{\aformulabis}$ over the alphabet $\aalphabet=\powerset{\varprop}$,
whose number of states is bounded by $\powerset{19.\size{\aformulabis}^4}$ such that for any $\omega$-word $\aword\in\aalphabet^{\omega}$,
$\aword\in\alang(\bauto_{\aformulabis})$ iff $\aword\models\aformulabis$.
\end{theorem}
Also just like the specification language $M$ considered in Section~\ref{general}, we would like to have a theorem for $\lmc$ corresponding to
Property~\ref{prop}. For this, we use the construction provided in \cite{Vardi88}. The construction uses a one-way \buchi~automata of exponential size
to check for satisfaction of the formula in a model whereas 2 \buchi~automata to check for non-foundedness of the model. A close inspection reveals
that all of these steps can be done ``on-the-fly'' and in fact is used in \cite{Vardi88} to obtain the upperbound of \pspace for the problem
MC($\lmc$,KS). This gives us the following equivalent lemma for Property~\ref{prop} in Section~\ref{general}.

\begin{lemma}\cite{Vardi88}
\label{thm:mu}
Linear $\mu$-caluclus formula satisfy Property~\ref{prop}.
\end{lemma}
This allows us to apply directly the result from Section~\ref{general} for $\lmc$. Thus, applying Algorithm~\ref{algo:gen-pspace} and using
Theorem~\ref{thm:mu}, Theorem~\ref{thm:vardi} and Lemma~\ref{lemma-membership-gen}, we get the following ,

\begin{lemma} \label{lemma-mu-calculus-membership}
Membership problem with linear $\mu$ calculus is in \pspace and it can be solved in polynomial space in $\size{\acps} + \size{\aformulabis}$.
\end{lemma}

\begin{corollary} \label{corollary-intersection-mu}
The intersection non-emptiness problem with linear $\mu$-calculus is in \pspace. 
\end{corollary}  


\newcommand{\aautomaton}{\mathcal{A}}
\newcommand{\transitions}{\delta}
\newcommand{\mytuple}[2]{\langle #1, \ldots,#2 \rangle}

%
%
It remains to establish complexity lower bounds, which we perform in the rest of this section. 

\begin{theorem} \label{theorem-membership-linear-mu-calculus}
The intersection non-emptiness problem with linear $\mu$-calculus is \pspace-complete. 
\end{theorem}

The \pspace-hardness is essentially obtained by reducing  nonemptiness problem for
alternating finite automata with a singleton alphabet (see e.g.~\cite{JancarSawa07})
into the vectorial linear $\mu$-calculus with a fixed simple constrained path schema.
Reduction in polynomial-time into linear $\mu$-calculus is then possible
when formula sizes are measured in terms of numbers of subformulae. 

\begin{proof} 

\cut{
Otherwise, let us assume we have $\aautomaton = (Q,\aalphabet, \transitions, q_0,F)$ such that $q_0 \in F$ and
for every $q_f \in F$, $\transitions(q_f,\set{\avarprop})\neq \perp$. We construct $\aautomaton' = (Q',\aalphabet, \transitions', q'_0,F')$, where
$Q'=Q\cup\{ q'_0\}\cup\{q^1_f | \forall q_f\in F \}$, $F'=\{q^1_f | \forall q_f\in F \}$. We construct transitions in $\transitions'$ as following,
\begin{enumerate}
  \item $\transitions'(\astate'_0,\aletter)=\astate_0$.
  \item for any state $\astate\in Q'$ and $\astate\not\in F'\cup\{\astate'_0\}$, we have
    $\transitions'(\astate,\aletter)=\transitions(\astate,\aletter)[\astate_f\leftarrow\astate_f\vee\astate^1_f]$ (i.e. replacing $\astate_f$ with
    $\astate_f\vee\astate^1_f$ in $\transitions(\astate,\aletter)$) for every $\astate_f\in F$.
  \item for any state $\astate^1_f\in F'$, $\transitions(\astate^1_f,\aletter)=\perp$.
\end{enumerate}
It is clear from the construction that $\aautomaton'$ follows the conditions of our assumption. We will now show that
$\alang(\aautomaton')\{\aletter.\aword|\aword\in\alang(\aautomaton)\}$. Note that for any state $\astate\in Q$, $\pair{\astate}{\aletter}\in Acc$ the
acceptance predicate for $\aautomaton$, iff $\pair{\astate}{\aletter}\in Acc'$ the acceptance predicate of $\aautomaton'$ as every occurance of
$\astate_f\in F$ is replaced by $\astate_f\vee\astate^1_f$ and although $\pair{\astate_f}{\varepsilon}\not\in Acc'$ but
$\pair{\astate^1_f}{\varepsilon}\in Acc'$. Since, $\transitions'(\astate'_0,\aletter)=\astate_0$, by induction we obtain that for every word $\aword$,
$\pair{\astate_0}{\aword}\in Acc$ iff $\pair{\astate'_0}{\aletter.\aword}\in Acc'$. 
In other words, $\alang(\aautomaton)$ is non-empty iff $\alang(\aautomaton')$. Hence, even with the assumptions, the problem
remains \pspace-hard. Also, note that $\size{\aautomaton'}=2.\size{\aautomaton}+1$
}
We recall that an \defstyle{alternating finite automaton} is a tuple $\aautomaton = (Q,\aalphabet, \transitions, q_0,F)$ 
such that $Q$ and
$\aalphabet$ are finite nonempty sets, $\transitions: Q \times \aalphabet \rightarrow Bool^{+}(Q)$ is the transition function 
($Bool^{+}(Q)$ is the set of positive Boolean formulae built over $Q$), $q_0 \in Q$ and $F \subseteq Q$. The \defstyle{acceptance predicate} 
$Acc \subseteq Q \times
\aalphabet^*$ is defined by induction on the length of the second component so that (1) $\pair{q_f}{\varepsilon} \in Acc$ whenever $q_f \in F$ and (2)
$\pair{q}{\aletter \cdot \aword} \in Acc$ iff $v \models \transitions(q,\aletter)$ where $v$ is the Boolean assignment such that $v(q') = \top$ iff
$\pair{q'}{\aword} \in Acc$.  We write $\alang(\aautomaton)$ to denote the language $\set{\aword \in \aalphabet^*: \pair{q_0}{\aword} \in Acc}$ and
more generally, $\alang(\aautomaton,q) \egdef \set{\aword \in \aalphabet^*: \pair{q}{\aword} \in Acc}$.  It has been shown in~\cite{JancarSawa07} that
checking whether an alternating finite automaton $\aautomaton$ with a singleton alphabet has a non-empty language $\alang(\aautomaton)$
is \pspace-hard.
Without loss of generality, we can assume that ($\star$) $q_0 \not \in F$,
($\star \star$) for every $q_f \in F$, $\transitions(q_f,\set{\avarprop}) = \perp$
and ($\star \star \star$) for every $q \in Q$,  $\transitions(q_f,\set{\avarprop}) \neq \top$ , assuming that
$\set{\avarprop}$ is the only letter, which preserves \pspace-hardness.
Indeed, let $\aautomaton = (Q,\set{\aletter}, \transitions, q_0,F)$ be an alternating finite automaton
and  $\aautomaton' = (Q',\set{\aletter}, \transitions', q_0^{\rm new},\set{q_f^{\rm new}})$ be its variant such that
$Q' = Q \uplus \set{ q_0^{\rm new}, q_f^{\rm new}}$, 
$\transitions'(q_0^{\rm new},\aletter) = q_0$,
$\transitions'(q_f^{\rm new}, \aletter) = \perp$ and for every $q  \in Q$, $\transitions'(q,\aletter)$
is obtained from $\transitions(q,\aletter)$ by simultaneously replacing every
occurrence of $q_f \in F$ by $(q_f \vee q_f^{\rm new})$. In the case,  $\transitions(q,\aletter) = \top$ with $q \in Q$,
$\transitions'(q,\aletter)$ is defined as $q \vee q_f^{\rm new}$. 
It is clear from the construction that $\aautomaton'$ follows the conditions of our assumption
and $\alang('\aautomaton') = \aletter \cdot \alang(\aautomaton)$; whence 
$\alang(\aautomaton)$ is non-empty iff $\alang(\aautomaton')$
is non-empty. 
In the sequel, for ease of presentation, we consider this latter class of alternating finite automata
and we present a logarithmic-space reduction into the intersection non-emptiness problem with linear $\mu$-calculus.
More precisely, for every alternating finite automaton $\aautomaton$ built over the
singleton alphabet $\set{ \set{\avarprop}}$, we build a formula $\aformula_{\aautomaton}$ in the linear $\mu$-calculus
 (without $\myprevious$ and the greatest fixed-point operator $\nu$) such that 
 $\alang(\aautomaton)$ is non-empty iff
there is $\set{\avarprop} \cdot \set{\avarprop}^{n_1} \cdot
\emptyset^{\omega}$ in $\alang(\acps)$ with the constrained path schema $\acps = \pair{\set{\avarprop} \cdot \set{\avarprop}^* \cdot \emptyset^{\omega}}{\top}$
and $\set{\avarprop} \cdot \set{\avarprop}^{n_1} \cdot
\emptyset^{\omega} \models \aformula_{\aautomaton}$.
In order to define $\aformula_{\aautomaton}$, we build first an intermediate formula
in the vectorial version of the linear $\mu$-calculus, see e.g. similar developments in~\cite{Walukiewicz01},
and then we translate it into an equivalent formula in the linear $\mu$-calculus by using the well-known
\defstyle{Beki\v{c}'s Principle}.

Let $\aautomaton = (Q,\set{\set{\avarprop}}, \transitions, q_0,F)$ be
a alternating  finite automaton with a singleton alphabet such that $q_0 \not \in F$,
and for every $q_f \in F$, $\transitions(q_f,\set{\avarprop}) = \perp$. 
We order the states of $Q \setminus F$ with $q_1, \ldots, q_{\alpha}$ such that
$q_1$ is the initial state.

We define the formulae in the vectorial version of linear $\mu$-calculus
$\aformulabis_1^0$, \ldots, $\aformulabis_{\alpha}^0$, 
$\aformulabis_1^1$, \ldots, $\aformulabis_{\alpha-1}^1$,
\ldots,
$\aformulabis_1^i$, \ldots, $\aformulabis_{\alpha-i}^i$,
\ldots,
$\aformulabis_{1}^{\alpha-1}$ and such that 
$ \mu \ \avar_1 \ \cdot \aformulabis_{1}^{\alpha-1}$ belongs to the (standard) linear $\mu$-calculus. 
Such formulae will satisfy the following conditions. 
\begin{itemize}
\itemsep 0 cm
\item[(I)] For all $n \geq 1$, 
      $\set{\avarprop}^n \in \alang(\aautomaton)$ iff
      $\set{\avarprop}^n \cdot \emptyset^{\omega} \models \mu \mytuple{\avar_1}{\avar_{\alpha}}
      \ \mytuple{\aformulabis_1^0}{\aformulabis_{\alpha}^0} \cdot \avar_1$.
\item[(II)] For all $j \in \interval{0}{\alpha-1}$,
      $\mu \mytuple{\avar_1}{\avar_{\alpha-j}}
      \ \mytuple{\aformulabis_1^j}{\aformulabis_{\alpha-j}^j} \cdot \avar_1$ is equivalent to
     $\mu \mytuple{\avar_1}{\avar_{\alpha-j-1}}
      \ \mytuple{\aformulabis_1^{j+1}}{\aformulabis_{\alpha-j-1}^{j+1}} \cdot \avar_1$.
\item[(III)] Consequently, 
       for all $n \geq 1$, 
      $\set{\avarprop}^n \in \alang(\aautomaton)$ iff
      $\set{\avarprop}^n \cdot \emptyset^{\omega} \models  \mu \ \avar_1 \ \cdot \aformulabis_{1}^{\alpha-1}$
      and we pose $\aformula_{\aautomaton} = \mu \ \avar_1 \ \cdot \aformulabis_{1}^{\alpha-1}$.
\end{itemize}

Let us define below the formulae: the substitutions are simple and done  hierarchically. 
\begin{description}
\itemsep 0 cm 
\item[(init)] For every $i \in \interval{1}{\alpha}$, $\aformulabis_i^0$ is obtained from
$\transitions(q_i, \set{\avarprop})$ by substituting each $q_j \in Q \setminus F$ by $\mynext \avar_j$
and each $q_f \in F$ by $\mynext \neg \avarprop$, and then by taking the conjunction with $\avarprop$.
So,  $\aformulabis_i^0$ can be written schematically 
as $\avarprop \wedge \transitions(q_i, \set{\avarprop})[q_j \leftarrow \mynext \avar_j,
q_f \leftarrow \mynext \neg \avarprop]$. 
\item[(ind)] For every $j \in \interval{1}{\alpha-1}$, for every $i \in \interval{1}{\alpha-j}$,
              $\aformulabis_i^{j}$ is obtained from  $\aformulabis_i^{j-1}$ by substituting every occurrence of
            $\avar_{\alpha - j + 1}$ by $\mu \avar_{\alpha -j + 1} \ \aformulabis_{\alpha - j +1}^{j-1}$. 
\end{description}
Note that $\mu \ \avar_1 \ \cdot \aformulabis_{1}^{\alpha-1}$ can be built in logarithmic space in
the size of $\aautomaton$ since formulae are represented as DAGs (their size is the number of subformulae)
and for all $j \in \interval{1}{\alpha-1}$ and $i \in \interval{1}{\alpha-j}$,
$\aformulabis_i^j$ has no free occurrences of $\avar_{\alpha -j+1}$, \ldots, $\avar_{\alpha}$. 

It remains to check that (I)--(III) hold. First, observe that (III) is a direct consequence of
(I) and (II). By \defstyle{Beki\v{c}'s Principle}, see e.g.~\cite[Section 1.4.2]{Arnold&Niwinski01}, 
$\mu \mytuple{\avar_1}{\avar_{j}}
      \ \mytuple{\aformulater_1(\avar_1, \ldots,\avar_j)}{\aformulater_j(\avar_1, \ldots,\avar_j)} \cdot \avar_1$ is equivalent to
     $$\mu \mytuple{\avar_1}{\avar_{j-1}}
      \ \mytuple{\aformulater_1(\avar_1, \ldots,\avar_{j-1}, \aformulater')}{
      \aformulater_{j-1}(\avar_1, \ldots, \avar_{j-1},\aformulater')} \cdot \avar_1$$
where $\aformulater' = \mu \avar_j \ \aformulater_j(\avar_1, \ldots,\avar_j)$. 
Note that the substitution performed to build the formula follows exactly the same principle.
For every $j \in \interval{1}{\alpha-1}$, we obtain 
$\mytuple{\aformulabis_1^{j}}{\aformulabis_{\alpha-j}^{j}}$ by replacing $\avar_{\alpha - j+1}$ by 
$\mu \avar_{\alpha -j+1} \ \aformulabis_{\alpha - j+1}^{j-1}$ in $\mytuple{\aformulabis_1^{j-1}}{\aformulabis_{\alpha-j+1}^{j-1}}$. 
Thus by Beki\v{c}'s Principle,  
      $$\mu \mytuple{\avar_1}{\avar_{\alpha-j}}
      \ \mytuple{\aformulabis_1^j}{\aformulabis_{\alpha-j}^j} \cdot \avar_1 \Leftrightarrow \mu \mytuple{\avar_1}{\avar_{\alpha-j-1}}
      \ \mytuple{\aformulabis_1^{j+1}}{\aformulabis_{\alpha-j-1}^{j+1}}\cdot \avar_1$$
is valid for  all $j \in \interval{0}{\alpha-1}$. It remains to verify that (I) holds true. 

In vectorial linear $\mu$-calculus, formulae with outermost fixed-point operators are of the form
$\mu \mytuple{\avar_1}{\avar_{\beta}} \mytuple{\aformula_1}{\aformula_{\beta}} \cdot \avar_j$ with $j \in \interval{1}{\beta}$. 
Whereas fixed points in linear $\mu$-calculus are considered for monotone functions over
the complete lattice $\pair{\powerset{\Nat}}{\subseteq}$, 
fixed points in vectorial linear 
$\mu$-calculus are considered for monotone functions over
the complete lattice $\pair{(\powerset{\Nat})^{\beta}}{\subseteq}$,
where $\mytuple{\asetbis_1}{\asetbis_{\beta}} \subseteq  \mytuple{\asetbis_1'}{\asetbis_{\beta}'}$
iff for every $i \in \interval{1}{\beta}$, we have $\asetbis_i \subseteq \asetbis'_i$. 
So, the satisfaction relation is defined as follows.
Given a model $\sigma \in (\powerset{\varprop})^{\omega}$, $\sigma, i \models_{\amap} 
\mu \mytuple{\avar_1}{\avar_{\beta}} \mytuple{\aformula_1}{\aformula_{\beta}} \cdot \avar_j$ (assuming that the
variables $\avar_k$ occurs positively in the $\aformula_l$'s) iff
$i \in \asetter_j^{\mu}$ where $\mytuple{\asetter_1^{\mu}}{\asetter_{\beta}^{\mu}}$
is the least fixed point of the monotone function $\mathcal{F}_{\amap, \sigma}: (\powerset{\Nat})^{\beta} \rightarrow (\powerset{\Nat})^{\beta}$
defined by $\mathcal{F}_{\amap, \sigma}(\mathcal{Y}_1, \ldots,\mathcal{Y}_{\beta}) =  \mytuple{\mathcal{Y}_1'}{\mathcal{Y}_{\beta}'}$
where 
$$
\mathcal{Y}_{l}'  \egdef \set{i' \in \Nat:
\sigma, i' \models_{\amap[\avar_1 \leftarrow \mathcal{Y}_1, \ldots,\avar_{\beta} \leftarrow \mathcal{Y}_{\beta}]} \aformula_l
}
$$
It is well-known that the least fixed point $\mytuple{\asetter_1^{\mu}}{\asetter_{\beta}^{\mu}}$
can be obtained by an iterative process:
$\mytuple{\asetter_1^{0}}{\asetter_{\beta}^{0}} \egdef \mytuple{\emptyset}{\emptyset}$,
$\mytuple{\asetter_1^{i+1}}{\asetter_{\beta}^{i+1}} \egdef \mathcal{F}_{\amap, \sigma}(\asetter_1^{i},\ldots,\asetter_{\beta}^{i})$
for all $i \geq 0$ 
and, 
$\mytuple{\asetter_1^{\mu}}{\asetter_{\beta}^{\mu}} = \bigcup_{i} \mytuple{\asetter_1^{i}}{\asetter_{\beta}^{i}}$. 

Let $\sigma_n$ be the model $\set{\avarprop}^{n} \cdot \emptyset^{\omega}$ with $n > 0$, $\amap_{\emptyset}$ be the constant assignment equal to
$\emptyset$ everywhere and $\mathcal{F}_{\amap_{\emptyset}, \sigma_n}$ be the monotone function 
$\mathcal{F}_{\amap_{\emptyset}, \sigma_n}: (\powerset{\Nat})^{\alpha} \rightarrow (\powerset{\Nat})^{\alpha}$
defined from $ \mu \mytuple{\avar_1}{\avar_{\alpha}} \ \mytuple{\aformulabis_1^0}{\aformulabis_{\alpha}^0} \cdot \avar_1$.  

Let us show  by induction that for every $i \in \interval{1}{n}$,
the $i$th iterated tuple $\mytuple{\asetter_1^{i}}{\asetter_{\alpha}^{i}}$ verifies that for every $l \in \interval{1}{\alpha}$, 
 $u \in \asetter_l^{i}$ iff $u \in \interval{n-i}{n-1}$ and $\set{\avarprop}^{n-u} \in
\alang(\aautomaton,q_l)$. 

\noindent
{\em Base Case}: $i = 1$. The propositions below are equivalent ($l \in \interval{1}{\alpha}$):
\begin{itemize}
\itemsep 0 cm
\item $u \in  \asetter_l^1$, 
\item $\sigma_n, u \models_{\amap_{\emptyset}[\avar_1 \leftarrow \emptyset, \ldots,\avar_{\alpha} \leftarrow \emptyset]}   \aformulabis_l^0$
     (by definition of $\mathcal{F}_{\amap_{\emptyset}, \sigma_n}$),
\item  $\sigma_n, u \models_{\amap_{\emptyset}[\avar_1 \leftarrow \emptyset, \ldots,\avar_{\alpha} \leftarrow \emptyset]}  
        \avarprop \wedge \transitions(q_i, \set{\avarprop})[q_j \leftarrow \mynext \avar_j,
        q_f \leftarrow \mynext \neg \avarprop]$
       (by definition of $\aformulabis_l^0$),
\item  $\sigma_n, u \models 
        \avarprop \wedge \transitions(q_l, \set{\avarprop})[q_j \leftarrow \perp,
        q_f \leftarrow \mynext \neg \avarprop]$
       (by definition of $\models$),
\item $\sigma_n, u \models \avarprop$ and there is a Boolean valuation $v:Q \rightarrow \set{\perp,\top}$
      such that for every $q \in (Q \setminus F)$, we have $v(q) = \perp$ and
      $v \models \transitions(q_l, \set{\avarprop})$,
\item  $\sigma_n, u \models \avarprop$, $\sigma_n, u+1 \models \neg \avarprop$ and 
       $\pair{q_l}{\set{\avarprop}} \in Acc$ (by definition of $Acc$ and by assumption ($\star \star \star$)),
\item  $u = n-1$ and $\set{\avarprop}^{n-u} \in \alang(\aautomaton,q_l)$ (by definition of $\sigma_n$ and 
$\alang(\aautomaton,q_l)$). 
\end{itemize}
\cut{
Consider the first iterated tuple $\mytuple{\asetter_1^{1}}{\asetter_{\alpha}^{1}}$. Clearly, from the definition of
$\mathcal{F}$ for any $l\in[1,\alpha]$ $\asetter_l^1$ consists of the single position $n-1$ only if $\aformula_l$ is of the form $\avarprop\wedge
Bool^+(\mynext\neg\avarprop)$. By construction, $\aformula_l$ is of the specified form iff $\pair{\astate_l}{\avarprop}\in Acc$. Thus, indeed for all $u
\in \asetter_l^{i}$, $\set{ \set{\avarprop} }^{n-u} \in \alang(\aautomaton,q_l)$.  
}

Before proving the induction, we observe that we can also show by induction, that 
for all $l \in \interval{1}{\alpha}$ and $i$, $\asetter_l^i \subseteq \interval{n-i}{n-1}$ ($\dag$). 

\noindent
{\em Induction Step}: Now let us assume that for some $i \in \interval{1}{n}$, the $i$th
iterated tuple $\mytuple{\asetter_1^{i}}{\asetter_{\alpha}^{i}}$ verifies that for every $l \in \interval{1}{\alpha}$, 
$u \in \asetter_l^{i}$ iff $u \in \interval{n-i}{n-1}$ and $\set{\avarprop}^{n-u} \in
\alang(\aautomaton,q_l)$. 
We will show that the same holds true for $(i+1)$th iteration $\mytuple{\asetter_1^{i+1}}{\asetter_{\alpha}^{i+1}}$. 
Since $\asetter_l^{i} \subseteq \asetter_l^{i+1}$ (monotonicity), 
for every $u \in \asetter_l^{i+1} \cap \asetter_l^{i}$, we have 
$u \in \interval{n-i-1}{n-1}$ and $\set{\avarprop}^{n-u} \in
\alang(\aautomaton,q_l)$ (since $\interval{n-i}{n-1} \subseteq \interval{n-i-1}{n-1}$). 
Similarly, if  $u \in \interval{n-i}{n-1}$ and $\set{\avarprop}^{n-u} \in
\alang(\aautomaton,q_l)$, then $u \in \asetter_l^{i}$ by induction hypothesis and therefore $u \in \asetter_l^{i+1}$.
Hence, it remains to show that $u \in  (\asetter_l^{i+1} \setminus 
 \asetter_l^{i})$ iff  $u = n-i-1$ and $\set{\avarprop}^{n-u} \in
\alang(\aautomaton,q_l)$ (i.e. $\set{\avarprop}^{i+1} \in
\alang(\aautomaton,q_l)$). 
By ($\dag$), it is sufficient to show that $(n-i-1) \in \asetter_l^{i+1}$ iff 
$\set{\avarprop}^{i+1} \in
\alang(\aautomaton,q_l)$.

The propositions below are equivalent ($l \in \interval{1}{\alpha}$, $i \geq 1$, $n-i-1 \geq 0$):
\begin{itemize}
\itemsep 0 cm
\item $(n-i-1) \in  \asetter_l^{i+1}$, 
\item $\sigma_n, n-i-1 \models_{\amap_{\emptyset}[\avar_1 \leftarrow \asetter_1^{i}, \ldots,\avar_{\alpha} \leftarrow \asetter_1^{i}]}   \aformulabis_l^0$
     (by definition of $\mathcal{F}_{\amap_{\emptyset}, \sigma_n}^{i+1}$),
\item  $\sigma_n,  n-i-1 \models_{\amap_{\emptyset}[\avar_1 \leftarrow \asetter_1^{i}, \ldots,\avar_{\alpha} \leftarrow \asetter_1^{i}]}  
        \avarprop \wedge \transitions(q_i, \set{\avarprop})[q_j \leftarrow \mynext \avar_j,
        q_f \leftarrow \mynext \neg \avarprop]$
       (by definition of $\aformulabis_l^0$),
\item $\sigma_n,  n-i-1 \models \avarprop$ and there is a Boolean valuation $v:Q \rightarrow \set{\perp,\top}$
      such that
      \begin{enumerate}
      \itemsep 0 cm 
      \item for every $q_{l'} \in (Q \setminus F)$, we have $v(q_{l'}) = \top$ iff $n-i \in  \asetter_{l'}^{i}$,
      \item for every $q_f \in F$, $v(q_f) = \perp$,
      \end{enumerate}
       $v \models \transitions(q_l, \set{\avarprop})$ (by definition of $\models$ and $i \geq 1$),
\item there is  $v:Q \rightarrow \set{\perp,\top}$
      such that 
      \begin{enumerate}
      \itemsep 0 cm 
      \item for every $q_{l'} \in (Q \setminus F)$, we have $v(q_{l'}) = \top$ iff 
            $n-i \in \interval{n-i}{n-1}$ and $\pair{q_{l'}}{\set{\avarprop}^{i}} \in Acc$,
      \item for every $q_f \in F$, $v(q_f) = \perp$,
      \end{enumerate}
       and $v \models \transitions(q_l, \set{\avarprop})$ (by induction hypothesis and since $n-i-1 \in \interval{0}{n-1}$),
\item there is  $v:Q \rightarrow \set{\perp,\top}$
      such that 
      \begin{enumerate}
      \itemsep 0 cm 
      \item for every $q_{l'} \in (Q \setminus F)$, we have $v(q_{l'}) = \top$ iff 
            $\pair{q_{l'}}{\set{\avarprop}^{i}} \in Acc$,
      \item for every $q_f \in F$, $v(q_f) = \perp$ 
      \end{enumerate}
       and $v \models \transitions(q_l, \set{\avarprop})$ (by propositional reasoning),
\item there is  $v:Q \rightarrow \set{\perp,\top}$
      such that 
      \begin{enumerate}
      \itemsep 0 cm 
      \item for every $q_{l'} \in (Q \setminus F)$, we have $v(q_{l'}) = \top$ iff 
            $\pair{q_{l'}}{\set{\avarprop}^{i}} \in Acc$,
      \item for every $q_f \in F$, $v(q_f) = \top$  iff $\pair{q_f}{\set{\avarprop}^{i}} \in Acc$,
      \end{enumerate}
       and $v \models \transitions(q_l, \set{\avarprop})$ (since $i \geq 1$, $\transitions(q_f, \set{\avarprop}) = \perp$ 
       and $\pair{q_f}{\set{\avarprop}^{i}} \not \in Acc$),
\item $\pair{q_l}{\set{\avarprop}^{i+1}} \in Acc$ (by definition of $Acc$),
\item $\set{\avarprop}^{i+1} \in \alang(\aautomaton,q_l)$. 
\end{itemize}

\cut{
Indeed,
by following the construction of $\aformula_l$, we know that $\aformula_l$ is of the form $\avarprop\wedge
Bool^+(\mynext\avar_1,\cdots,\mynext\avar_{\alpha})$. Also, on the model $\set{\avarprop}^n\emptyset^{\omega}$, for any position $u\leq n$,
$u\in\asetter_l^{i+1}$ iff the assignement $v$, where $v(\mynext\avar_j)=\top$ when $u+1\in\asetter_j^{i}$ for $j\in[1,\alpha]$, satisfies
$Bool^+(\mynext\avar_1,\cdots,\mynext\avar_{\alpha})$ from $\aformula_l$. By our assumption, for any $j\in[1,\alpha]$, $u-1\in\asetter_j^{i}$ implies
$\set{ \set{\avarprop} }^{n-u+1} \in\alang(\aautomaton,q_j)$. Again, by construction of $\aformula_l$, we have that $\transitions(\astate_l,\avarprop)$
is satisfied by the asignment function $v'$ where $v'(\astate)=\top$ iff $\pair{\astate}{\set{\avarprop}^{n-u+1}}\in Acc$. Hence,
$\pair{\astate_l}{\set{\avarprop}^{n-u}}\in Acc$ for all $u\in\asetter_l^{i+1}$. 
}

Thus, for every $i \in \interval{1}{n}$,
the $i$th iterated tuple $\mytuple{\asetter_1^{i}}{\asetter_{\alpha}^{i}}$ verifies that for every $l \in \interval{1}{\alpha}$,  
$u \in \asetter_l^{i}$ iff $u \in \interval{n-i}{n-1}$ and $\set{\avarprop}^{n-u} \in
\alang(\aautomaton,q_l)$.
So, $\set{\avarprop}^{n} \in \alang(\aautomaton)$ iff $0 \in \asetter_1^{n}$.
Since $\mytuple{\asetter_1^{\mu}}{\asetter_{\alpha}^{\mu}}$ is precisely equal 
to  $\mytuple{\asetter_1^{n}}{\asetter_{\alpha}^{n}}$ because of the simple structure of
 $\sigma_n$ (see ($\dag$), we conclude that
$\set{\avarprop}^{n} \in \alang(\aautomaton)$ iff
$\sigma_n, 0 \models  \mu \mytuple{\avar_1}{\avar_{\alpha}}
      \ \mytuple{\aformulabis_1^0}{\aformulabis_{\alpha}^0} \cdot \avar_1$, whence (I) holds. 

From (III), we conclude that $\alang(\aautomaton)$ is non-empty iff
there is $\set{\avarprop} \cdot \set{\avarprop}^{n_1} \cdot
\emptyset^{\omega}$ in $\alang(\acps)$ with $\acps = \pair{\set{\avarprop} \cdot \set{\avarprop}^* \cdot \emptyset^{\omega}}{\top}$
such that $\set{\avarprop} \cdot \set{\avarprop}^{n_1} \cdot
\emptyset^{\omega} \models \aformula_{\aautomaton}$. Since $\acps$ and $\aformula_{\aautomaton}$ can be computed in logarithmic space in
the size of $\aautomaton$, this provides a reduction from the nonemptiness problem for  alternating
finite automata with a singleton
alphabet to the intersection non-emptiness problem with linear $\mu$-calculus. 
Hence, the intersection non-emptiness problem is \pspace-hard (we use only  
 a fixed constrained 
path schema and a formula without past-time operators and without greatest fixed-point operator).

In order to get the \pspace \ upper bound, we take advantage of the methods used so far 
(see Corollary~\ref{corollary-intersection-mu}). \qed 
\end{proof}

\begin{theorem} \label{theorem-membership-linear-mu-calculus}
Membership problem with linear $\mu$-calculus is \pspace-complete.
\end{theorem}

\begin{proof} Let $\aautomaton = (Q,\set{\set{\avarprop}}, \transitions, q_0,F)$ be
a alternating finite  automaton with a singleton alphabet such that $q_0 \not \in F$,
and for every $q_f \in F$, $\transitions(q_f,\set{\avarprop}) = \perp$. 
Let 
$\aformula_{\aautomaton}$ be the formula from linear $\mu$-calculus
such that $\set{\avarprop}^n \in \alang(\aautomaton)$ iff
          $\set{\avarprop}^n \cdot \emptyset^{\omega} \models \aformula_{\aautomaton}$  for all $n \geq 1$. 
Note that if $\alang(\aautomaton)$ is non-empty, then there is $N \leq 2^{\length{\aautomaton}}$
such that $\set{\avarprop}^N \in \alang(\aautomaton)$. Indeed, $\aautomaton$ can be transformed into
a nondeterministic finite automaton with an exponential number of states by using a standard subset
construction. One can check 
that $\alang(\aautomaton)$ is non-empty iff $\set{\avarprop}^{2^{\length{\aautomaton}}} \cdot \emptyset^{\omega} \models \sometimes
\ \aformula_{\aautomaton}$.  Indeed $\alang(\aautomaton)$ is non-empty iff $\set{\avarprop}^n \in \alang(\aautomaton)$ for some $n\leq
2^{\length{\aautomaton}}$ and consequently $\set{\avarprop}^n \cdot \emptyset^{\omega},0 \models \aformula_{\aautomaton}$. This is equivalent to,
$\set{\avarprop}^{2^{\length{\aautomaton}}} \cdot \emptyset^{\omega},2^{\length{\aautomaton}}-n \models \aformula_{\aautomaton}$ and hence implies
$\set{\avarprop}^{2^{\length{\aautomaton}}} \cdot \emptyset^{\omega},0 \models \sometimes \ \aformula_{\aautomaton}$.
$\set{\avarprop}^{2^{\length{\aautomaton}}} \cdot \emptyset^{\omega} \models \sometimes \ \aformula_{\aautomaton}$ is precisely an instance of the
membership problem with linear $\mu$-calculus that can be computed in logarithmic space in the size of $\aautomaton$.  The temporal operator
$\sometimes$ admits a simple definition in linear $\mu$-calculus by using the least fixed point operator.  Hence, membership problem with linear
$\mu$-calculus is \pspace-hard.  In order to get the \pspace \ upper bound, we take advantage of the 
methods used so far (see Lemma~\ref{lemma-mu-calculus-membership}). \qed
\end{proof}
We will now proof the complexity of MC($\lmc$, $\flatcs$).  Let us start by briefly explaining what we mean by $\lmc$ formulas when runs from flat
counter systems are involved. In case of flat counter systems, the nodes are labelled by alphabets from the set $\aalphabet=\powerset{at\cup ag_n}$
whereas the atomic propositions in $\lmc$ formulas are from the alphabet set $at \cup ag_n$. Note that the complexity results for the intersection
non-emptiness problem for $\lmc$ formulas hold for unconstrained alphabet where the flat counter system is labelled with the same set of alphabet as
the atomic propositions in $\lmc$ formulas. Thus, to take advantage of the complexity characterization of the intersection non-emptiness
problem, we would transform $\lmc$ formula to have atomic propositions from $\aalphabet$. Following similar development and reasoning as in
Section~\ref{section-pspace-fo} for FO, we can conclude that $\lmc$ formulas have the nice sub-alphabet property.
\begin{theorem} 
\label{theorem-linear-mu-calculus-pspace-complete}
MC($\lmc$, $\flatcs$) is \pspace-complete. 
\end{theorem}

\begin{proof}
Given a $\lmc$ formula $\aformula$ and a counter system $\asys$, we can obtain a formula $\aformula'$ such that the atomic propositions are from the
set $\powerset{at\cup ag_n}$ and $\size{\aformula'}$ is polynomial in $\size{\aformula}$ due to the fact that $\lmc$ formulas have nice sub-alphabet
property.  By using Lemma~\ref{lemma-complexity-subalphabet} and by using the fact that the intersection non-emptiness problem with linear
$\mu$-calculus is in \pspace \ (Theorem~\ref{theorem-membership-linear-mu-calculus}), we can obtain the \pspace~upper bound for MC($\lmc$, $\flatcs$).
Now, let us establish the lower bound.  Let $\aautomaton = (Q,\set{\set{\avarprop}},
\transitions, q_0,F)$ be an alternating finite automaton with a singleton alphabet such that $q_0 \not \in F$, and for every $q_f \in F$,
$\transitions(q_f,\set{\avarprop}) = \perp$.  Let $\aformula_{\aautomaton}$ be the formula from linear $\mu$-calculus such that $\set{\avarprop}^n \in
\alang(\aautomaton)$ iff $\set{\avarprop}^n \cdot \emptyset^{\omega} \models \aformula_{\aautomaton}$ for all $n \geq 1$.  Let $\asys$ be the flat
Kripke structure (no need for counters) with three states $q_0,q,q_f$, with transitions $q_0 \rightarrow q$, $q \rightarrow q$ (self-loop), $q
\rightarrow q_f$ and $q_f \rightarrow q_f$ (self-loop) and such that $\alabelling(q_0) = \alabelling(q) = \set{\avarprop}$ and $\alabelling(q_f) =
\emptyset$.  One can check that $\alang(\aautomaton)$ is non-empty iff there is an infinite run from $q_0$ satisfying the formula
$\aformula_{\aautomaton}$, which is precisely an instance of MC($\lmc$, $\flatcs$).  Hence, MC($\lmc$, $\flatcs$) is \pspace-hard.  \qed
\end{proof}

According to the proof of Theorem~\ref{theorem-linear-mu-calculus-pspace-complete},
\pspace-hardness already holds true for a fixed Kripke structure, that is actually a simple
path schema. This contrasts with the fact that model-checking loops with linear $\mu$-calculus
is equivalent to model-checking finite graphs with modal $\mu$-calculus~\cite{Markey&Schnoebelen06}, as far as 
worst-case complexity is concerned (in UP$\cap$co-UP~\cite{Jurdzinski98}).
Hence, for linear $\mu$-calculus, there is already a complexity gap between model-checking
unconstrained path schemas with two loops and model-checking rudimentary unconstrained path 
schemas (Kripke structures) made of  
a single loop. This illustrates also an interesting difference with Past LTL for which
model-checking unconstrained path schemas  (from Kripke structures) with a bounded number of loops 
can be done in polynomial time~\cite[Theorem 9]{DemriDharSangnier12}. 



}%
 \else \section{Taming Linear $\mu$-calculus and Other Languages}
We now consider several specification languages defining $\omega$-regular properties
on atomic propositions and arithmetical constraints. First, we deal with BA  by establishing 
Theorem~\ref{thm:stuttering-buchi}
and then deduce results for ABA, ETL and $\lmc$.

\begin{theorem} \label{thm:stuttering-buchi}
Let $\bauto =\tuple{\states,\aalphabet,\astate_0,\edges,F}$
be a B\"uchi automaton (with standard definition)
and $\acps=\pair{\cpschema}{\aconstraintsystem}$ be a constrained path schema over $\Sigma$. We have
$\alang(\acps)\cap\alang(\bauto)\neq \emptyset$ iff there exists 
$\vect{y}\in[0,\powerset{\mathtt{pol}_1(\size{\acps})}+2.\card{\states}^{k}\times
  \powerset{\mathtt{pol}_1(\size{\acps})+\mathtt{pol}_2(\size{\acps})}]^{k-1}$ such that
$\aseg_1(\aloop_1)^{\vect{y}[1]}\linebreak[0]\ldots\aseg_{k-1} \linebreak[0] 
(\aloop_{k-1})^{\vect{y}[k-1]}\aseg_k\aloop_k^{\omega}\in\alang(\bauto)\cap\alang(\acps)$
($\mathtt{pol}_1$ and $\mathtt{pol}_2$ are from Theorem~\ref{theorem-pottier91}). 
\end{theorem}

Theorem~\ref{thm:stuttering-buchi} can be viewed as a pumping lemma involving an automaton and semilinear sets. Thanks to it we  obtain an
exponential bound for the map $\amap_{{\rm BA}}$ so that $\amap_{{\rm
    BA}}(\bauto,\acps)=\powerset{\mathtt{pol}_1(\size{\acps})}+2.\card{\states}^{\size{\acps}}\times \powerset{\mathtt{pol}_1(\size{\acps})+
  \mathit{pol}_2(\size{\acps})}$.  So checking $\alang(\acps) \cap \alang(\bauto)\neq\emptyset$ amounts to guess some $\vect{n} \in
\interval{0}{\powerset{\mathtt{pol}_1(\size{\acps})}+2.\card{\states}^{\size{\acps}}\times
  \powerset{\mathtt{pol}_1(\size{\acps})+\mathtt{pol}_2(\size{\acps})}}^{k-1}$ and to verify whether the word $\aword = \aseg_1 (\aloop_1)^{\vect{n}[1]}
\cdots \aseg_{k-1} (\aloop_{k-1})^{\vect{n}[k-1]} \aseg_k (\aloop_k)^{\omega} \in \alang(\acps) \cap \alang(\bauto).  $ Checking whether $\aword \in
\alang(\acps)$ can be done in polynomial time in $\size{\bauto} + \size{\acps}$ since this amounts to check $\vect{n} \models \aformula$.  Checking
whether $\aword \in \alang(\bauto)$ can be also done in polynomial time by using the results from~\cite{Markey&Schnoebelen03}.  Indeed, $\aword$ can
be encoded in polynomial time as a pair of straight-line programs and by~\cite[Corollary 5.4]{Markey&Schnoebelen03} this can be done in polynomial
time.  So, the membership problem for \buchi~automata is in \ptime.  By using that $\BA$ has the nice subalphabet property and that we can create a
polynomial size \buchi~automata from a given BA specification and $\acps$, we get the following result.
\begin{lemma} \label{corollary-intersection-ba}
The intersection non-emptiness problem with BA is \np-complete. 
\end{lemma}  
Now, by  Lemma~\ref{lemma-nice-subalphabet-property},
Lemma~\ref{lemma-complexity-subalphabet} and Lemma~\ref{corollary-intersection-ba},
we get the  result below for which the lower bound is obtained from an easy reduction of SAT.
\begin{theorem} \label{theorem-ba-np-complete}
$\mcpb{\BA}{\flatcs}$ is \np-complete. 
\end{theorem}

We are now ready to deal with ABA, ETL and linear $\mu$-calculus. 
A  language $\aspeclanguage$ has the \defstyle{nice BA property} iff for every
specification $\aspecification$ from  $\aspeclanguage$, we can build a B\"uchi automaton 
$\bauto_\aspecification$ such that 
$\alang(\aspecification)=\alang(\bauto_\aspecification)$, 
each state of $\bauto_{\aspecification}$ is of polynomial size, 
it can be checked if a state is initial [resp. accepting] in polynomial space
and the transition relation can be decided in polynomial space too. 
So, given a language $\aspeclanguage$ having the nice BA property, a constrained path schema $\acps$
and a specification in $\aspecification \in \aspeclanguage$, 
if $\alang(\acps) \cap \alang(\aspecification)$ is non-empty, then there is an $\omega$-word in 
$\alang(\acps) \cap \alang(\aspecification)$ such that each loop is taken at most a number of times
bounded by $\amap_{{\rm BA}}(\bauto_{\aspecification},\acps)$.
So  $\amap_{\aspeclanguage}(\aspecification,\acps)$ is obviously bounded by 
$\amap_{{\rm BA}}(\bauto_{\aspecification},\acps)$.
Hence,  checking whether  $\alang(\acps) \cap \alang(\aspecification)$  is non-empty amounts to guess some
$\vect{n} \in \interval{0}{\amap_{\aspeclanguage}(\aspecification,\acps)}^{k-1}$ and check
whether 
$\aword = \aseg_1 (\aloop_1)^{\vect{n}[1]} \cdots \aseg_{k-1} (\aloop_{k-1})^{\vect{n}[k-1]} \aseg_k (\aloop_k)^{\omega} \in 
\alang(\acps) \cap \alang(\aspecification)
$.  
Checking whether $\aword \in \alang(\acps)$ can be done in polynomial time in 
$\size{\aspecification} + \size{\acps}$ since this amounts to check $\vect{n} \models \aformula$.
Checking whether $\aword \in \alang(\aspecification)$ can be done in
nondeterministic polynomial space by reading $\aword$ while guessing an accepting run 
for $\bauto_{\aspecification}$. Actually, one guesses a state $q$  from $\bauto_{\aspecification}$
and check whether  the prefix
$\aseg_1 (\aloop_1)^{\vect{n}[1]} \cdots \aseg_{k-1} (\aloop_{k-1})^{\vect{n}[k-1]} \aseg_k$ 
can reach it and then nonemptiness between $ (\aloop_k)^{\omega}$ and the 
B\"uchi automaton $\bauto_{\aspecification}^{q}$ 
in which $q$ is an initial state is checked. Again, this can be done in 
nondeterministic polynomial space thanks to the nice BA property. 
We obtain the  lemma below.

\begin{lemma} 
\label{lemma-membership-gen}
%
%
Membership problem and intersection non-emptiness problem for $\aspeclanguage$ having the nice BA property are in \pspace.
\end{lemma}
Let us recall consequences of results from the literature.
ETL has the nice BA property by~\cite{Vardi&Wolper94}, 
linear $\mu$-calculus has the nice BA property by~\cite{Vardi88}
and  ABA has the nice BA property by~\cite{MiyanoHayashi84}.
Note that the results for ETL and ABA can be also obtained thanks to translations into  linear $\mu$-calculus.
By Lemma~\ref{lemma-membership-gen}, Lemma~\ref{lemma-complexity-subalphabet} and the above-mentioned results, we
obtain the following results. 
\begin{theorem} \label{theorem-pspace-all}
$\mcpb{\ABA}{\flatcs}$, $\mcpb{\ETL}{\flatcs}$ and $\mcpb{\lmc}{\flatcs}$ are in \pspace.
\end{theorem}
Note that for obtaining the \pspace~upper bound, we use the same procedure for all the logics.
Using that the emptiness problem for finite alternating automata over a single letter alphabet is \pspace-hard \cite{JancarSawa07}, we are also able to get lower bounds.
\begin{theorem} \label{theorem-hardness}
(I) The intersection non-emptiness problem for $\ABA$ [resp. $\lmc$] is \pspace-hard.
(II) $\mcpb{\ABA}{\flatcs}$ and $\mcpb{\lmc}{\flatcs}$ are \pspace-hard.
\end{theorem}

According to the proof of Theorem~\ref{theorem-hardness} 
(Appendix~\ref{section-proof-theorem-hardness}),
\pspace-hardness already holds for a fixed Kripke structure, that is actually a simple
path schema. Hence, for linear $\mu$-caluclus, there is a complexity gap between model-checking
unconstrained path schemas with two loops (in UP$\cap$co-UP~\cite{Jurdzinski98}) and model-checking 
 unconstrained path 
schemas (Kripke structures) made of  
a single loop, which is in contrast to Past LTL for which model-checking unconstrained path schemas with a bounded number of loops is in~\ptime~\cite[Theorem 9]{DemriDharSangnier12}.

As an additional corollary, we can solve the global model-checking problem with existential Presburger formulae. The global model-checking consists in
characterizing the set of initial configurations from which there exists a run satisfying a given specification. We knew that Presburger formulae
exist for global model-checking~\cite{demri-model-10} for Past LTL (and therefore for FO) but we can conclude that they are structurally simple and we
provide an alternative proof. Moreover, the question has been open for $\lmc$ since the decidability status of $\mcpb{\lmc}{\flatcs}$ has been only
resolved in the present work.

\begin{corollary}  \label{corollary-global-model-checking}
Let $\aspeclanguage$ be a specification language among FO, BA, ABA, ETL or $\lmc$.
Given a flat counter system $\asys$, a control state $\astate$ and a specification
$\aspecification$ in  $\mathcal{L}$,  one can effectively build an existential Presburger formula
 $\aformula(\avariableter_1, \ldots,\avariableter_n)$ such that for all $\avect \in \nat^n$. $\avect \models \aformula$ iff there is
a run $\arun$ starting at $\pair{\astate}{\vec{\avect}}$ verifying
$\arun \models \aspecification$.
\end{corollary}
 
 \fi 
\section{Conclusion}
\label{section-conclusion}
We characterized the complexity of 
\iftechreport 
model-checking problems on flat counter systems for prominent linear-time specification languages
\else
$\mcpb{\aspeclanguage}{\flatcs}$ for prominent linear-time 
 specification languages $\aspeclanguage$
\fi 
whose letters are made of atomic propositions and 
linear constraints.  We proved the \pspace-completeness of the problem
with linear $\mu$-calculus (decidability was  open), for alternating B\"uchi automata and also for 
$\FO$.
When specifications are expressed with B\"uchi automata, the problem is
shown \np-complete. Global model-checking is also possible on flat counter systems with such specification languages.
Even though the core of our work relies on
small solutions of quantifier-free Presburger formulae, stuttering properties,
automata-based approach and on-the-fly algorithms, 
our approach is designed  to be  generic.
\iftechreport 
For instance, the introduction of constrained path schemas,
the abstract algorithm, the nice subalphabet property and the nice BA property (quite standard
for linear-time temporal logics in \pspace) have been important to reveal  the essential issues. 
\fi 
Not only this witnesses the robustness of our method but our complexity characterization justifies further
why verification of flat counter systems can be at the core of methods for model-checking counter systems.
Our main results are in the table below with useful comparisons 
(`Ult. periodic KS'  stands for ultimately periodic Kripke structures namely a path followed by a loop). 
\iftechreport

We already know that  model-checking flat counter systems 
with branching-time logic CTL$^{\star}$ is decidable~\cite{demri-model-10} but complexity characterization is still open.
More interestingly, the decidability status with modal $\mu$-calculus on  flat counter systems 
is unknown and we are not aware of any method that can deal with it, unfortunately including
the method we used  for linear $\mu$-calculus. 
\fi 
\iftechreport
\begin{center}
\scalebox{0.78}{
\begin{tabular}{|c||c|c|c|c|}
\hline
& 
{\small Flat counter systems}
& 
{\small Kripke struct.}
 & 
{\small Flat Kripke struct.}
& 
{\small Ult. periodic KS} \\ 
\hline
\hline
$\lmc$ & \pspace-C (Thm.~\ref{theorem-pspace-all}) 
                      & \pspace-C~\cite{Vardi88} 
                      & \pspace-C (Thm.~\ref{theorem-pspace-all}) 
                      & in UP$\cap$co-UP~\cite{Markey&Schnoebelen06} \\ \hline
ABA & \pspace-C (Thm.~\ref{theorem-pspace-all}) 
    & \pspace-C 
    & \pspace-C (Thm.~\ref{theorem-pspace-all}) 
    & in \ptime \ (see e.g.~\cite[p.~3]{Kupferman&Vardi01}) \\ \hline
ETL & in \pspace~(Thm.~\ref{theorem-pspace-all})
    & \pspace-C~\cite{Sistla&Clarke85}  
    & in \pspace~\cite{Sistla&Clarke85} 
    & in \ptime \ (see e.g.~\cite{Piterman00,Kupferman&Vardi01}) \\ \hline
BA & \np-C (Thm.\ref{theorem-ba-np-complete}) 
   & in \ptime 
   & in \ptime 
   &  in \ptime \\ \hline\hline
FO & \pspace-C (Thm.~\ref{theorem-fo-pspace-complete}) 
   & Non-el.~\cite{Stockmeyer74} 
   & \pspace-C (Thm.~\ref{theorem-fo-pspace-complete}) 
   & \pspace-C~\cite{Markey&Schnoebelen03}\\ \hline
Past LTL & \np-C~\cite{DemriDharSangnier12} 
         & \pspace-C~\cite{Sistla&Clarke85} 
         & \np-C~\cite{Kuhtz&Finkbeiner11,DemriDharSangnier12} 
         & \ptime~\cite{Laroussinie&Markey&Schnoebelen02}\\
\hline
\end{tabular}
}
\end{center}
\else%
\begin{center}
\scalebox{0.74}{
\begin{tabular}{|c||c|c|c|c|}
\hline
& 
{\small Flat counter systems}
& 
{\small Kripke struct.}
 & 
{\small Flat Kripke struct.}
& 
{\small Ult. periodic KS} \\ 
\hline
\hline
$\lmc$ & \pspace-C (Thm.~\ref{theorem-pspace-all}) 
                      & \pspace-C~\cite{Vardi88} 
                      & \pspace-C (Thm.~\ref{theorem-pspace-all}) 
                      & in UP$\cap$co-UP~\cite{Markey&Schnoebelen06} \\ \hline
ABA & \pspace-C (Thm.~\ref{theorem-pspace-all}) 
    & \pspace-C 
    & \pspace-C (Thm.~\ref{theorem-pspace-all}) 
    & in \ptime \ (see e.g.~\cite[p.~3]{Kupferman&Vardi01}) \\ \hline
ETL & in \pspace~(Thm.~\ref{theorem-pspace-all})
    & \pspace-C~\cite{Sistla&Clarke85}  
    & in \pspace~\cite{Sistla&Clarke85} 
    & in \ptime \ (see e.g.~\cite{Piterman00,Kupferman&Vardi01}) \\ \hline
BA & \np-C (Thm.\ref{theorem-ba-np-complete}) 
   & in \ptime 
   & in \ptime 
   &  in \ptime \\ \hline\hline
FO & \pspace-C (Thm.~\ref{theorem-fo}) 
   & Non-el.~\cite{Stockmeyer74} 
   & \pspace-C (Thm.~\ref{theorem-fo}) 
   & \pspace-C~\cite{Markey&Schnoebelen03}\\ \hline
Past LTL & \np-C~\cite{DemriDharSangnier12} 
         & \pspace-C~\cite{Sistla&Clarke85} 
         & \np-C~\cite{Kuhtz&Finkbeiner11,DemriDharSangnier12} 
         & \ptime~\cite{Laroussinie&Markey&Schnoebelen02}\\
\hline
\end{tabular}
}
\end{center}
\fi


\bibliographystyle{abbrv}
\bibliography{biblio-icalp13}

\begin{thebibliography}{10}

\bibitem{Arnold&Niwinski01}
A.~Arnold and D.~Niwinski.
\newblock {\em Rudiments of $\mu$-calculus}.
\newblock Elsevier, 2001.

\bibitem{Boigelot98}
B.~Boigelot.
\newblock {\em Symbolic methods for exploring infinite state spaces}.
\newblock PhD thesis, Universit\'e de Li\`ege, 1998.

\bibitem{Bozga&Iosif&Konecny10}
M.~Bozga, R.~Iosif, and F.~Konecn{\'y}.
\newblock Fast acceleration of ultimately periodic relations.
\newblock In {\em CAV'10}, volume 6174 of {\em LNCS}, pages 227--242. Springer,
  2009.

\bibitem{Comon&Jurski98}
H.~Comon and Y.~Jurski.
\newblock Multiple counter automata, safety analysis and {PA}.
\newblock In {\em CAV'98}, volume 1427 of {\em LNCS}, pages 268--279. Springer,
  1998.

\bibitem{DemriDharSangnier12}
S.~Demri, A.~Dhar, and A.~Sangnier.
\newblock Taming {P}ast {LTL} and {F}lat {C}ounter {S}ystems.
\newblock In {\em IJCAR'12}, volume 7364 of {\em LNAI}, pages 179--193.
  Springer, 2012.
\newblock See also \url{http://arxiv.org/abs/1205.6584}.

\bibitem{demri-model-10}
S.~Demri, A.~Finkel, V.~Goranko, and G.~van Drimmelen.
\newblock Model-checking \(\textsf{CTL}^{*}\) over flat {P}resburger counter
  systems.
\newblock {\em JANCL}, 20(4):313--344, 2010.

\bibitem{Etessami&Wilke00}
K.~Etessami and T.~Wilke.
\newblock An until hierarchy and other applications of an
  {E}hrenfeucht-{F}ra\"{\i}ss{\'e} game for temporal logic.
\newblock {\em I\&C}, 160(1--2):88--108, 2000.

\bibitem{Finkel&Leroux02b}
A.~Finkel and J.~Leroux.
\newblock How to compose {P}resburger accelerations: Applications to broadcast
  protocols.
\newblock In {\em FST\&TCS'02}, volume 2256 of {\em LNCS}, pages 145--156.
  Springer, 2002.

\bibitem{JancarSawa07}
P.~Jan\v{c}ar and Z.~Sawa.
\newblock A note on emptiness for alternating finite automata with a one-letter
  alphabet.
\newblock {\em IPL}, 104(5):164--167, 2007.

\bibitem{Jurdzinski98}
M.~Jurdzi{\'n}ski.
\newblock Deciding the winner in parity games is in {UP} $\cap$ co-{UP}.
\newblock {\em IPL}, 68(3):119--124, 1998.

\bibitem{Kuhtz&Finkbeiner11}
L.~Kuhtz and B.~Finkbeiner.
\newblock Weak {K}ripke structures and {LTL}.
\newblock In {\em {CONCUR}'11}, volume 6901 of {\em LNCS}, pages 419--433.
  Springer, 2011.

\bibitem{Kupferman&Vardi01}
O.~Kupferman and M.~Vardi.
\newblock Weak alternating automata are not that weak.
\newblock {\em ACM Transactions on Computational Logic}, 2(3):408--429, 2001.

\bibitem{Kucera&Strejcek05}
A.~Ku\v{c}era and J.~Strej\v{c}ek.
\newblock The stuttering principle revisited.
\newblock {\em Acta Informatica}, 41(7--8):415--434, 2005.

\bibitem{Laroussinie&Markey&Schnoebelen02}
F.~Laroussinie, N.~Markey, and P.~Schnoebelen.
\newblock Temporal logic with forgettable past.
\newblock In {\em LICS'02}, pages 383--392. IEEE, 2002.

\bibitem{Leroux&Sutre05}
J.~Leroux and G.~Sutre.
\newblock Flat counter systems are everywhere!
\newblock In {\em ATVA'05}, volume 3707 of {\em LNCS}, pages 489--503.
  Springer, 2005.

\bibitem{Libkin04}
L.~Libkin.
\newblock {\em Elements of {F}inite {M}odel {T}heory}.
\newblock Springer, 2004.

\bibitem{Markey&Schnoebelen03}
N.~Markey and P.~Schnoebelen.
\newblock Model checking a path.
\newblock In {\em CONCUR'03, Marseille, France}, volume 2761 of {\em LNCS},
  pages 251--261. Springer, 2003.

\bibitem{Markey&Schnoebelen06}
N.~Markey and P.~Schnoebelen.
\newblock Mu-calculus path checking.
\newblock {\em IPL}, 97(6), 2006.

\bibitem{minsky-computation-67}
M.~Minsky.
\newblock {\em Computation, Finite and Infinite Machines}.
\newblock Prentice Hall, 1967.

\bibitem{MiyanoHayashi84}
S.~Miyano and T.~Hayashi.
\newblock Alternating finite automata on $\omega$-words.
\newblock {\em Theor. Comput. Sci.}, 32:321--330, 1984.

\bibitem{Piterman00}
N.~Piterman.
\newblock Extending temporal logic with $\omega$-automata.
\newblock Master's thesis, The Weizmann Institute of Science, 2000.

\bibitem{Pottier91}
L.~Pottier.
\newblock {Minimal Solutions of Linear Diophantine Systems: Bounds and
  Algorithms}.
\newblock In {\em RTA'91}, pages 162--173. Springer, 1991.

\bibitem{Sistla&Clarke85}
A.~Sistla and E.~Clarke.
\newblock The complexity of propositional linear temporal logic.
\newblock {\em JACM}, 32(3):733--749, 1985.

\bibitem{Stockmeyer74}
L.~J. Stockmeyer.
\newblock {\em The complexity of decision problems in automata and logic}.
\newblock PhD thesis, MIT, 1974.

\bibitem{Vardi88}
M.~Vardi.
\newblock A temporal fixpoint calculus.
\newblock In {\em POPL'88}, pages 250--259. ACM, 1988.

\bibitem{Vardi&Wolper94}
M.~Vardi and P.~Wolper.
\newblock Reasoning about infinite computations.
\newblock {\em I\&C}, 115, 1994.

\bibitem{Walukiewicz01}
I.~Walukiewicz.
\newblock Automata and logic, 2001.
\newblock Lecture notes.

\bibitem{Wolper83}
P.~Wolper.
\newblock Temporal logic can be more expressive.
\newblock {\em I\&C}, 56:72--99, 1983.

\end{thebibliography}
\ifconference
\newpage
\appendix
\section{Proof of Theorem~\ref{theorem-reduction}}
Below, we provide the main steps of the proof, details can be found in~\cite{DemriDharSangnier12}. 

\begin{proof} (sketch)
Let us explain how to build the set $\aset$.
\begin{enumerate}
\itemsep 0 cm
\item Given a flat counter system $\asys$ and a state $\astate$ from $\aconf$, 
      there is at most an exponential number of minimal path schemas
      starting at $\astate$ in the sense of~\cite[Lemma 4]{DemriDharSangnier12}. 
      Let $\asetbis_1$ be this set of minimal path schemas.
\item For each path schema $P$ in $\asetbis_1$, there is a set of path schemas $\asetbis_{P}$  
      such that the path schemas in $\asetbis_{P}$   have no disjunctions in guards and satisfaction
      of guards can be concluded from the states, see~\cite[Theorem 14]{DemriDharSangnier12}. 
      Let $\asetbis_2$ be this set of unfolded path schemas
      and it is of cardinality at most exponential.
\item Following~\cite[Lemma 12]{DemriDharSangnier12}, every path schema from $\asetbis_2$ is equivalent 
      to a constrained path schema. The set $\aset$ is precisely the set of constrained path schemas
      obtained from all the unfolded path schemas from $\asetbis_2$.
\end{enumerate}
Completeness of the set $\aset$ is a consequence of~\cite[Lemma 12]{DemriDharSangnier12} 
and~\cite[Theorem 14(4--6)]{DemriDharSangnier12}. Satisfaction of the size constraints is
a consequence of~\cite[Lemma 12]{DemriDharSangnier12} 
and~\cite[Theorem 14(2--3)]{DemriDharSangnier12}.
\qed
\end{proof}

\section{Proof of Lemma~\ref{lemma-nice-subalphabet-property}} 
\label{section-proof-lemma-nice-subalphabet-property}
\begin{proof} 
Let $\aspecification = \triple{Q}{E}{q_0,F}$ be a specification in BA. 
over the alphabet $\triple{at}{ag_n}{\aalphabet}$ and
$\aalphabet' \subseteq \aalphabet$. The specification $\aspecification' =  \triple{Q}{E'}{q_0,F}$
such that $\alang(\aspecification') = \alang(\aspecification) \cap (\aalphabet')^{\omega}$
is defined as follows: for every $q \step{\aformulabis} q' \in E$, we include in
$E'$ the edge 
$q \step{(\bigvee_{\aletter \in \aalphabet'} \ \aformulabis_{\aletter}) \wedge \aformulabis} q'$
where  $\aformulater_{\aletter}$ is defined as a conjunction made of positive
literals from $\aletter$ and negative literals from $(at \cup ag_n) \setminus \aletter$. 
A similar transformation can be performed with specifications in ABA.

Let $\aformula$ be a formula for $\aspeclanguage$ among linear $\mu$-calculus, ETL or Past LTL
built over atomic formulae in $at \cup ag_n$ and  $\triple{at}{ag_n}{\aalphabet'}$ be a constrained alphabet.
The formulae $\aformula'$ such that $\alang(\aformula') = \alang(\aformula) \cap (\aalphabet')^{\omega}$
is obtained from $\aformula$ by replacing every atomic formula $\aformulabis$ by
$\bigvee_{\set{\aletter \in \aalphabet' \mid \aformulabis \in \aletter}} \aformulater_{\aletter}$.
\qed
\end{proof}

\iftechreport
\section{Correctness of Algorithm~\ref{algorithm:abstact}}
\else
\section{Correctness of Algorithm 1}
\fi 
\label{section-proof-correctness-algo1}
\begin{proof}
First assume there exists a run $\arun$ of $\asys$ starting at $\aconf_0$ such that $\arun \models A$. By Theorem \ref{theorem-reduction},  there is a constrained path schema $\acps$ with an alphabet of the form 
      $\triple{at}{ag_n}{\aalphabet'}$ in $\aset$ and $\aword \in \alang(\acps)$ such that $\arun \models \aword$. Consequently we deduce that  $\aword \in \alang(A)$ and that $\alang(\acps) \cap \alang(A) \neq \emptyset$. Since $\alang(\acps) \subseteq (\aalphabet')^{\omega}$ and since $\alang(\aspecification) \cap (\aalphabet')^{\omega} = \alang(\aspecification')$, we deduce that $\alang(\acps)\cap\alang(\aspecification') \neq \emptyset$. Hence the Algorithm has an accepting run. 

Now if the \iftechreport Algorithm~\ref{algorithm:abstact}  \else Algorithm~1 \fi 
has an accepting run, we deduce that there exists a constrained path schema $\acps$ with an alphabet of the form $\triple{at}{ag_n}{\aalphabet'}$ in $\aset$ such that there exists a word $\aword$ in $\alang(\acps) \cap \alang(A')$. Using the nice subalphabet property we deduce that $\aword \in \alang(A)$ and by the last point of Theorem \ref{theorem-reduction}, we know that there exists a run $\arun$ from $\asys$ starting at $\aconf_0$ such that $\arun \models \aword$. This allows us to conclude that $\arun \models A$. \qed
\end{proof}

\section{Proof of Lemma~\ref{lemma-fo-subalphabet-property}}
\begin{proof}
Consider a FO formula $\aformula$ that defines a specification
over the  constrained alphabet $\triple{at}{ag_n}{\aalphabet}$ with $\aalphabet = \powerset{at \cup ag_n}$.
Consider a subalphabet $\aalphabet'\subseteq\aalphabet$.
Let $\aformula''$ be the formula obtained from $\aformula$
by replacing  every occurrence of  $\avarprop(\avar)$ by 
      $\bigvee_{\set{\aletter \in \aalphabet' \mid \avarprop \in \aletter}} \aletter(\avar)$
and  every occurrence of  $\aguard(\avar)$ is replaced by 
      $\bigvee_{\set{\aletter \in \aalphabet' \mid \aguard \in \aletter }} \aletter(\avar)$.
It is easy to see that, by construction, $\alang(\aformula'')= \alang(\aformula')\cap(\aalphabet)^{\omega}$
\qed
\end{proof}

\section{EF Games and Proof of Theorem~\ref{theorem-stuttering-fo}}
\label{appendix:ef-games}

\section{Proof of Lemma~\ref{lemma:fo-cps-small-loop}}
\begin{proof} Let $\acps = \pair{\aseg_1 (\aloop_1)^* \cdots \aseg_{k-1} (\aloop_{k-1})^* \aseg_k (\aloop_k)^{\omega}}{\aformula(\avariable_1, 
\ldots, \avariable_{k-1})}$ be a constrained path schema and $\aformulabis$ be a first-order sentence. 
Suppose that 
$$
\aseg_1 (\aloop_1)^{\vect{n}[1]} \cdots \aseg_{k-1} (\aloop_{k-1})^{\vect{n}[k-1]} \aseg_k (\aloop_k)^{\omega} 
\in \alang(\acps) \cap \alang(\aformulabis).
$$
Let  $\basis \subseteq  \interval{0}{2^{p_1(\size{\acps})}}^{k-1}$ and 
           $\aperiod_1, \ldots, \aperiod_{\alpha} \in \interval{0}{2^{p_1(\size{\acps})}}^{k-1}$ defined for the guard $\aformula$
following Theorem~\ref{theorem-pottier91}. 
Since $\vect{n} \models \aformula$,
there are $\abasis \in \basis$ and $\vect{a} \in \Nat^{\alpha}$ such that
           $\vect{n} = \abasis + \vect{a}[1] \aperiod_1 + \cdots + \vect{a}[\alpha] \aperiod_{\alpha}$.
Let  $\vect{a}' \in \Nat^{\alpha}$ defined from $\vect{a}$ such that
$\vect{a}'[i] = \vect{a}[i]$ if $\vect{a}[i] \leq 2^{\qheight{\aformulabis} +1} + 1$ otherwise 
$\vect{a}'[i] = 2^{\qheight{\aformulabis} +1} + 1$. Note that 
$\vect{n}'  = \abasis + \vect{a}'[1] \aperiod_1 + \cdots + \vect{a}'[\alpha] \aperiod_{\alpha}$ still satisfies
$\aformula$ and for every loop $i \in \interval{1}{k-1}$, 
$\vect{n}[i] >  2^{\qheight{\aformulabis} +1}$ iff $\vect{n}'[i] >  2^{\qheight{\aformulabis} +1}$.
By  Theorem~\ref{theorem-stuttering-fo}, 
$\aseg_1 (\aloop_1)^{\vect{n}'[1]} \cdots \aseg_{k-1} (\aloop_{k-1})^{\vect{n}'[k-1]} \aseg_k (\aloop_k)^{\omega} 
\in \alang(\aformulabis)$. 

Now, let us bound the values in  $\vect{n}'$. 
\begin{itemize}
\itemsep 0 cm
\item There are at most $2^{p_2(\size{\acps})}$ periods.
\item Each basis or period has values in $\interval{0}{2^{p_1(\size{\acps})}}$.
\item Each period in $\vect{n}'$ is taken at most $2^{\qheight{\aformulabis} +1} + 1$ times.
\end{itemize}
Consequently, each $\vect{n}'[i]$ is bounded by
$$
2^{p_1(\size{\acps})} + (2^{\qheight{\aformulabis} +1} + 1) 2^{p_2(\size{\acps})} \times 2^{p_2(\size{\acps})}
$$
which is itself bounded by $2^{(\qheight{\aformulabis} + 2) + p_1(\size{\acps}) + p_2(\size{\acps})}$. 
\qed
\end{proof}

\section{Proof of Lemma~\ref{lemma:fo-membership}}
\ifconference
\begin{proof}
\fi
We want to show that the membership problem with first-order logic (with unconstrained alphabets) can be solved in
polynomial space in $\size{\acps} + \size{\aformulabis}$. 
Let $\acps$,  $\aformulabis$ and $\vect{n} \in \Nat^{k-1}$ be an instance of the problem.
For $i \in \interval{1}{k-1}$, let $\vect{n}'[i] \egdef \mathtt{min}(\vect{n}[i],2^{\qheight{\aformulabis}+1} +1)$. 
By  Theorem~\ref{theorem-stuttering-fo},  the propositions below are equivalent:
\begin{itemize}
\item  $\aseg_1 (\aloop_1)^{\vect{n}[1]} \cdots \aseg_{k-1} (\aloop_{k-1})^{\vect{n}[k-1]} \aseg_k (\aloop_k)^{\omega} \in 
\alang(\aformulabis)$,
\item $\aseg_1 (\aloop_1)^{\vect{n}'[1]} \cdots \aseg_{k-1} (\aloop_{k-1})^{\vect{n}'[k-1]} \aseg_k (\aloop_k)^{\omega} \in 
\alang(\aformulabis)$.
\end{itemize}
Without any loss of generality, let us assume then that 
$\vect{n} \in \interval{0}{2^{\qheight{\aformulabis}+1} +1}^{k-1}$.

Let us decompose $\aword = \aseg_1 (\aloop_1)^{\vect{n}[1]} \cdots \aseg_{k-1} 
(\aloop_{k-1})^{\vect{n}[k-1]} \aseg_k (\aloop_k)^{\omega}$ as
$u \cdot (v)^{\omega}$ where $u = \aseg_1 (\aloop_1)^{\vect{n}[1]} \cdots \aseg_{k-1} (\aloop_{k-1})^{\vect{n}[k-1]} \aseg_k$
and $v = \aloop_k$. 
Note that the length of $u$ is exponential in the size of the instance.  
We write $\hat{\aformulabis}$ to denote the formula $\aformulabis$ in which every existential quantification is
relativized to positions less than $\length{u} + \length{v} \times 2^{\qheight{\aformulabis}}$. 
This means that every quantification '$\exists \ \avariable \ \cdots$' is replaced by
'$\exists \ \avariable <  (\length{u} + \length{v} \times 2^{\qheight{\aformulabis}}) \ \cdots$'. 
By~\cite{Markey&Schnoebelen03}, we know that $\aword \models \aformulabis$ iff 
$\aword \models \hat{\aformulabis}$. Now, checking $\aword \models \hat{\aformulabis}$ can be done in polynomial space 
by using a standard first-order model-checking algorithm by restricting ourselves to positions in
$\interval{0}{\length{u} + \length{v} \times 2^{\qheight{\aformulabis}}}$ for existential quantifications.
Such positions can be obviously encoded in polynomial space. Moreover, note that given 
$i \in \interval{0}{\length{u} + \length{v} \times 2^{\qheight{\aformulabis}}}$, one can check in polynomial time
what is the $i$th letter of $\aword$. Details are standard and omitted here. By way of example, the $i$th letter of $\aword$
is the first letter of $l_k$ iff 
$i \geq \alpha$ and $(i-\alpha) = 0 \ {\rm mod} \ \length{\aloop_k})$
with $\alpha = (\Sigma_{j \in \interval{1}{k-1}} (\length{p_j} + \length{l_j} \times \vect{n}[j])) + \length{p_j}$.  

\begin{algorithm}
{\footnotesize
\caption{FOSAT$(\acps,\vect{n},\aformulabis,\amap)$}
\label{algorithm:FOSAT}
\begin{algorithmic}[1]
\IF{$\aformula = \aletter(\avar)$}
   \STATE Calculate the $\amap(\avar)$th letter $\aletterbis$ of $\aword$ and  return $\aletter = \aletterbis$.
\ELSIF{$\aformulabis$ is of the form $\neg \aformulabis'$}
      \STATE return not FOSAT$(\acps,\vect{n}, \aformulabis',\amap)$ 
\ELSIF{$\aformulabis = \aformulabis_1 \wedge \aformulabis_2$}
   \STATE return FOSAT$(\acps,\vect{n},\aformulabis_1,\amap)$ and FOSAT$(\acps,\vect{n},\aformulabis_2,\amap)$. 
\ELSIF{$\aformulabis$ is of the form $\exists \avar < m \  \aformulabis'$}
   \STATE guess a position $k\in\interval{0}{m-1}$.
   \STATE return FOSAT$(\acps,\vect{n},\aformulabis',\amap[\avar \mapsto k])$.
\ELSIF{$\aformula$ is of the form $R(\avar,\avar')$ for some $R \in \set{=,<,S}$}
   \STATE return $R(\amap(\avar),\amap(\avar'))$.
\ENDIF
\end{algorithmic}
}
\end{algorithm}

Polynomial space algorithm for membership problem is obtained by computing
FOSAT$(\acps,\vect{n}, \hat{\aformulabis}, \amap_0)$ with the algorithm FOSAT defined below ($\amap_0$ is a zero assignment
function). Note that the polynomial space bound is obtained since the recursion depth is linear in
$\size{\aformulabis}$ and positions in $\interval{0}{\length{u} + \length{v} \times 2^{\qheight{\aformulabis}}}$
can be encoded in polynomial space in $\size{\acps} + \size{\aformulabis}$. Furthermore, since model-checking ultimately periodic words with first-order logic is \pspace-hard~\cite{Markey&Schnoebelen03}, we deduce directly the lower bound for the membership problem with FO.
\ifconference
\qed \end{proof}
\fi

\section{Proof of Theorem~\ref{thm:stuttering-buchi}} 

First, we establish the result below. 
\begin{lemma}
\label{lemma-buchi-aux}
Let $\aword\in\alang(\bauto)$ for a \buchi~automaton $\bauto=\tuple{\states,\aalphabet,\astate_0,\edges,F}$, such that $\aword=\aword_1.u^{2.|\states|^k}.\aword_2$ for some $k$, then there exist an integer $K\in[1,|\states|]$
such that for all $N\in[1,|\states|^{k-2}]$, $\aword_1.u^{2.|\states|^k-(K\times N)}.\aword_2\linebreak[0]\in\alang(\bauto)$.
\end{lemma}


\section{Proof of Lemma~\ref{corollary-intersection-ba}}
\begin{proof}
Consider any given specification $\aspecification$ in BA, a constrained path schema $\acps$.
We would first construct the \buchi~automaton
$\bauto_{\aspecification}$ corresponding to $\aspecification$ as explained in Section~\ref{spec-def}. Recall, that $\aspecification$ in BA has
transitions labelled with Boolean combination over $at\cup ag_n$ whereas the equivalent $\bauto_{\aspecification}$ has transitions labelled with
elements of $\aalphabet=\powerset{at\cup ag_n}$. Hence, in effect, $\bauto_{\aspecification}$ could have an exponential number of transitions. On the
other hand by definition $\acps$ is defined over an alphabet $\aalphabet'\subseteq\aalphabet$. By Lemma~\ref{lemma-nice-subalphabet-property}, we know
that BA has the nice subalphabet property. Hence, we can transform $\aspecification$ over $\aalphabet$ to $\aspecification'$ over $\aalphabet'$ in
polynomial time. The \buchi~automata $\bauto_{\aspecification'}$ obtained from $\aspecification'$ following the construction in Section~\ref{spec-def}
has transitions labelled by letters from $\aalphabet'$. Clearly in this case the number of transitions in $\bauto_{\aspecification'}$ is polynomial in
$\size{\acps}$. 
We obtain the following equivalences,
\begin{itemize}
\item using Lemma~\ref{lemma-nice-subalphabet-property} and the fact that $\alang(\acps)\subseteq(\aalphabet')^{\omega}$,
  $\alang(\aspecification)\cap\alang(\acps)$ is non-empty iff $\aword\in\alang(\aspecification')\cap\alang(\acps)$ is non-empty.
\item Since, $\bauto_{\aspecification'}$ is obtained from $\aspecification'$ following the construction from Section~\ref{spec-def},
  $\alang(\bauto_{\aspecification'})\cap\alang(\acps)$ is non-empty iff $\aword\in\alang(\aspecification')\cap\alang(\acps)$ is non-empty.
\end{itemize}
Checking $\alang(\bauto_{\aspecification'})\cap\alang(\acps)$ amounts to guessing $\vect{n}\in\amap_{\rm BA}(\bauto_{\aspecification'},\acps)$ and checking for $\aword=\aseg(\aloop_1)^{\vect{n}[1]}\cdots\aseg_{k-1}(\aloop_{k-1})^{\vect{n}[k-1]}\aseg_k\aloop_k^{\omega}$, $\aword\in\alang(\bauto_{\aspecification'})\cap\alang(\acps)$.
We know that checking $\aword\in\alang(\bauto_{\aspecification'})\cap\alang(\acps)$ is in \ptime~and the construction of $\bauto_{\aspecification'}$
from $\aspecification$ takes only polynomial time. Thus, checking $\alang(\aspecification)\cap\alang(\acps)$ is non-empty can be done in polynomial time.
\qed
\end{proof}

\section{proof of Lemma~\ref{lemma-membership-gen}}
\begin{proof}
First we prove that the membership problem for $\aspeclanguage$ having the nice BA property is in \pspace. Let $\aspecification \in \aspeclanguage$ over the constrained alphabet $\triple{at}{ag_n}{\aalphabet}$ and let $\aword=\aseg_1 (\aloop_1)^{\vect{n}[1]} \cdots \aseg_{k-1} (\aloop_{k-1})^{\vect{n}[k-1]} \aseg_k (\aloop_k)^{\omega}$ be a word over $\aalphabet$. We would like to check whether $\aword \in \alang(\aspecification)$ which is equivalent, thanks to the nice BA property, to check whether $\aword \in \alang(\bauto_{\aspecification})$.

To verify if   $\aword \in \alang(\bauto_{\aspecification})$, we try to find a ``lasso'' structure in \buchi~automaton. Assume $\bauto_{\aspecification}=\tuple{\states,\aalphabet,\edges,\astate_i,F}$. We proceed as follows. First, we guess a state $\astate_f\in F$ and a position $j\in[1,\length{\aloop_k}]$. Then we consider the two finite state automata $\fsa_1=\tuple{\states,\aalphabet,\astate_0,\edges,\{\astate_f\}}$ and $\fsa_2=\tuple{\states,\aalphabet,\astate_f,\edges,\{\astate_f\}}$. And our method returns true iff both the following conditions are true:
\begin{enumerate}
  \item $\aseg_1(\aloop_1)^{\vect{y}[1]}\ldots\aseg_{k-1}(\aloop_{k-1})^{\vect{y}[k-1]}\aseg_k\aloop_k[1,j]\in\alang(\fsa_1)$.
  \item $\alang(\fsa_2)\cap\alang(\aloop_k[j+1,\length{\aloop_k}]\aloop_k^{*}\aloop_k[1,j])\neq\emptyset$.
\end{enumerate}

We will now show the correctness of the above procedure. First let us assume that
$\aword=\aseg_1(\aloop_1)^{\vect{y}[1]}\ldots\aseg_{k-1}(\aloop_{k-1})^{\vect{y}[k-1]}\aseg_k\aloop_k^{\omega}\in\alang(\bauto_{\aspecification})$. Thus,
there is an accepting run $\arun\in\edges^{\omega}$ for $\aword$ in $\bauto_{\aspecification}$. According to the \buchi~acceptance condition there exists
a state $\astate_f\in F$ which is visited infinitely often. In $\aword$ only $\aloop_k$ is taken infinitely many times. Thus, $\aloop_k$ being of
finite size, there exists a position $j\in[1,\length{\aloop_k}]$ such that transitions of the form $\astate\labtrans{\aloop_k(j)}\astate_f$ for some
$\astate\in\states$ occurs infinitely many times in $\arun$. Thus, for $\arun$ to be an accepting run, there exists
$\aword'=\alang(\aseg_1(\aloop_1)^{\vect{y}[1]}\ldots\aseg_{k-1}(\aloop_{k-1})^{\vect{y}[k-1]}\aseg_k\aloop_k[1,j])$, which has a run in $\bauto_{\aspecification}$ from
$\astate_i$ to $\astate_f$ and there must exists words $\aword''\in\alang(\aloop_k[j+1,\length{\aloop_k}]\aloop_k^{*}\aloop_k[1,j])$ which has a run
from $\astate_f$ to $\astate_f$. Hence we deduce that $\aword'\in\alang(\fsa_1)$ and $\alang(\fsa_2)\cap\alang(\aloop_k[j+1,\length{\aloop_k}]\aloop_k^{*}\aloop_k[1,j])\neq\emptyset$. Thus, there
exists at least one choice of $\astate_f$ and $j$, for which both the checks return true and hence the procedure returns true.

Now let us assume that the procedure returns true. Thus, there exists $\astate_f\in F$ and $j\in[1,\length{\aloop_k}]$ such that
$\aword_1= \aseg_1(\aloop_1)^{\vect{y}[1]}\ldots\aseg_{k-1}(\aloop_{k-1})^{\vect{y}[k-1]}\aseg_k\aloop_k[1,j]$ is in $\alang(\fsa_1)$ and $\alang(\fsa_{2})\cap\alang(\aloop_k[j+1,\length{\aloop_k}]\aloop_k^{*}\aloop_k[1,j])\neq\emptyset$.
From the second point, we deduce that there exists a word
$\aword_2=\aloop_k[j+1,\length{\aloop_k}]\aloop_k^{n}\aloop_k[1,j]\in\alang(\fsa_{2})$ for some $n$. Consider the word
$\aword=\aword_1.(\aword_2)^{\omega}$. First we have directly that $\aword \in\alang(\acps)$. And by construction of $\fsa_1$ and $\fsa_2$ we know that $\aword_1$ has a run in $\bauto_{\aspecification}$ starting from $\astate_i$ to
$\astate_f$ and $\aword_2$ has a run in $\bauto_{\aspecification}$ starting from $\astate_f$ to
$\astate_f$. Since $\astate_f$ is an accepting state of $\bauto_{\aspecification}$, we deduce that  $\aword \in \alang(\bauto_{\aspecification})$.

The proof that the above procedure belongs to \pspace~is standard and used the nice BA property which allows us to perform the procedure "on-the-fly". First note that for $\aspecification$ in $\aspeclanguage$ having the nice BA property, the corresponding B\"uchi automaton  $\bauto_{\aspecification}$ can be of exponential size in the size of the $\aspecification$, so we cannot  
construct the transition relation of
$\bauto_{\aspecification}$ explicitly, instead we do it on-the-fly.
We consider the different steps of the procedure and show that they can be done in polynomial space.
\begin{enumerate}
\itemsep 0 cm 
  \item $\fsa_1$ and $\fsa_2$ are essentially copies of $\bauto_{\aspecification}$ and hence their transition 
  relations are also not constructed
    explicitly. But, 
    by the nice BA property, 
    their states can be represented in polynomial space.
  \item Checking $\aseg_1(\aloop_1)^{\vect{y}[1]}\ldots\aseg_{k-1}(\aloop_{k-1})^{\vect{y}[k-1]}\aseg_k\aloop_k[1,j]\in\alang(\fsa_1)$ can be done by simulating $\fsa_1$ on this word. Note that for simulating $\fsa_1$, at any position we
    only need to store the previous state and the letter 
    at current position to obtain the next state of $\fsa_1$. Thus, this can be performed in
    polynomial space in $\size{\aspecification}$ and $\size{\acps}+\size{\vect{y}}$.
  \item Checking $\alang(\fsa_2)\cap\alang(\aloop_k[j+1,\length{\aloop_k}]\aloop_k^{*}\aloop_k[1,j])\neq\emptyset$ can be done by constructing a finite state automaton $\fsa_{loop}$ for
    $\alang(\aloop_k[j+1,\length{\aloop_k}]\aloop_k^{*}\aloop_k[1,j])$ and by checking for reachability of final state in the automaton 
     $\fsa_2\times\fsa_{loop}$. Note that $\size{\fsa_{loop}}$ is
    polynomial, but since $\size{\fsa_2}$ can be of exponential size, 
    $\size{\fsa_2\times\fsa_{loop}}$ can also be of exponential magnitude. 
    However, the graph accessibility problem (GAP) is in \nlogspace, so
    $\alang(\fsa_2)\cap\alang(\fsa_{loop})\neq\emptyset$ can also be done in nondeterministic 
    polynomial space.
\end{enumerate}
Thus, the whole procedure can be completed in nondeterministic polynomial space and by applying Savitch's theorem, we obtain that for any $\aspeclanguage$, satisfying the nice BA property, the membership problem for $\aspeclanguage$ having the nice BA property is in \pspace.\\

Now we will prove that the intersection non-emptiness problem for $\aspeclanguage$ having the nice BA property is in \pspace~ too. Let $\aspecification$ in $\aspeclanguage$ with the nice BA property and let  $\acps=\pair{\cpschema}{\aconstraintsystem}$ be a constrained path schema. Thanks to the nice BA property we have that $\alang(\acps)\cap\alang(\aspecification)\neq \emptyset$ iff $\alang(\acps)\cap\alang(\bauto_{\aspecification})\neq \emptyset$. Using Theorem \ref{thm:stuttering-buchi} we have
$\alang(\acps)\cap\alang(\bauto_{\aspecification})\neq \emptyset$ iff there exists 
$\vect{y}\in[0,\amap_{BA}(\bauto_{\aspecification},\acps)]^{k-1}$ such that
$\aseg_1(\aloop_1)^{\vect{y}[1]}\ldots\aseg_{k-1} \linebreak[0] 
(\aloop_{k-1})^{\vect{y}[k-1]}\aseg_k\aloop_k^{\omega}\in\alang(\bauto_{\aspecification})\cap\alang(\acps)$ where $\amap_{BA}(\bauto_{\aspecification},\acps)$ is equal to $\powerset{\mathtt{pol}_1(\size{\acps})}+2.\card{\states}^{\size{\acps}}\times \powerset{\mathtt{pol}_1(\size{\acps})+
  \mathit{pol}_2(\size{\acps})}$ ($Q$ being the set of states of $\bauto_{\aspecification}$ whose cardinality is, thanks to the nice BA property, at most exponential in the size of $\aspecification$). Hence our algorithm  amounts to guess some
$\vect{y} \in \interval{0}{\amap_{BA}(\bauto_{\aspecification},\acps)}^{k-1}$ and check
whether 
$\aword = \aseg_1 (\aloop_1)^{\vect{y}[1]} \cdots \aseg_{k-1} (\aloop_{k-1})^{\vect{y}[k-1]} \aseg_k (\aloop_k)^{\omega} \linebreak[0] \in 
\alang(\acps) \cap \alang(\aspecification)
$. Since the membership problem for $\aspecification$ can be done in \pspace~ and since the $\vect{y}[i]$ can be encoded in polynomial space in the size of $\aspecification$ and $\cps$, we deduce that the intersection non-emptiness problem for $\aspeclanguage$ with the nice BA property is in \pspace.
\qed
\end{proof}

\section{Proof of  Theorem~\ref{theorem-hardness}} 
\label{section-proof-theorem-hardness}

\section{Proof of Corollary~\ref{corollary-global-model-checking}}
\label{section-proof-corollary-global-model-checking}
\begin{proof} The proof takes advantage of a variant of Theorem~\ref{theorem-reduction} 
(whose proof is also based on developments from~\cite{DemriDharSangnier12}) in which
initial counter values are replaced by variables.
Below, we prove the results for BA, which immediately leads to a similar result for ABA, ETL and 
$\lmc$. 

Let  $\asys$  be a flat counter system  of dimension $n$ built over atomic constraints in $at \cup ag_{n}$, 
$\astate$ be  a control state and  $\aspecification$  be a specification in BA (i.e.
a B\"uchi automaton whose underlying constrained alphabet is $\triple{at}{ag_n}{\aalphabet}$).
A \defstyle{parameterized} constraint path schema (PCPS) is defined as a constrained path schema except
that the second argument (a guard) has also the free variables $\avariableter_1$, \ldots,$\avariableter_{n}$ 
dedicated to the initial counter values. Remember that a constrained path schema has already a constraint
about the number of times loops are visited. In its parameterized version, this constraint expresses also
a requirement on the initial counter values. 
Following the proof of Theorem~\ref{theorem-reduction}, one can 
construct in exponential time a set $\aset$ of parameterized constrained path schemas such that:
\begin{itemize}
\itemsep 0 cm 
\item Each parameterized constrained path schema $\apcps$ in $\aset$ has an alphabet of the form 
      $\triple{at}{ag_n}{\aalphabet'}$ ($\aalphabet'$ may vary) and 
     $\apcps$ is of polynomial size.
\item Checking  whether a parameterized constrained path schema belongs to $\aset$ can be done in polynomial time. 
\item For every run $\arun$ from $\pair{\astate}{\avect}$, 
      there is a parameterized constrained path schema $\apcps$ and $\aword \in \alang(\apcps[\avect])$ such that
      $\arun \models \aword$ where $\apcps[\avect]$ is the contrained path obtained from $\apcps$ by replacing
      the variables  $\avariableter_1$, \ldots,$\avariableter_{n}$ by the counter values from $\avect$. 
\item  For every parameterized constrained path schema $\apcps$, 
       for every counter values $\avect$,  
       for every $\aword \in \alang(\apcps[\avect])$,  
       there is a run $\arun$ from $\pair{\astate}{\avect}$ such that $\arun \models \aword$. 
\end{itemize}

The  existential Presburger formula $\aformula(\avariableter_1, \ldots,\avariableter_n)$ has the form below
$$
\bigvee_{\apcps = \pair{\cdot}{\aformulabis} \in \aset} \ 
\bigvee_{\astate_{init}, \astate, (\aloop_k)^{\omega} \in \alang(\aspecification'_{q})} \ 
(
\exists \ \avariablebis_1, \ldots,\avariablebis_M \ 
\exists  \ \avariable_1, \ldots,\avariable_{k-1} 
$$
$$
\aformulabis_{\astate_{init},\astate}(\avariablebis_1, \ldots, \avariablebis_M)
\wedge
(\avariablebis_1 = \alpha^1_0 + \alpha_1^1 \avariable_1 + \cdots + \alpha_{k-1}^1 \avariable_{k-1}) 
\wedge \cdots 
$$
$$
\ldots \wedge 
(\avariablebis_{M} = \alpha^M_0 + \alpha_1^M \avariable_1 + \cdots + \alpha_{k-1}^M \avariable_{k-1}) 
\wedge
\aformulabis(\avariable_1, \ldots, \avariable_{k-1},\avariableter_1, \ldots, \avariableter_n)  
)
$$
where
\begin{enumerate}
\itemsep 0 cm
\item $\triple{at}{ag_n}{\aalphabet'}$ is the alphabet of $\apcps$, $M = \card{\aalphabet'}$ and 
      by the \defstyle{nice subalphabet property}, there is a specification $\aspecification'$ such that
      $\alang(\aspecification') =  \alang(\aspecification) \cap (\aalphabet')^{\omega}$. 
\item $\astate_{init}$ is an initial state of $\aspecification'$ and $\astate$ is a state of $\aspecification'$.
\item $\aformulabis_{\astate_{init},\astate}(\avariablebis_1, \ldots, \avariablebis_M)$ 
      is the quantifier-free Presburger formula for the Parikh image of finite words over the alphabet
      $\aalphabet'$ accepted by $\aspecification'$ (viewed as a finite-state automaton) 
      with initial state $\astate_{init}$ and final state 
      $\astate$.  $\aformulabis_{\astate_{init},\astate}(\avariablebis_1, \ldots, \avariablebis_M)$ is of polynomial
      size in the size of B\"uchi automaton. 
\item $\apcps = \pair{\aseg_1 (\aloop_1)^* \cdots \aseg_{k-1} (\aloop_{k-1})^* \aseg_k (\aloop_k)^{\omega}}{
\aformulabis(\avariable_1, \ldots, \avariable_{k-1},\avariableter_1, \ldots, \avariableter_n)}$.
\item For each letter $\aletter_j$, we write $\alpha^M_0$, \ldots, $\alpha^M_{k-1}$ to denote
the natural numbers such that if each loop $i$ in $\apcps$ is taken $\avariable_i$ times, then
the letter $\aletter_j$ is visited  
$\alpha^j_0 + \alpha_1^j \avariable_1 + \cdots + \alpha_{k-1}^j \avariable_{k-1}$
times along  $\aseg_1 (\aloop_1)^* \cdots \aseg_{k-1} (\aloop_{k-1})^* \aseg_k$. 
Those coefficients can be easily computed from  $\aseg_1 (\aloop_1)^* \cdots \aseg_{k-1} (\aloop_{k-1})^* \aseg_k$
(for instance $\alpha_1^j$ is the number of times the letter $a_j$ is present in the first loop).
\item Finally, observe that checking whether
$(\aloop_{k})^{\omega} \in \alang(\aspecification'_{q})$ where 
$\aspecification'_{q}$ is defined as the specification $\aspecification'$ in which the unique initial state is $q$, 
amounts to perform a nonemptiness test between two B\"uchi automata. 
\end{enumerate}

FO admits a similar proof but it is based
on Theorem~\ref{theorem-stuttering-fo} (actually the proof is much simpler because the number of times
loops can be visited depends essentially on a threshold value). 
For FO,  it is sufficient to consider the formula below:
$$
\bigvee_{\apcps = \pair{\aseg_1 (\aloop_1)^* \cdots \aseg_{k-1} (\aloop_{k-1})^* \aseg_k (\aloop_k)^{\omega}}{\aformulabis} \in \aset} \ 
\bigvee_{\vec{y} \ s.t. \ \aseg_1 \aloop_1^{\vec{y}[1]} \aseg_2 \aloop_2^{\vec{y}[2]} \ldots \aloop_{k-1}^{\vec{y}[k-1]} \aseg_k \aloop_k^\omega \models \aspecification'} \ \
$$
$$
\exists \ \avariablebis_1 \cdots \ \avariablebis_{k-1} \
\aformulabis(\avariablebis_1, \ldots, \avariablebis_{k-1},\avariableter_1, \ldots,\avariableter_{n})
\wedge \aformulabis_1 \wedge \cdots \wedge \aformulabis_{k-1}
$$
where
\begin{enumerate}
\itemsep 0 cm
\item 
 $\apcps = \pair{\aseg_1 (\aloop_1)^* \cdots \aseg_{k-1} (\aloop_{k-1})^* \aseg_k (\aloop_k)^{\omega}}{
\aformulabis(\avariable_1, \ldots, \avariable_{k-1},\avariableter_1, \ldots, \avariableter_n)}$,
\item $\triple{at}{ag_n}{\aalphabet'}$ is the alphabet of $\apcps$ and 
      by the \defstyle{nice subalphabet property}, there is a specification $\aspecification'$ such that
      $\alang(\aspecification') =  \alang(\aspecification) \cap (\aalphabet')^{\omega}$. 
\item the third generalized disjunction deals with 
      $\vec{y} \in \interval{0}{2^{\size{\aspecification'}+1}+1}^{k-1}$.
\item For $i \in \interval{1}{k-1}$, 
      $\aformulabis_i \egdef (\avariablebis_1 = \alpha)$ if
      $\vec{y}[i] < 2^{\size{\aspecification'}+1}+1$ otherwise $\aformulabis_i \egdef 
      (\avariablebis_i \geq 2^{\size{\aspecification'}+1}+1$.
\end{enumerate}
\qed
\end{proof}

\fi 
\end{document}